\pdfoutput=1
\RequirePackage{ifpdf}
\ifpdf 
\documentclass[pdftex]{sigma}
\else
\documentclass{sigma}
\fi

\usepackage{lscape,longtable}

\def\be{\begin{equation}} \def\ee{\end{equation}}

\def\ri{\mathrm{i}}

\def\Zh{\widehat{Z}}

\def\BR{\mathbb{R}}
\def\BN{\mathbb{N}}
\def\BC{\mathbb{C}}

\def\BZ{\mathbb{Z}}
\def\BE{\mathbb{E}}
\def\BH{\mathbb{H}}

\def\CF{\mathcal{F}}

\def\CG{\mathcal{G}}

\def\CS{\mathcal{S}}
\def\wv{\vec{\rho}}

\def\ss{ \tilde{\mathcal{S}} }

\def\m{m_\ast}

\newcommand{\CC}{\mathbb{C}}
\newcommand{\RR}{\mathbb{R}}
\newcommand{\QQ}{\mathbb{Q}}

\newcommand{\HH}{\mathbb{H}}
\newcommand{\ZZ}{\mathbb{Z}}
\newcommand{\NN}{\mathbb{N}}

\newcommand{\Zhlong}[3]{\ensuremath{\widehat{Z}^{#1}_{#2}\left(#3;\tau\right)}}
\newcommand{\Zhshort}[3]{\ensuremath{\widehat{Z}^{#1}_{#2}\left(#3\right)}}
\newcommand{\Zchlong}[3]{\ensuremath{\widecheck{Z}^{#1}_{#2}\left(#3;\tau,\bar{\tau}\right)}}

\def\a{\alpha}

\def\bon{\mathbf{n}}
\def\bos{\varrho}

\def\al{\alpha}
\def\bv{{\vec{\underline{b}}}}
\def\bal{\boldsymbol{\alpha}}
\def\bmu{\boldsymbol{\mu}}
\def\bb{\boldsymbol{\beta}}
\def\bol{\boldsymbol{\ell}}
 \def\bon{\boldsymbol{n}}
\def\brho{\boldsymbol{1}}

\def\boxx{\boldsymbol{x}}
\def\bosi{\bos}
\def\ex{{\rm e}}
\def\sgn{\mathrm{sgn}}
\def\tmod{\mathrm{mod}}
\def\id{\boldsymbol{1}}
\def\t{\tau}
\def\im{\operatorname{Im}}

\DeclareFontFamily{U}{mathx}{\hyphenchar\font45}
\DeclareFontShape{U}{mathx}{m}{n}{
	<5> <6> <7> <8> <9> <10>
	<10.95> <12> <14.4> <17.28> <20.74> <24.88>
	mathx10
}{}
\DeclareSymbolFont{mathx}{U}{mathx}{m}{n}
\DeclareFontSubstitution{U}{mathx}{m}{n}
\DeclareMathAccent{\widecheck}{0}{mathx}{"71}

\makeatletter
\newcommand{\pder}[3][\@nil]{%
	\def\tmp{#1}
	\ifx\tmp\@nnil
	\ensuremath{\dfrac{\partial#2}{\partial#3}}
	\else
	\ensuremath{\dfrac{\partial^{#1}#2}{\partial#3^{#1}}}
	\fi}
\makeatother

\makeatletter
\newcommand{\der}[3][\@nil]{%
	\def\tmp{#1}
	\ifx\tmp\@nnil
	\ensuremath{\dfrac{d #2}{d #3}}
	\else
	\ensuremath{\dfrac{d ^{#1}#2}{d #3^{#1}}}
	\fi}
\makeatother

\newcommand\labelAndRemember[2]
 {\expandafter\gdef\csname labeled:#1\endcsname{#2}%
 \label{#1}#2}
\newcommand\recallLabel[1]
 {\csname labeled:#1\endcsname\tag{\ref{#1}}}

\newtheorem{thm}{Theorem}[section]
\newtheorem{prop}[thm]{Proposition}
\newtheorem{lem}[thm]{Lemma}
\newtheorem{conj}[thm]{Conjecture}
\newtheorem*{conj*}{General Conjecture}

\theoremstyle{definition}
\newtheorem{defn}[thm]{Definition}

\numberwithin{equation}{section}

\begin{document}

\allowdisplaybreaks

\newcommand{\arXivNumber}{2304.03934}

\renewcommand{\PaperNumber}{018}

\FirstPageHeading

\ShortArticleName{Quantum Modular $\widehat Z^G$-Invariants}

\ArticleName{Quantum Modular $\boldsymbol{\widehat Z^G}$-Invariants}

\Author{Miranda C.N. CHENG~$^{\rm abc}$, Ioana COMAN~$^{\rm bd}$,
Davide PASSARO~$^{\rm b}$ and Gabriele SGROI~$^{\rm b}$}

\AuthorNameForHeading{M.C.N.~Cheng, I.~Coman, D.~Passaro and G.~Sgroi}

\Address{$^{\rm a)}$~Korteweg-de Vries Institute for Mathematics, University of Amsterdam,\\
\hphantom{$^{\rm a)}$}~Amsterdam, The Netherlands}

\Address{$^{\rm b)}$~Institute of Physics, University of Amsterdam, Amsterdam, The Netherlands}
\EmailD{\href{mailto:c.n.cheng@uva.nl}{c.n.cheng@uva.nl}, \href{mailto:d.passaro@uva.nl}{d.passaro@uva.nl}, \href{mailto:gabrielesgroi94@gmail.com}{gabrielesgroi94@gmail.com}}

\Address{$^{\rm c)}$~Institute for Mathematics, Academica Sinica, Taipei, Taiwan}
\Address{$^{\rm d)}$~Kavli Institute for the Physics and Mathematics of the Universe, University of Tokyo,\\
\hphantom{$^{\rm d)}$}~Kashiwa, Japan}
\EmailD{\href{mailto:ioana.coman@ipmu.jp}{ioana.coman@ipmu.jp}}

\ArticleDates{Received May 25, 2023, in final form February 07, 2024; Published online March 09, 2024}

\Abstract{We study the quantum modular properties of $\widehat Z^G$-invariants of closed three-manifolds. Higher depth quantum modular forms are expected to play a central role for general three-manifolds and gauge groups $G$. In particular, we conjecture that for plumbed three-manifolds whose plumbing graphs have $n$ junction nodes with definite signature and for rank $r$ gauge group $G$, that \smash{$\widehat Z^G$} is related to a quantum modular form of depth $nr$. We prove this for $G={\rm SU}(3)$ and for an infinite class of three-manifolds (weakly negative Seifert with three exceptional fibers). We also investigate the relation between the quantum modularity of \smash{$\widehat Z^G$}-invariants of the same three-manifold with different gauge group~$G$. We~conjecture~a~recursive relation among the iterated Eichler integrals relevant for \smash{$\widehat Z^G$} with~$G={\rm SU}(2)$ and~${\rm SU}(3)$, for negative Seifert manifolds with three exceptional fibers. This is reminiscent of the recursive structure among mock modular forms playing the role of Vafa--Witten invariants for ${\rm SU}(N)$. We prove the conjecture when the three-manifold is moreover an integral homological sphere.}


\Classification{57K31; 57K16; 11F37; 11F27}

\vspace{-2mm}

\section{Introduction and summary}
\label{sec:intro_sum}
\subsection{Introduction}

Modular forms play an indispensable role in modern mathematics and theoretical physics, as various topics in number theory, geometry, topology, conformal field theory, string theory, holography and more require their presence. Extending the theory of modular forms, the past two decades have seen the development of the theory of mock modular forms, starting with the seminal papers \cite{ BringmannOno,2002math12286H, ZagierBourbaki,zwegers2008mock}. In the meanwhile, mock modular forms have found applications in various areas in mathematics, including combinatorics, representation theory, topology, and mathematical physics. See \cite{Bringmann2017HarmonicMF} for a selection of these applications. In physics, mock modularity often captures the subtle adjustment of modularity properties when the appropriately defined space of quantum states of the physical system is non-compact, for instance due to the presence of a continuous part of the spectrum. See \cite{Dabholkar:2021lzt} and the references therein.

A bit more than a decade ago, the notion of modified modularity was further expanded into that of {\it quantum modular forms} \cite{zagier2010quantum}, which were initially introduced as a broad philosophy to include many very different types of functions. This expanded notion encompasses that of mock modularity (in the sense of \cite{10.2307/procamermathsoci.144.6.2337}). Another familiar family of functions that falls under this category as specific examples is the so-called false theta functions \cite{zagier2010quantum}, which are $q$-series of the form
\[\sum_{\vec n\in \vec\mu + \Lambda} \prod_{i}\operatorname{sgn}(\langle \vec n, \vec v_i\rangle) q^{|\vec n|^2/2}\]
for some choices of $\vec v_i$ and $\vec\mu$ for a given lattice $\Lambda$. Roughly speaking, quantum modular forms are functions whose modular anomalies, measured in terms of differences between the functions and their images under the appropriate modular group action evaluated near a rational number, behave better than the original functions themselves (cf.\ Definition \ref{dfn:qmf}). In particular, they are in general not modular since these differences (the ``cocycles'') do not vanish.\footnote{If the cocycles were to vanish, then the corresponding functions would be rather boring, only taking finitely many values due precisely to the behaviour under the action of the congruence subgroup.}
Nevertheless, we can think of them as having some sort of modular behaviour since these differences unexpectedly have nicer analytic properties than they do. In particular, the quantum modular forms we discuss in the context of this paper are the so-called {\it holomorphic quantum modular forms}~\cite{ZagierHausdorf}, which are holomorphic functions on the upper-half plane whose cocycles can be defined everywhere on ${\mathbb C}$ with a half-line removed. Some of the initial examples of quantum modular forms arose from quantum invariants of knots and from Chern--Simons (or Witten--Reshetikhin--Turaev, WRT) invariants of three-manifolds \cite{lawrence1999modular}, and we might wonder where else in physics or topology does this type of functions play a role.

Before going into the specific mathematical context of the present paper, from the role modular and mock modular forms have so far played in two-dimensional conformal field theories (CFTs) and in the holographic principle of quantum gravity, and from the defining properties of quantum modular forms, we expect quantum modular forms to have the following interpretations in theoretical physics.
 \begin{itemize}\itemsep=0pt
\item{A more subtle UV/IR connection in 2d CFT and vertex operator algebra. }
The symmetry $Z(\tau)=Z(-1/\tau)$ under $S$-transformation of the partition function of a 2d CFT can be viewed as a manifestation of the UV/IR connection of the theory. In this context, quantum modularity of a $q$-series \big(as usual we write $q:={\rm e}^{2\pi {\rm i} \tau}$\big) indicates a more subtle UV/IR connection: while the leading behaviour of the higher energy states is still determined by the low-temperature limit, there are additional sub-leading contributions arising due to the non-vanishing cocycle.
\item{Summing over geometries.} When a $q$-series gives rise to a quantum modular form, it is often the case that its Fourier coefficients can be expressed using the circle method, whose applicability only requires a reasonable or understandable transformation of the function of interest, generalizing the circle method to express the coefficients of modular forms.
The $q$-series in this case admits an expression in terms of an appropriately-regularised sum of images of a certain function under the action of the modular group. See, for instance, \cite{Bringmann:2021wuk, Bringmann:2019vyd, Cheng:2012qc} for discussions on false theta functions and numerous work, starting with \cite{10.2307/1997003},
on the Rademacher sums of mock modular forms.
This structure has played a role in the gravitational black hole interpretation of the BPS indices in the context of ${\rm AdS}_3/{\rm CFT}_2$ \cite{deBoer:2006vg, Denef:2007vg,Dijkgraaf:2000fq,Maldacena:1998bw} and ${\rm AdS}_2/{\rm CFT}_1$ \cite{Dabholkar:2014ema, Ferrari:2017msn,Murthy:2009dq}, where the summation over the modular images is interpreted as a sum over different saddle point geometries.
We expect that a similar gravitational interpretation also holds for the Rademacher sum expression of certain quantum modular $q$-series playing the role of BPS indices, as we will discuss further in Section~\ref{sec:discussions} in the context of 3d indices.
\end{itemize}
Next we come to the specific context in which we will discuss quantum modularity in the present paper. The main goal of this paper is to further investigate the quantum modular properties of certain $q$-series topological invariants of three-manifolds. These invariants, schematically denoted as $\widehat Z$-invariants and sometimes also referred to as the $q$-series invariant or the {\it homological blocks}, have been recently proposed in \cite{Gukov:2017kmk,Gukov:2016gkn} in terms of 3d-3d correspondence in M-theory.
This novel type of invariants has since been the subject of various recent studies in the literature. See \cite{Chauhan:2022cni,chengUpC,cheng20193d,
Chung:2018rea, Chung:2019jgw, Chung:2020efy, Chung:2022ypb,
Ekholm:2020lqy,2018arXiv181105376E, Ferrari:2020avq, Gukov:2020frk, Matsusaka:2021vbw, Murakami:2023oam,Park:2019xey, Park:2021ufu} for some examples.
In short, $\widehat Z_b(M_3)$ is physically defined as the half-index (also called vortex partition function) of the three-dimensional ${\cal N}=2$ supersymmetric quantum field theory $T[M_3]$ obtained by compactifying a~six-dimensional ${\cal N}=(2,0)$ superconformal field theory on the closed three-manifold $M_3$. Here~$b$ labels the specific choices of boundary conditions that leave some of the supersymmetries unbroken. That said,
as the exact content of $T[M_3]$ is in general still not known, this physical definition does not always lead to a method to explicitly compute the $\widehat Z_b(M_3)$ invariants in practice. On the other hand, a relation is conjectured between $\widehat Z_b(M_3)$ and the Chern--Simons invariant {CS$(M_3)$}
\cite{Gukov:2017kmk}. Specializing to the case $b_1(M_3)=0$, this relation reads
\begin{equation}
\label{ZhatfromCS-SU2-intro}
 \mathrm{CS}(M_3;k) = \frac{1}{{\rm i}\sqrt{2k}}
 \sum_{a,b }
 {\rm e}^{2\pi {\rm i} k {\rm CS}(a) }
 S_{ab}^{(A)} \widehat{Z}_b ( M_3 ; \tau)\lvert_{\tau \to \frac{1}{k}},
\end{equation}
where the sum can be thought of as over the connected components of the moduli spaces of Abelian flat connections up to Weyl group actions, or the inequivalent ${\rm Spin}^c$ structure and \smash{$S_{ab}^{(A)}$} is a concrete matrix whose form can be found in \cite{cheng20193d,gukov2019two,Gukov:2017kmk}. This relation suggests that \smash{$\widehat Z$}-invariants can be viewed as a function that extends, and categorifies via its BPS states counting interpretation, the WRT invariants. Using the above, the known expression for~$\mathrm{CS}(M_3;k)$, and inspired by the localization expressions for the half-indices of certain known theories, a~mathematical definition for \smash{$\Zhshort{G}{\bv}{M_3}$} invariants has been proposed for classes of three-manifolds~$M_3$~\cite{Gukov:2017kmk}, as well as knot complements in \cite{gukov2019two}. As the six-dimensional ${\cal N}=(2,0)$ superconformal field theories are labelled by an ADE gauge group $G$, we expect \smash{$\widehat Z^G_\bv(M_3)$} to be similarly defined for all ADE gauge groups $G$. Indeed, the mathematical definition for an arbitrary simply-laced gauge group $G$ is given in \cite{Park:2019xey}, generalizing the definition of \cite{Gukov:2017kmk} which corresponds to $G={\rm SU}(2)$. Concrete examples for $G={\rm SU}(N)$ and an infinite family of Seifert manifolds have been investigated in \cite{Chung:2018rea}.

For $G={\rm SU}(2)$, a relation between \smash{$\Zhshort{{\rm SU}(2)}{\bv}{M_3}$} and quantum modular forms, in particular false \cite{Bringmann_2020,bringmann2019quantum,cheng20193d} and mock theta functions \cite{cheng20193d, Cheng:2019uzc, Park:2021ufu} 
have been proposed.\footnote{See also Cheng M.C.N., Coman I., Kucharski P., Passaro D., Sgroi G., in preparation.} Generally, we have the following conjecture \cite{cheng20193d}, which can be traced all the way back to the relation between false and mock theta functions and WRT invariants of three-manifolds \cite{2005RCD10509H, HIKAMI_2005, 2004mathph9016H,lawrence1999modular}.
\begin{conj*}
\smash{$\widehat Z^G_\bv(M_3)$} is closely related to a quantum modular form of some kind for any closed three-manifold $M_3$, any ADE gauge group $G$, and any boundary condition label $\bv$.
\end{conj*}

Consider for instance $G={\rm SU}(2)$ and $M_3$ a Seifert manifold, while \smash{$\widehat Z^G_\bv(M_3)$} is a linear combination of quantum modular forms when $M_3$ is a Seifert manifold with three or four exceptional fibers, it is given more generally by linear combinations of derivatives of quantum modular forms when $M_3$ has at least five exceptional fibers \cite{bringmann2019quantum}. To expand our understanding of the above conjectural phenomenon, in this paper we study quantum modularity of \smash{$\widehat Z^G_\bv(M_3)$} for gauge groups~$G$ with rank larger than one, and $G={\rm SU}(3)$ in particular. In short, in these cases the type of quantum modular forms will be {\it higher depth} quantum modular forms. Analogous to higher depth mock modular forms\footnote{See also Zagier D., Zwegers S., {u}npublished.} (see also \cite{Alexandrov:2020dyy,Alexandrov:2020bwg,Alexandrov:2016enp,Alexandrov:2018lgp, bringmann2019vectorvalued, bringmann2016indefinite,Bringmann:2018cov, KumarGupta:2018rac, males2021higher,Manschot:2017xcr}), higher depth quantum modular forms can be defined recursively: the cocycles of a depth two quantum modular form are sums of depth one or zero quantum modular forms multiplied by analytic functions, and so on (cf.\ Definition \ref{dfn:higher_depth_qmf}). Before we go to the concrete results, let us mention a~manifestation of quantum modularity in this context.

As discussed in \cite{cheng20193d} in the context of $\widehat Z$ invariants for $G={\rm SU}(2)$, the transseries expression of the WRT invariant at the semi-classical regime can be understood as a consequence of the following two facts: (1)~the relation between WRT and $\widehat Z$ invariant \eqref{ZhatfromCS-SU2-intro}, and~(2)~the quantum modularity of the $\widehat Z$ invariant. Schematically, when the rank-one $\widehat Z$ is a component of a vector-valued quantum modular form (see Definition \ref{dfn:qmf}) $z=(z_{b'})$ with weight $w$ and $S$-matrix $S^{(q)}$, the above leads to
\begin{align}
 \mathrm{CS}(M_3; k)& ={\frac{1}{{\rm i}\sqrt{2k}}} \sum_{a }{\rm e}^{2\pi {\rm i} k {\rm CS}(a)}\sum_b S^{(A)}_{ab} \lim_{\tau\to {\frac{1}{ k}}} \Zhlong{{\rm SU}(2)}{\bv}{M_3} \nonumber \\
 &={\frac{1}{{\rm i}\sqrt{2k}}}
 \sum_{a } {\rm e}^{2\pi {\rm i} k {\rm CS}(a)}\sum_b S^{(A)}_{ab} \left( k^w \sum_{b'} S^{(q)}_{bb'} \lim_{\tau\to {-k}}z_{b'}(\tau) + r_b\left(\frac{1}{ k}\right)\right).\label{CS_transseries}
\end{align}
In the second line of the above equation, the first term inside the bracket arises from the $S$-transformation of $\widehat Z$, while $r_b\bigl(\tfrac{1}{ k}\bigr)$ is an asymptotic perturbative series in $\tfrac{1}{ k}$ capturing the non-vanishing cocycle. The above equality turns out to capture many intricate structures related to flat ${\rm SL}_2({\mathbb C})$ connections on $M_3$. Note first that the terms involving $z_{b'}(-k)$ are responsible for the contributions from the saddle points corresponding to {non-Abelian} flat connections. Moreover, since the summation over $a$ can be interpreted as a summation over the Abelian flat connections, it is clear from above that the transseries for semi-classical WRT invariants of this class of three-manifolds has the feature that only the saddle contributions from Abelian flat connections carry a factor given by a perturbative series, having the form \smash{${\rm e}^{-k{\rm CS}(a)} \frac{1}{ \sqrt{k}}R_a\bigl(\tfrac{1}{k}\bigr)$} where $R_a$ is again a~perturbative series. When the $\widehat Z$ invariant is a depth $N$ quantum modular form, one sees that the above structure gets generalized. Now there are up to $N$ ``classes'' of saddle points, with different complexity of the accompanying perturbative series. For instance, as before there will be no asymptotic series of $\tfrac{1}{k}$ multiplying the terms arising from the $S$-transformation of $\widehat Z_b$, and more generally there are saddle point contributions multiplied by products of $\ell$ perturbative series, for $0\leq \ell \leq N$. Again, we expect that the quantum modularity structure controls the intricate topological structure of the flat connections on the 3-manifolds, and we will return to this point in Section~\ref{sec:discussions}. Finally, in light of the proposed relation between vertex algebras and $\widehat Z$-invariants \cite{chengUpC,cheng20193d}, we expect it to be also fruitful to understand quantum modularity of $\widehat Z$-invariants in the context of vertex algebras.

\subsection{Summary of results}\label{subsec:sum}
In the remainder of this section we briefly discuss the main results of the paper.

\subsubsection*{Quantum modularity}
First, we make the following conjecture about the quantum modular properties of \smash{$\Zh^G_{\vec{\underline{b}}}(M_3)$} invariants for weakly-negative or weakly-positive plumbed three-manifolds $M_3$, where we let $G$ to be an arbitrary ADE gauge group and $\vec{\underline{b}}$ to be any allowed boundary condition as detailed in Section~\ref{sec:Zhat} (cf.\ \eqref{eq:bfromcoker}). As will be explained in Section~\ref{sec:Zhat}, a plumbed three-manifold can be defined in terms of its plumbing matrix $M$ encoding its plumbing graph. Following \cite{cheng20193d}, we~say a plumbed manifold is weakly-negative/positive when $M^{-1}$, when restricted to the subspace generated by the ``junction vertices'' (those with degree at least~3), is negative/positive-definite. A Seifert manifold can always be realized as a plumbed manifold with one junction vertex, and in this case we will simply refer to a Seifert manifold as a {\it negative/positive Seifert} manifold depending on the signature of the inverse plumbing matrix in the direction of the unique junction vertex.\looseness=-1

We first make the following conjecture, specializing and refining the {\it General Conjecture} outlined earlier.
\begin{conj}\label{conj_qmf}
Let $G$ be an ADE group with rank $r$. For $M_3$ a weakly-negative or weakly-positive plumbed manifold with $n$ junction vertices, the invariant \smash{$\Zhlong{G}{\bv}{M_3}$} is related to quantum modular forms of depth up to $r\times n$.
\end{conj}

We also prove the following special case, where $r=2$ and $n=1$.

\begin{thm}\label{thm1_intro}
 Let $G={\rm SU}(3)$. For a negative Seifert manifold $M_3$ with three exceptional fibers and for all the allowed ${\vec{\underline{b}}}$, the invariant \smash{$\Zhlong{G}{\bv}{M_3}$} is a sum of depth-one and depth-two quantum modular forms.
\end{thm}

In this paper, we prove the above statement by studying the so-called companion function of~\smash{$\Zh^G_{\vec{\underline{b}}}$}, denoted by \smash{$\widecheck Z^G_{\vec{\underline{b}}}$} which is defined as a function that has the same asymptotic expansion near $\tau\to \QQ$ up to a naive $\tau\mapsto -\tau$ transformation. See Definition \ref{def:companions}. We construct this in terms of iterated non-holomorphic Eichler integrals (cf.\ \eqref{dfn:Eich_period} and \eqref{dfn:double_eichler}), using a method similar to that of \cite{bringmann2018higher}, and we will refer to the companion function constructed in this specific way simply as the companion function. One can also translate our analysis into the language of two-variable completion \cite{Bringmann:2021dxg} instead of the companion in a relatively straightforward fashion. Specifically, in the language of \cite{Bringmann:2021dxg}, \smash{$\Zh^G_{\vec{\underline{b}}}$} is a sum of depth-one and depth-two false modular forms.

In Section~\ref{sec:discussions}, we briefly discuss the possible forms of quantum modularity for the cases beyond Conjecture \ref{conj_qmf}.

\subsubsection*{A recursive structure}

Next, consider changing the gauge group $G$ while fixing the three-manifold $M_3$ in \smash{$\widehat Z^G(M_3)$},
we ask the following question:
\begin{center}
\begin{tabular}{c}
\makebox[0.7\width]
{\it Question: Given a three-manifold $M_3$, are the quantum modular properties}\\ {\it of \smash{$\widehat Z^{{\rm SU}(N)}_{\vec{\underline{b}}}(M_3)$} for different $N$ related?}
\end{tabular}
\end{center}
To motivate this question, we find it illuminating to recall the following. Higher-depth quantum modular forms have been playing a prominent role in the study of the Vafa--Witten partition functions $Z^{G}_{\rm VW}(\tau; {M}_4)$ for twisted four-dimensional ${\cal N}=4$ super Yang--Mills on four-dimensional manifolds ${M}_4$. In more details, when $b_2^+(M_4)=1$, the invariant $Z^{G}_{\rm VW}(\tau; {M}_4)$ displays mock modular properties and the ``depth'' of the corresponding (mixed) modular forms is given in terms of the rank of the gauge group $G$. The mock modular properties in particular imply that there is a modular {\it completion} of $Z^{G}_{\rm VW}(\tau; {M}_4)$, denoted $\widetilde Z^{G}_{\rm VW}(\tau; {M}_4)$, which is non-holomorphic with a canonically defined holomorphic part equaling $ Z^{G}_{\rm VW}(\tau; {M}_4)$. While \smash{$\widetilde Z^{G}_{\rm VW}(\tau; {M}_4)$} transforms as a modular object, it has a non-trivial $\bar \tau$-dependence referred to as its holomorphic anomaly. In other words, the $\bar\tau$-dependence of the completion function $\widetilde Z^{G}_{\rm VW}( {M}_4)$ captures the mock modularity of the Vafa--Witten invariant $ Z^{G}_{\rm VW}( {M}_4)$. Notably, the holomorphic anomaly of \smash{$\widetilde Z^{G}_{\rm VW}(\tau; {M}_4)$} for $G=U(N)$ is given by \smash{$\widetilde Z^{G}_{\rm VW}(\tau; {M}_4)$} for $G=U(n)$ with $0< n< N$. Schematically, the conjecture states \cite{Alim:2010cf,Minahan:1998vr}
\begin{gather}\label{eqn:VW_recur}
\partial_{\bar{\tau}} \widetilde Z^{U(N)}_{\rm VW} \sim \sum_{n_1+n_2=N} n_1 n_2 \widetilde Z^{U(n_1)}_{\rm VW}\widetilde Z^{U(n_2)}_{\rm VW}.
\end{gather}
The above recursive relation, and more generally the mock modularity in this context, has been given a physical explanation from various perspectives including four-dimensional gauge theories, two-dimensional sigma models \cite{Dabholkar:2020fde}, curve counting \cite{Minahan:1998vr}, and DT invariants \cite{Alexandrov:2016tnf, Alexandrov:2017qhn,Alexandrov:2019rth, Alexandrov:2018lgp}. Roughly speaking, the presence of a holomorphic anomaly is related to the presence of reducible connections from the gauge theory point of view, and to the possibility of separating multiple~M5 branes from the M-theory point of view.
The recursive structure (\ref{eqn:VW_recur}) then naturally follows from these interpretations.
The similar M5 brane origin of the three-manifold invariants \smash{$\widehat Z^G(\tau; M_3)$}, as detailed in \cite{Gukov:2016gkn}, in particular motivates the question on the recursive structure of the quantum modularity of $\widehat Z$ invariants that we mentioned earlier.

To explore this question, we now focus on negative Seifert manifolds with three exceptional fibers, corresponding to plumbing graphs with one junction vertex of degree three. Based on the relation between $\widehat Z$ invariants and VOA characters shown in \cite{chengUpC,cheng20193d}, we expect \smash{$\widehat Z^G_\bv (M_3)$} for $G={\rm SU}(r+1)$ to be a linear combination of rank-$r'$ false theta functions, with $r'\leq r$, up to an overall rational power of $q$ and possibly the addition of a finite polynomial in $q$ and $q^{-1}$. Since many statements in the remaining part of the section are true up to an overall rational power of $q$ and the addition of a finite polynomial in $q$ and $q^{-1}$, for the sake of simplicity we will introduce the special notation $\stackrel{\dots}{=}$, $\stackrel{\dots}{\in}$, etc., where the $\dots$ is added on top of the symbols to denote that the relation holds when replacing $\widehat Z$ with $Cq^{\Delta}\widehat Z+f(q)$, for some $C\in \CC$, $\Delta\in \QQ$, $f(q)\in \CC\big[q,q^{-1}\big]$, and similarly for $\widecheck Z$. As usual, we write $q={\rm e}^{2\pi {\rm i} \tau}$ throughout this paper.

More specifically, we expect \smash{$\widehat Z^G_b(M_3)$} to be a linear combination of functions of the following form:
\begin{gather*}
t_{}^{(0),A_r}= \sum_{\vec n\in \vec\mu + \Lambda_r} \Biggr(\prod_{i}\operatorname{sgn}(\langle \vec n, \vec v_i\rangle) q^{|\vec n|^2/2}\Biggl), \\
 t_{}^{(1),A_r}= \sum_{\vec n\in \vec\mu + \Lambda_r}\langle \vec n, \vec \sigma\rangle \left(\prod_{i}\operatorname{sgn}(\langle \vec n, \vec v_i\rangle) q^{|\vec n|^2/2}\right)
\end{gather*}
for some chosen $\vec \mu$, $\vec v_i$ and $\vec \sigma$ and rank $r$ lattice $\Lambda_r$. We denote their companion functions, which we expect to be given by linear combinations of iterated Eichler integrals, by $\check t^{(\nu),A_r}_{}$ \cite{Bringmann:2021dxg}. Then the general structure of higher rank false theta functions suggests the following. Schematically,
\begin{gather}
 {\frac{\partial}{\partial \bar \tau}}
 {\widecheck Z}^{{\rm SU}(r+1)}\nonumber\\
 \qquad\stackrel{\dots}{\in}\operatorname{span}\bigl(\bigl\{ (\operatorname{Im}\tau)^{\nu_0-3/2}\,\overline{\theta^{\nu_0} }
 \check t^{(\nu_1),A_{r_1}} \check t^{(\nu_2),A_{r_2}} \cdots \mid \nu_i\in \{0,1\},\,r\geq 1+r_1+r_2+\cdots \bigr\}\bigr)\!\!\!\label{general_recursion}
\end{gather}
with $\theta^{\nu}$ denotes the function of the type $\theta^\nu_{m,r}$, defined as the following. For $m$ a positive integer, let $\Theta_m$ be the $2m$-dimensional Weil
representation of the metaplectic group \smash{$\widetilde{{\rm SL}_2(\BZ)}$} spanned by the column
vector $\theta_m=(\theta_{m,r})_{r \bmod 2m}$ with theta function components
\begin{equation}
\label{eq:thetafunc}
 \theta_{m,r}(\tau,z) := \sum_{\ell \equiv r\bmod 2m}
 q^{\frac{\ell^2}{4m}} y^\ell, \qquad y:={\rm e}^{2\pi\mathrm{i}z} .
\end{equation}
Derivatives of \eqref{eq:thetafunc} define unary theta functions $\theta^\nu_{m,r}\colon \HH\to \CC$ for $\nu=0,1,$ as
\begin{equation}\label{dfn:theta}
 \theta^\nu_{m,r}(\tau) := \left(\left({\frac{1}{ 2\pi {\rm i}}} \frac{\partial}{\partial z}\right)^{\nu} \!\theta_{m,r}(\tau,z) \right)\bigg\lvert_{z=0}.
\end{equation}
An important role in the study of $\widehat Z^{{\rm SU}(2)}$ is played by Eichler integrals unary theta functions. The Eichler integral of a weight $w\in \tfrac{1}{2}{\mathbb Z}$ cusp form $g(\tau)=\sum_{n>0}a_g(n) q^n$ is a function given by\footnote{To avoid an unnecessary proliferation of constants we adopt a different normalization of the Eichler integral than that chosen in \cite{cheng20193d}.}
\begin{equation}\label{dfn:Eichler_integral}
 \tilde g(\tau) := C(w)\sum_{n>0}a_g(n)n^{1-w} q^n,
\end{equation}
where $C(w)={\rm i}\frac{\Gamma(w-1)}{(-2\pi)^{w-1}}$, or
\[
 \tilde g(\tau) = \int_{\tau}^{{\rm i}\infty} g(z')(-{\rm i}(z'-\tau))^{-2+w} {\rm d}z',
\]
with a carefully chosen contour. In particular, the Eichler integral of $\theta^1_{m,r}$ is proportional the false theta function
\[
 \tilde\theta^1_{m,r}:= \sum_{k\equiv r~(2m)}
 \operatorname{sgn}(k)q^{k^2/4m}.
\]

In the present paper we focus on the next simplest non-trivial case where $G={\rm SU}(3)$. What we find is an interesting recursion structure, which we will describe in terms of the Weil representations of the metaplectic group \smash{$\widetilde{{\rm SL}_2(\BZ)}$} that are subrepresentations of $\Theta_m$. We denote by~$\mathrm{Ex}_m$ the group of exact divisors of $m$, where a divisor $n$ of $m$ is exact if $\left(n,\frac{m}{n}\right)=1$ and the group multiplication for $\mathrm{Ex}_m$ is
\[n\ast n':=\frac{nn'}{(n,n')^2}.\] In what follows we consider a subgroup $K$ of~$\mathrm{Ex}_m$. The group $K$ labels a subrepresentation of~$\Theta_m$, denoted $\Theta^{m+K}$ which we will describe in more details in Section~\ref{sec:Zhat}. It has the property that $\Theta^{m+K'}\subset \Theta^{m+K} $ when $K'\supset K$, and in particular $\Theta^{m+K}=\Theta_m$ when $K=\{1\}$. In general, we have, in terms of the projectors $P^{m+K}$ defined in \eqref{dfn:proj} and \eqref{dfn:projB},
\begin{equation}\label{linear_weil}
 \theta^{m+K}_{r}(\tau,z) = \sum_{r'\in \ZZ/2m} P_{r,r'}^{m+K} \theta_{m,r'}(\tau,z).
\end{equation}
We will write $\theta^{\nu,m+K}_{r}$, $\nu=0,1$, to be the corresponding linear combination of $\theta^{\nu}_{m,r}$ (see \eqref{WeylorbitSU2} for $\nu=1$), and write \smash{$\sigma^{m+K}\subset \ZZ/2m$} as the set labelling (through $r$) the linearly independent~\smash{$\theta^{m+K}_{r}(\tau,z)$}. Similar notations \smash{$\tilde \theta^{\nu,m+K}_r$} and \smash{$\bigl(\theta^{\nu,m+K}_r\bigr)^\ast$} are used for the same linear combination~\eqref{linear_weil} of the corresponding Eichler integrals and the non-holomorphic Eichler integral
 \[
 ({\theta^\nu_{m,r}})^\ast(\tau) :=
 \int_{-\bar{\tau} }^{{\rm i}\infty} {\rm d}w \,
\frac{\overline{\theta^\nu_{m,r}(-\bar w)}}{(-\ri(w+\tau))^{3/2-\nu}}
 \]
 (cf.\ (\ref{dfn:Eich_period})) which is up to an overall factor a companion for $ \tilde\theta^\nu_{m,r}$. See \cite[Section~7.3]{cheng20193d} for more details in the present context.

Now we explain the role of the representations $\Theta^{m+K}$ in the study of $\widehat Z^G$-invariants.
It is shown \cite{bringmann2019quantum,cheng20193d} that for any negative Seifert $M_3$ with three exceptional fibers, there exists a~unique $m$ and some $K\subset {\rm Ex}_m$ such that for all allowed choices of $\bv$
\begin{equation}\label{eqn:intro_Weyl1}
\Zhlong{{\rm SU}(2)}{\bv}{M_3} \in \operatorname{span}\bigl(
 \bigl\{\tilde\theta_{r}^{1,m+K}\mid r \in \sigma^{m+K} \bigr\}\bigr),
\end{equation}
 which implies
\begin{equation}\label{eqn:intro_Weyl2}
 \Zchlong{{\rm SU}(2)}{\bv}{M_3} \in \operatorname{span}
 \bigl(\bigl\{\bigl(\theta^{1,m+K}_r\bigr)^\ast\mid r \in \sigma^{m+K} \bigr\}\bigr).
\end{equation}
 From now on we will take the largest $K$ such that the above is true. The following conjecture, based on observations and proven for homological spheres, indicates that the recursion of $\widehat Z^G$ have in fact finer structure than indicated in \eqref{general_recursion}.
\begin{conj}\label{conj_rec}
Let $M_3$ be a negative Seifert manifold three exceptional fibers and let $\bv$ a choice of the boundary condition. Let $m$ be the unique positive integer and $K$ be the largest subgroup of ${\rm Ex}_m$ such that
\[
 \Zhlong{{\rm SU}(2)}{\bv}{M_3} \stackrel{\dots}{\in} \operatorname{span}\bigl(
 \bigl\{\tilde\theta_{r}^{1,m+K}\mid r \in \sigma^{m+K} \bigr\}\bigr) .
\]

Let \smash{$\Zchlong{{\rm SU}(3)}{\bv}{M_3}$} be the companion of \smash{$\widehat Z^{{\rm SU}(3)}_{\bv}(\tau;M_3)$}. Then it satisfies
\begin{gather}
 {\frac{\partial}{ \partial \bar \tau}}
 \bigl( \Zchlong{{\rm SU}(3)}{\bv}{M_3} + z_{1\rm d}\bigr)\nonumber \\
 \qquad\stackrel{\dots}{\in}
 {\frac{1}{\sqrt{\operatorname{Im} \tau}}} \operatorname{span}\bigl(\bigl\{ \overline{\theta^{1,m+K}_{r'}} \bigl({\theta^{\nu}_{m,r''}}\bigr)^\ast \mid \nu=0,1, \,r'' \in \ZZ/2m,\, r' \in \sigma^{m+K} \bigr\} \bigr),\label{eqn:conj:recursive}
\end{gather}
where the $1\rm d$ piece is of the form
\begin{equation}\label{intro_1dpiece}
 z_{1\rm d}\stackrel{\dots}{\in} \operatorname{span}\left(\{ ({\theta^{\nu}_{m,r}})^\ast \mid r \in \ZZ/2m,\, \nu=0,1 \} \right).
\end{equation}
\end{conj}

We see that the same Weil representation $\Theta^{m+K}$ that governs the structure of \smash{$\widehat Z^{{\rm SU}(2)}(M_3)$} also governs the structure of \smash{$\widehat Z^{{\rm SU}(3)}(M_3)$}. We will comment on its potential interpretation in Section~\ref{sec:discussions}.

When $M_3$ is moreover a homological sphere, namely when $H_1(M_3,\ZZ)$ is trivial, it is topologically equivalent to a Brieskorn sphere $\Sigma(p_1,p_2,p_3)$ with coprime $p_i$'s \eqref{dfn:brieskorn}. In this case there is only one homological block $\bv=\bv_0$ and it is known that \cite{cheng20193d}
 \[
 \frac{\partial}{\partial \bar \tau} \bigl(\Zchlong{{\rm SU}(2)}{\bv}{M_3}\bigr) \stackrel{\dots}{=} \overline{\theta_{r}^{1,m+K}}
 \]
for
$ m=p_1p_2p_3$, $K=\{1,\bar p_1,\bar p_3,\bar p_2\}$, $r=m-\bar p_1-\bar p_2-\bar p_3$,
where $\bar p_i:=m/p_i$.

For this (infinite) family of $M_3$, we explicitly show that the conjecture is true, and we have
\begin{thm}\label{thm:intro_brieskorn_recursive}
Conjecture {\rm \ref{conj_rec}} is true when $M_3$ is a homological sphere.
\end{thm}
In other words, in this case we have
\begin{gather}
 \frac{\partial}{\partial \bar \tau}\Zchlong{{\rm SU}(2)}{\bv}{M_3} \stackrel{\dots}{=} \frac{1}{\sqrt{\operatorname{Im} \tau} }\overline{\theta_{r}^{1,m+K}},\nonumber\\
 \frac{\partial}{\partial \bar \tau}
 \bigl( \Zchlong{{\rm SU}(3)}{\bv}{M_3} + z_{1\rm d} \bigr)
 \stackrel{\dots}{=} \frac{1}{\sqrt{\operatorname{Im} \tau} }\sum_{r' \in \sigma^{m+K}}
 \overline{ \theta^{1,m+K}_{r'}} B_{r' },\label{thm:recursive}
\end{gather}
where $B_{r' }$ is a linear combination of \smash{$({\theta^{\nu}_{m,r''}})^\ast$} with $\nu=0,1$, $r''\in \ZZ/2m$. In particular, note that while only one component of $\Theta^{m+K}$ appears to play a role in the quantum modularity of \smash{$\widehat Z^{{\rm SU}(2)}$}, its modular images also play a role in \smash{$\widehat Z^{{\rm SU}(3)}$}.

Here we have stated the recursive conjecture in terms of the companion function. As before, the above analysis on the recursive relation can be translated in the language of modular completions of higher-depth false theta functions \cite{Bringmann:2021dxg}. Roughly speaking, the role of $\frac{\partial}{\partial \bar\tau}$ will be played by $\frac{\partial}{\partial w}$, acting on the two-variable completion that depends on $(\tau,w)\in \HH\times \HH$ and transforms as a bi-modular form. An analogous statement should then hold also for (the $\partial_w$ derivatives of) the
 completion function defined in \cite{Bringmann:2021dxg} in a natural fashion.

\section{Notation guide}
\label{sec:notation}
\begin{table}[!h]\renewcommand{\arraystretch}{1.25}\centering
	\begin{tabular}{r l}
		${\rm e}(x)$ & Shorthand notation ${\rm e}(x)={\rm e}^{2\pi\mathrm{i}x}$.
		\\
		$B_m(x)$ & Bernoulli polynomials with generating function $\frac{t{\rm e}^{xt}}{ {\rm e}^t-1} = \sum_{n=0}^\infty B_n(x)\frac{t^n}{ n!}$.
		\\
		$\Lambda=\Lambda_{\mathfrak{g}}$ & The root lattice associated to the simply-laced Lie algebra $\mathfrak{g}$.
\\
		$\Lambda^{\vee}$ & The dual root lattice.
\\
		$\Phi_{\rm s}$ & The set of simple roots $\{\vec{\al}_i\}_{i=1}^{\rm rankG}$.
			\end{tabular}
\end{table}

\begin{table}[!ht]\renewcommand{\arraystretch}{1.3} \centering
	\begin{tabular}{r l}
		$\Phi_{\pm}$ & The sets of positive and negative roots.

\\		$\vec{\rho}$ & The Weyl vector of the root system $\vec{\rho}:=\frac{1}{2} \sum_{\vec{\al}\in\Phi_+} \vec{\al}$.
		\\
		$\langle \cdot,\cdot\rangle$ & The scalar product in the dual space of the Cartan subalgebra of $\mathfrak{g}$.
		\\
		$|\vec{x}|^2$ & For $\vec{x}\in\mathbb{C}\otimes_\mathbb{Z} \Lambda$, the norm is defined by $|\vec{x}|^2=\langle \vec{x},\vec{x}\rangle$.
		\\
		$\{\vec{\omega}_i\}_{i=1}^{\rm rankG}$ & The set of fundamental weights, satisfying $\langle \vec{\omega}_i,\vec{\alpha}_j \rangle = \delta_{i,j}$.
		\\
		$P^+$ & The set of dominant integral weights. See \eqref{eqn:P+}. For $\bar{P}^+$, see \eqref{eqn:barP+}.
		\\
		$\Delta \vec\omega$ &The difference $\Delta \vec\omega:= \vec\omega_1-\vec\omega_2$
		of the two fundamental weights in $A_2$ Lie algebra.
 		\\
		$W$ & The Weyl group of the root system.
		\\
		$w(\cdot )$ & The action of the element $w\in W$. \\ $\ell(w)$&
 The length of $w$.
		\\
		$Q(\mathbf{m})$ & The norm $ Q(\mathbf{m}):= \frac{1}{2}|\vec {\mathbf m}|^2= \bigl(3m_{1}^2+3m_1m_2+m_2^2\bigr)$ for $(m_1,m_2)\in {\mathbb R}^2$ \eqref{def:quadraticformQx}.\\
 $\bos$ & Shorthand notation $\bos = \bigl( \vec{s},\vec{k},m,D\bigr)$ introduced in Section~\ref{sec:gena2}.
		\\
		$\vec{\sigma}$ & Shorthand notation $\vec{\sigma}=\vec{s}-\tfrac{m}{D
 }\vec{k}$.
		\\
		$F^{(\bos)}$ & Generalized $A_2$ false theta function defined in equation \eqref{dfn:gen_false}.
		\\
		$F_\nu^{(\bos)}$ & Partial theta functions, defined for $\nu=0,1$ in \eqref{def:Fiq2_1} and \eqref{def:Fiq}. 
 \\
 $F_{\nu,\bal}$ & Components of false theta functions, defined in equation \eqref{def:Fiq2_1-comp}.
		\\
 $\CF_\nu$ & $\CF_\nu(\boldsymbol{x})=x_2^\nu {\rm e}^{-Q(\boldsymbol{x})}$ \eqref{dfn:Fcalnu-O} \\
		$\CS$ & The set \eqref{def:Scal} of parameters $\bal$ \eqref{eq:alphawidef0A} of the partial theta functions $F_{\nu,\bal}(\tau)$. \\
 $\ss$ & Subset of the set $\CS$ defined in equation \eqref{def:Scal2}.
		\\
		$\BE_\nu^{(\bos)}(\tau)$ & The companion functions of the functions $F_\nu^{(\bos)}(\tau)$. See \eqref{def:E01starless}.
 \\
		$M$ & The adjacency matrix \eqref{adjacencymatrix} of the weighted graph $(V,E,a)$.
		\\
		$D$ & Smallest positive integer such that $D M^{-1}_{v_0,v}\in\mathbb{Z}$ for $\forall  v\in V$; $m=D^2 \big|M^{-1}_{v_0,v_0}\big|$.
		\\
		$\vec{\underline{b}}$ & Generalised Spin$^c$ structure \eqref{eq:bfromcoker} on a plumbed three-manifold $M_3$, labelling \\
 & the boundary conditions of $T[M_3]$.
 \\
		$\widehat{Z}_{\vec{\underline{b}}}^G(M_3)$ & Topological invariant of a plumbed three-manifold $M_3$ \eqref{def:ZhatHR-3mfdvoa}.
		\\		\smash{$\widecheck{Z}_{\vec{\underline{b}}}^G(M_3)$} & Companion function of \smash{$\widehat{Z}_{\vec{\underline{b}}}^G(M_3)$}. \\
 $\stackrel{\dots}{=}$, $\stackrel{\dots}{\in}$, etc. & relations hold when replacing $\widehat Z$ with $Cq^{\Delta}\widehat Z+f(q)$, for some \\
 & $C\in \CC$, $\Delta\in \QQ$, $f(q)\in \CC\big[q,q^{-1}\big]$, and similarly for $\widecheck Z$.
	\end{tabular}
\end{table}

\section[Generalized A\_2 false theta functions]{Generalized $\boldsymbol{A_2}$ false theta functions}\label{sec:gena2}

Let $\Lambda=\Lambda_{A_2}$ be the $A_2$ root lattice, $W$ the corresponding Weyl group with $\ell\colon W \to \ZZ$ its length function. We denote by $W_+\cong \ZZ/3$ the rotation subgroup of $W$ given by the kernel of the map $w\mapsto (-1)^{\ell(w)}$. We also denote by $\Phi_{\rm s}=\{\vec \alpha_1, \vec \alpha_2\}$ a set of simple roots and $\{\vec \omega_1, \vec \omega_2\}$ the corresponding fundamental weights, $\Phi_\pm$ the set of positive resp.\ negative roots, and by
\be
\label{eqn:P+}
P^+:=\bigl\{\vec \lambda \in \Lambda^\vee
\mid \langle \vec \lambda, \vec \alpha\rangle {>} 0\, \forall \vec\alpha \in \Phi_+
\bigr\}
\ee
the set of dominant integral weights, where $\langle \cdot, \cdot \rangle$ is a quadratic form given by the $A_2$ Cartan matrix. For $\vec{x} \in \mathbb{C} \otimes_\mathbb{Z} \Lambda $, we define the norm $| \vec{x} |^2 := \langle \vec{x}, \vec{x} \rangle$ as usual. We will also define
\be
\label{eqn:barP+}
\bar P^+:=\bigl\{\vec \lambda \in \Lambda^\vee
\mid \langle \vec \lambda, \vec \alpha\rangle \geq 0 \, \forall \vec\alpha \in \Phi_+
\bigr\} .
\ee

It will be convenient to introduce the map
\[
\ZZ^2 \to \Lambda, \qquad \mathbf{m}= (m_1,m_2)\mapsto \vec {m} :=m_2 \vec\omega_1 + (3m_1+m_2)\vec \omega_2 ,
\]
the corresponding norm
\begin{equation}\label{def:quadraticformQx}
	Q(\mathbf{m}):= \frac{1}{2}|\vec { m}|^2=\bigl(3m_{1}^2+3m_1m_2+m_2^2\bigr),
\end{equation}
and the following functions on $\BR^2$
\begin{gather}\label{dfn:Fcalnu-O}
	\mathcal{F}_{\nu}( \mathbf{x} ):=x_2^\nu {\rm e}^{-Q( \mathbf{x} )}, \qquad\nu=0,1.
\end{gather}

Then given a vector $\vec s$ in the root lattice, a positive integer $m$, a divisor $D$ of $m$, and $\vec k\in \Lambda/ D\Lambda$, we define the {\em generalized $A_2$ false theta function}
\be\label{dfn:gen_false}
F^{(\bos)} (\tau) = \sum_{w\in W} (-1)^{\ell(w)} \sum_{\substack{ \vec n\in \Lambda \cap P^+\\ \vec n \in w^{}(\vec k ) + D\Lambda}} {\rm min}(n_1,n_2) q^{\frac{1}{2m} |-w(\vec s) + \frac{m}{D} \vec n|^2},
\ee
where $\bos$ encodes the data
$(\vec{s},\vec{k},m,D)$.
The $A_2$ false theta functions, whose quantum modularity has been studied in \cite{bringmann2018higher} and which appear in the character of higher rank logarithmic vertex algebra {log-${\cal V}_{\bar \Lambda}(m)$} \cite{feigin2010logarithmic}, always have $D=1$.

These generalized $A_2$ false theta functions are the building blocks of the \smash{$\widehat Z^{{\rm SU}(3)}_\bv(M_3)$} invariants, when $M_3$ is a negative Seifert manifold with three exceptional fibers \eqref{s3:hombl-2}. The study of their explicit quantum modular properties will be the subject of this section.

In the above and elsewhere in this paper, unless stated otherwise, we use the weight basis notation. For instance, we use $(n_1, n_2)$ to denote $\vec n := n_1\vec \omega_1 + n_1\vec \omega_2 \in \Lambda^\vee$. We also write $\vec k\lvert_i := \langle \vec k, \vec \alpha_i\rangle$ for $\vec k\in \CC\otimes_{\ZZ}\Lambda$, so $\vec n\lvert_i =n_i$ for $\vec n=(n_1,n_2)$.

\subsection{Identities}
\label{subsec:defns}

We now rewrite the generalised $A_2$ false theta function in a form which allows us to determine its asymptotic behaviour in the limit where the modular parameter $\tau$ approaches a rational number. Similar to \cite{bringmann2018higher} we will first rewrite \eqref{dfn:gen_false} as a sum over partial theta functions. Concretely, we have the following lemma.

\begin{lem}\label{lem:split_2_terms}
	With the notation of \eqref{dfn:gen_false}, we choose a representative of $\vec k\in \Lambda/D\Lambda$ such that $0\leq \langle \vec k, \vec\omega_i\rangle <D$ for $i=1,2$, and write $\vec{s}=\vec{\sigma}+\frac{m}{D}\vec{k}$.
 Then we have
	\begin{equation}\label{Section2-Fsplit}
		F^{(\bos)}(\tau) =
		F_{0}^{(\bos)}(m\tau)+D F_{1}^{\left(\bos\right)}(m\tau),
	\end{equation}
	where
	\begin{gather}\label{def:Fiq2_1}
			F_{0}^{\left( \bos \right)}( \tau)
			:= \frac{D}{m}\sum_{w\in W_+}\sum_{i \in \{1,2\}} w(\vec s)\lvert_i
			F_{0,\bal^{(i)}_w}(\tau),\qquad
			F_{1}^{\left( \bos \right)}( \tau) :=\sum_{w\in W_+}\sum_{i \in \{1,2\}}
			F_{1,\bal^{(i)}_w}(\tau)
	\end{gather}
	with
	\begin{equation}\label{def:Fiq2_1-comp}
		F_{\nu,\bal}(\tau) = \Bigg( \sum_{\mathbf{n}\in \bal+\mathbb{N}_{0}^{2}}+(-1)^\nu\sum_{\mathbf{n}\in {\bf 1}- \bal+\mathbb{N}_{0}^{2}} \Bigg) n_2^\nu q^{Q( \mathbf{n})}.
	\end{equation}
	The $\bal^{(i)}_{w}$ vectors are defined in terms of $\vec{s}$, $\vec{k}$, and Weyl group element $w\in W$
	\begin{gather}
		\bal_w^{(1)}=
		\left(x+ \frac{\Delta w(\vec{\sigma}) }{m},
		\xi_{w,1} - \frac{w(\vec{\sigma})\lvert_1}{m}\right), \qquad
		\bal_w^{(2)}=
		\left(1- x-\frac{\Delta w(\vec{\sigma}) }{m},
		\xi_{w,2} - \frac{w(\vec{\sigma})\lvert_2}{m}\right)
		\label{eq:alphawidef0A}
	\end{gather}
	in which
	\begin{gather*}
{ \xi_{w,i}:= \biggl\lceil - \frac{w\bigl(\vec{k}\bigr)\vert_i}{D} \biggr\rceil }, \qquad \Delta w(\vec{\sigma}) : = \frac{w(\vec{\sigma})\lvert_1-w(\vec{\sigma})\lvert_2 }{3}, \\
		x= \begin{cases} 0 & {\rm when}\ w\bigl(\vec{k}\bigr)\lvert_2\geq w\bigl(\vec{k}\bigr)\lvert_1, \\ 1 & {\rm otherwise}.
		\end{cases}
	\end{gather*}
\end{lem}
The proof can be found in Appendix~\ref{pf_lem:split_2_terms}.

For later convenience, we will also use the following rewriting of \eqref{def:Fiq2_1}:
\begin{equation}\label{def:Fiq}
	F_{0}^{( \bos )}( \tau)
	=\sum_{\bal\in\mathcal{S}}\eta_0( \bal)
	\sum_{\mathbf{n}\in \bal+\mathbb{N}_{0}^{2}}q^{Q( \bon )},
	\qquad
	F_{1}^{( \bosi )}(\tau )
	=\sum_{\bal\in\mathcal{S}}\eta_1(\bal)
	\sum_{\bon\in \bal+\mathbb{N}_{0}^{2}}n_{2}\,q^{Q( \mathbf{n} )},
\end{equation}
where we write
\begin{equation}\label{def:Scal}
	\CS = \bigcup_{w\in W_+} \bigl\{ \bal_w^{(1)},\, \bal_w^{(2)}, \bar\bal_w^{(1)},\, \bar{\bal}_{w}^{(2)} \bigr\},
\end{equation}
and
\begin{gather}\label{def:vareps}
	\bar\bal_w^{(i)} := {\bf 1} - \bal_w^{(i)}, \qquad
	\eta_0\bigl(\bal_w^{(i)}\bigr) = \eta_0\bigl(\bar{\bal}_{w}^{(i)}\bigr) = \frac{D}{m} w(\vec{s})\lvert_i
	, \qquad\eta_1\bigl(\bal_w^{(i)}\bigr) = -\eta_1\bigl(\bar{\bal}_{w}^{(i)}\bigr)= 1,\!\!\!
\end{gather}
where ${\bf 1}=(1,1)$. We will also write
\begin{equation}\label{def:Scal2}
	\ss = \bigcup_{w\in W_+} \bigl\{ \bal_w^{(1)}, \bal_w^{(2)} \bigr\},
\end{equation}
so that
\[
	F_{\nu}^{( \bos )}( \tau)
	=\sum_{\bal\in\ss}\eta_\nu( \bal)
	F_{\nu,\bal}.
\]

Note that $m\bal \in \ZZ^2$, since $\vec \sigma \in \Lambda$ and hence $\Delta w(\vec{\sigma}) = \langle \Delta\vec \omega, w(\vec\sigma)\rangle \in \ZZ$, where we write \[\Delta\vec \omega = \vec \omega_1-\vec \omega_2 = \frac{1}{3} ( \vec \alpha_1-\vec \alpha_2 ).\]

To study the radial limit of $F^{(\bos)}$, later we will be working with functions of the form $\sum_{\mathbf{n}\in \bal+\mathbb{N}_{0}^{2}} {\cal F}_\nu({\bf n}) $ for $\nu=0,1$ with $0\leq \alpha_i \leq 1$ for $i=1,2$. See (\ref{dfn:Fcalnu-O}). It will therefore be useful to note the following result on the effect of integral shifts of $\bal$.

\begin{lem}\label{lem:shifting}
	Let ${\bb}=\bal + (\delta \a_1,\delta \a_2)$ for $\delta \a_1, \delta \a_2 \in {\mathbb Z}$.
	Consider $F_{\nu,\bal}(\tau)$ for $\nu=0,1$ as defined in \eqref{def:Fiq2_1}. Then $F_{\nu,\bb}(\tau) - F_{\nu,\bal}(\tau)$ is in the integral linear span of one-dimensional lattice sums $\bigl\{{\tilde\theta}^1[\kappa,a], {\tilde\theta}^0[\kappa,a] \mid \kappa, a \in \QQ \bigr\}$, up to the addition of a finite polynomial $p(\tau)\in q^{\Delta}{\mathbb Z}[q]$, where
	\begin{gather*}
			{\tilde \theta}^0 [\kappa,a](\t) := \sum_{n\in \ZZ} |n+a| q^{\kappa(n+a)^2},\qquad
			{\tilde\theta}^1[\kappa,a](\t) := \sum_{n\in \ZZ} \operatorname{sgn}(n+a) q^{\kappa(n+a)^2}.
	\end{gather*}
\end{lem}

Note that ${\tilde \theta}^0$ and ${\tilde \theta}^1$ are themselves Eichler integrals of weight $1/2$ resp.\ $3/2$ theta functions, up to finite polynomials. See \eqref{dfn:Eichler_integral}.
The proof can be found in Appendix~\ref{pf_lem:shifting}.
Using the above lemma to shift vectors by integers, in the following section we will consider vectors $\bmu=(\mu_1,\mu_2)$ satisfying $ 0\leq \mu_1, \mu_2 \leq 1$.

\subsection{Radial limits}
\label{subsec:rad_lim}

In this subsection we aim to study the radial limit $\tau \to \frac{h}{k}\in\mathbb{Q}$, approached from the upper-half~plane $\mathbb H$, of the generalized $A_{2}$ false theta functions \smash{$F^{(\bos)}(\tau)$}. To do so, we will use the Euler--Maclaurin summation formula, a strategy also employed in \cite{bringmann2018higher}. First we will recall the following asymptotic expansion formula, which goes back to \cite{Zagier77}\footnote{To apply the formulas in \cite{Zagier77} correctly, it is important that the shift vector, denoted $\bos$ in this paper, must satisfy $\mu_1, \mu_2 \geq 0$.}.

For $\bmu=(\mu_1,\mu_2)$ with $ \mu_1, \mu_2\geq 0$, and $F\colon \BR^2_{\geq 0}\to\BR$ a smooth rapidly decaying $C^\infty$ function, the asymptotic expansion in the limit $t\to 0^+$ of $F$ is given by \cite{Zagier77}
\begin{gather}
		\sum_{\mathbf{n}\in\mathbb{N}_0^2} F((\mathbf{n} + \bmu) t ) \sim \frac{\mathcal{I}_F}{t^2} - \sum_{n\in\mathbb{N}} \frac{t^{n-2} }{ n! } \int_0^\infty {\rm d}x \bigl( {B_{n}(\mu_1)} F^{(n-1,0)} (0,x) +{B_{n}(\mu_2)} F^{(0,n-1)} (x,0) \bigr) \nonumber\\
\hphantom{\sum_{\mathbf{n}\in\mathbb{N}_0^2} F((\mathbf{n} + \bmu) t ) \sim}{}
 +
		\sum_{\mathbf{n}\in\mathbb{N}^2} \frac{ t^{n_1+n_2-2} }{ n_1! n_2! }\,B_{n_1}(\mu_1) B_{n_2}(\mu_2)
		F^{(n_1-1,n_2-1)} (0,0),\label{def:EM}
\end{gather}
where $\sim$ means that the two sides agree up to $O\bigl(t^N\bigr)$ terms for any $N\in\mathbb{N}$ and $\mathcal{I}_F$ is given by
\[
\mathcal{I}_F:=\int_0^\infty\int_0^\infty F(x_1,x_2){\rm d}x_1{\rm d}x_2 .
\]
In the above expression, $B_m(x)$ are the Bernoulli polynomials whose generating function is given by $\frac{t{\rm e}^{xt}}{ {\rm e}^t-1} = \sum_{n=0}^\infty B_n(x)\frac{t^n}{n!}$. A key feature of these polynomials that follows directly from the generating function is their reflection property
\begin{equation}\label{reflectionpropertyBernoulli}
	B_{m}(x)=(-1)^m B_{m}(1-x) .
\end{equation}

In order to apply \eqref{def:EM} to derive the radial limit, we will further rewrite our generalized false theta function \eqref{def:Fiq2_1} for when $\operatorname{Re}\tau \in \QQ$: for coprime integers $h$, $k$ and $t\in \RR_{>0}$, we have for~$\nu=0,1$
\begin{gather}
	F_\nu^{(\bos)}\bigl(\tfrac{h}{k} +\tfrac{it}{2\pi}\bigr) = \bigl(\sqrt{t}\bigr)^{-\nu} \sum_{\boldsymbol{\mu} \in \CS} \eta_\nu(\boldsymbol{\mu})
	\sum_{ \boldsymbol{\ell}\in (\ZZ/k\bar m)^2 } \! \ex\bigl(\tfrac{h}{k}Q(\bol+\bmu)\bigr)
 	\sum_{\substack{\boldsymbol{n} \in \frac{1}{k\bar m}(\boldsymbol{\ell} +\boldsymbol{\mu} ) + \mathbb{Z}^2\\ k\bar m \boldsymbol{n} \in \bmu + {\mathbb N}_0^2}}\!
	\CF_\nu \bigl( k\bar m \sqrt{t} \boldsymbol{n} \bigr),\!\!\!\label{F112new}
\end{gather}
where we have defined
\begin{equation}\label{dfn:symb}
	\delta:=(h,m),\qquad\bar m:= \frac{m}{ \delta},
\end{equation}
and $\CF_\nu(\boldsymbol{x})$ is given as in (\ref{dfn:Fcalnu-O}).
To see that the sum over $\bol$ is well-defined, note that $m \bmu \in \BZ^2$ for all $\bmu \in {\cal S}$. To derive the asymptotic expansion (see Proposition \ref{prop:Fs12asymptotics}) of \smash{$F_\nu^{(\bos)}$}, we first establish the following lemma.
\begin{lem}
	\label{lem:vanishingmainterm}
	Let $\CS$, $\eta_\nu$, and $\bar m$ be as given in \eqref{def:Scal},
	\eqref{def:vareps} and \eqref{dfn:symb}.
	Then for $\nu=0,1$
	\begin{gather*}
		\sum_{\bmu\in\CS}\eta_\nu(\bmu)\sum_{\boldsymbol{\ell}\in (\ZZ/k\bar m)^2 }\ex\bigl(\tfrac{h}{k}Q(\bol+\bmu)\bigr) =0.
	\end{gather*}
\end{lem}

See Appendix \ref{pf_lem:vanishingmainterm} for the proof.

{\samepage\begin{lem}
	\label{lem:vanishingtermsWB-odd}
	Given $w\in W_+$
	\begin{equation*}
		\sum_{\bmu\in \{ \bmu_w^{(1)},\bmu_w^{(2)} \}}\sum_{0\leq \ell_1,\ell_2 < k\bar m} B_n\left( \frac{\ell_{1}+\mu_{1}}{k\bar{m}} \right)\ex{\left(\frac{h}{k}Q\left( \bol+\bmu \right)\right)}=0
	\end{equation*}
for any odd positive integer $n$.
\end{lem}

See Appendix \ref{pf_lem:vanishingtermsWB-odd} for the proof.}

After establishing the above lemmas, upon using equations \eqref{def:EM} and \eqref{F112new} we are now ready to prove the following asymptotic formula.

\begin{prop}\label{prop:Fs12asymptotics}
	For $\nu=0,1$, the asymptotic limit
 near $\tfrac{h}{k}$ is given by
	\begin{align}
				F_\nu^{(\bos)}\left(\frac{h}{k} +\frac{{\rm i}t}{2\pi}\right) \sim{}&-2
			\sum_{\bal\in\ss} \eta_\nu{(\bmu)}
			\sum_{ 0\leq \ell_1,\ell_2 < k\bar m } \ex\left(\frac{h}{k}Q(\bol+\bmu)\right) \Bigg(
			\sum_{\substack{ n> 1\\n\equiv \nu (2)} }
			\frac{ \left( k\bar{m} \right)^{n-2} t^{\frac{n-2-\nu}{2}} }{n!}\nonumber\\
			&\times\int_0^\infty {\rm d}x \left( B_{n}\left(\frac{\ell_2+\mu_2}{k\bar{m}}\right) \CF_\nu^{(0,n-1)} (x,0) +B_{n}\left(\frac{\ell_1+\mu_1}{k\bar{m}}\right)\right.\nonumber
			\\
			&\times\left.\CF_\nu^{(n-1,0)} (0,x) \right) -
			\sum_{ \substack{ \mathbf{n}\in\mathbb{N}^2 \\ n_1 \equiv n_2 +\nu (\tmod\, 2) }}
			\frac{\left( k\bar{m}\right)^{n_1+n_2-2} t^{\frac{n_1+n_2-2-\nu}{2}} }{n_1! n_2!}\nonumber\\
&\times B_{n_1}\left(\frac{\ell_1+\mu_1}{k\bar{m}}\right)B_{n_2}\left(\frac{\ell_2+\mu_2}{k\bar{m}}\right)\CF_\nu^{(n_1-1,n_2-1)} (\boldsymbol{0}) \bigg) .\label{def:F11new-2}
\end{align}
\end{prop}

\begin{proof}
	In \eqref{F112new}, choose the sum over $\bol$ to be over the range $ 0\leq \ell_1, \ell_2 < k\bar m$ and apply the Euler--Maclaurin summation formula \eqref{def:EM} to $\sum_{\mathbf{n}\in\mathbb{N}_0^2} F((\mathbf{n} + \boldsymbol\mu') t' )$, with
	\[
	\bmu'= \frac{\bol+\bmu }{k\bar m},\qquad t'=k\bar m \sqrt{t},\qquad F({\boldsymbol x}) = \CF_\nu(\boldsymbol x) .
	\]
	First note that the potential divergent, $\bol$- and $\bmu$-independent term ${{{\cal I}_F}/{t^2}}$ actually vanishes contribution due to Lemma \ref{lem:vanishingmainterm}. Second, note that the reflection property \eqref{reflectionpropertyBernoulli} of the Bernoulli polynomials leads to the identity
	\begin{gather}
			\sum_{0\leq \ell_1,\ell_2< k\bar m}\ex\left(\frac{h}{k}Q(\bol+\bmu)\right) B_n\left(\frac{\ell_i +\mu_i}{k\bar m}\right)\nonumber
			\\ \qquad = (-1)^n
			\sum_{0\leq \ell_1,\ell_2< k\bar m}\ex\left(\frac{h}{k}Q(\bol+{\bf 1}-\bmu)\right) B_n\left(\frac{\ell_i +1-\mu_i}{k\bar m}\right)\label{id:ref1}
\end{gather}
	for $i=1,2$, and
	\begin{gather}
		\sum_{0\leq \ell_1,\ell_2< k\bar m}\ex\left(\frac{h}{k}Q(\bol+\bmu)\right) B_{n_1}\left(\frac{\ell_1 +\mu_1}{k\bar m}\right)B_{n_2}\left(\frac{\ell_2 +\mu_2}{k\bar m}\right)\nonumber
			\\
\qquad= (-1)^{n_1+n_2}
			\sum_{0\leq \ell_1,\ell_2< k\bar m}\ex\left(\frac{h}{k}Q(\bol+{\bf 1}-\bmu)\right)\nonumber\\
\phantom{\qquad=}{}\times B_{n_1}\left(\frac{\ell_1 +1-\mu_1}{k\bar m}\right)
			B_{n_2}\left(\frac{\ell_2 +1-\mu_2}{k\bar m}\right),\label{id:ref2}
\end{gather}
	where we have shifted the sum over $\bol$ to $-\bol + {\bf 1}\left(k\bar{m} - 1\right)$. From \eqref{def:vareps} and \eqref{def:Scal2}, since $\bmu$ and $\bar \bmu:={\bf 1}-\bmu$ appear in $\CS$ in pairs, we can fold the sum into a sum over $\tilde \CS$. Moreover, from~$\eta_\nu(\bar \bmu) = (-1)^\nu \eta_\nu(\bmu)$, the above identity implies that terms in the sum in the third line of \eqref{def:F11new-2} vanish unless $n_1+ n_2 \equiv \nu (2)$. Similarly, the terms in the second line of \eqref{def:F11new-2} vanish unless $n\equiv \nu (2)$. To show that the potentially divergent term with $n=1$ when $\nu=1$ vanishes, we first note that $\CF_1^{(0,0)}(x,0) =\CF_1(x,0) = 0$ and we are hence left to show that
	\begin{gather*}
			\int_0^\infty {\rm d}x \CF_1(0,x) \,\sum_{\bmu\in\ss} \eta_\nu{(\bmu)}
			\sum_{ 0\leq \ell_1,\ell_2 < k\bar m } \ex\left(\frac{h}{k}Q(\bol+\bmu)\right) B_1\left( \frac{\ell_{1}+\mu_{1}}{k\bar{m}} \right) \\
			\qquad=
			\int_0^\infty {\rm d}x \CF_1(0,x) \,\sum_{w\in W_+}\sum_{\bmu \in \{ \bmu_w^{(1)}, \bmu_w^{(2)}\}}
			\sum_{ 0\leq \ell_1,\ell_2 < k\bar m } \ex\left(\frac{h}{k}Q(\bol+\bmu)\right) B_1\left( \frac{\ell_{1}+\mu_{1}}{k\bar{m}} \right) = 0,
\end{gather*}
	which is true by Lemma \ref{lem:vanishingtermsWB-odd}.
\end{proof}

\subsection{Companions}\label{subsec:companions}

Having established the asymptotic expansions of the functions $F_{\nu}^{( \bosi)}$ in the limit $\tau \to \frac{h}{k} \in \QQ$, in this subsection we will show that certain functions \smash{$\BE_\nu^{*(\bos)}(\tau)$}, consisting of generalised complementary error functions, are their companion functions in the sense that they have compatible asymptotic behaviour.

\begin{defn}\label{def:companions}
We say two functions $\widehat F$ and $\widecheck F$ on the upper-half plane are {\it companions} of each other if their asymptotic expansions near the rationals satisfy
\[
\widehat F\left(\frac{h}{k} +\frac{{\rm i}t}{2\pi}\right) \sim\sum_{m\geq 0} a_{h,k}(m) t^m,
\]
and
\[
\widecheck F\left(\frac{h}{k} +\frac{{\rm i}t}{2\pi}\right) \sim\sum_{m\geq 0} a_{-h,k}(m) (-t)^m,
\]
for all coprime integers $h$, $k$ with $k>0$.
\end{defn}
Importantly, given a function on the upper-half-plane, its companion is anything but unique; the definition of the companion function is insensitive to the addition of functions vanishing at all rationals.

To establish companions of the generalised $A_2$ false theta functions, we define for $\nu=0, 1$
\begin{gather}\label{def:E01star}
\BE_\nu^{*(\bos)}(\tau)
:= \frac{1}{2} \sum_{\bmu \in \CS}\eta_\nu(\bmu)\Biggl(\sum_{\bon \in \bmu + \NN_0^2} g_\nu(n_1,n_2)
+ \sum_{\bon \in (1-\mu_1,\mu_2) + \NN_0^2} g_\nu(-n_1,n_2) \Biggl),
\end{gather}
where
\begin{align}
 &g_\nu(n_1,n_2):= q^{-Q(\bon)} \bigg(n_2^\nu M_2^\ast \bigl( \sqrt{3} ; \sqrt{3v} (2n_1 + n_2), \sqrt{v}n_2\bigr)\nonumber\\
& \hphantom{g_\nu(n_1,n_2):=}{}
+ \delta_{\nu,1}\frac{ {\rm e}^{-\pi v (3n_1+2n_2)^2}}{2 \pi \sqrt{v}} M^\ast\bigl(\sqrt{3v} n_1\bigr) \bigg)\label{gnu}
\end{align}
and
$v:=\operatorname{Im}\tau$,
and show that, when writing the asymptotic expansion of $F_\nu^{(\bos)}$ as
\[
F_\nu^{(\bos)}\left(\frac{h}{k} +\frac{{\rm i}t}{2\pi}\right) \sim\sum_{m\geq 0} a_{h,k}^{(\nu)}(m) t^m,
\]
we have the following proposition.
\begin{prop}\label{prop:E01asymp}
	$\BE_{\nu}^{*(\bos)}\left( \tau \right)$ as defined in equation \eqref{def:E01star} is a companion of $F_\nu^{(\bosi)}(\tau)$, whose asymptotic expansion in the limit $\tau=\frac{h}{k}+\frac{{\rm i} t}{2\pi}$, $t\to 0^+$ satisfies
	\[
		\BE_{\nu}^{*(\bosi)}\left(\frac{h}{k}+\frac{{\rm i}t}{2\pi} \right)\sim
		\sum_{m\ge 0}a_{-h,k}^{(\nu)}\left( m \right)\left( -t \right)^{m}.
	\]
\end{prop}

For the proof of the proposition, it will be convenient to define the following functions and establish the identities in Lemma \ref{lem:F1F3identities}. For $\nu=0,1$, let the functions $\CG_\nu, \tilde \CG_\nu\colon \RR^2\to \RR$ be given~by
\[
	\CG_\nu(x_1,x_2) = \frac{1}{2} x_2^\nu M_2^* \left( \sqrt{3}; \frac{ \sqrt{3} (2x_1+x_2) }{\sqrt{2\pi}}, \frac{x_2}{\sqrt{2\pi}} \right) {\rm e}^{Q(\boxx)} = \tilde \CG_\nu(-x_1,x_2),
\]
where the functions $M_2$, $M_2^*$ are defined in equations \eqref{def:M2} and \eqref{def:M2star}, respectively, following \cite{Alexandrov:2016enp}. The following relations to ${\cal F}_\nu$ have been established in \cite[Section~7]{bringmann2018higher}.
\begin{lem}\label{lem:F1F3identities}
	For $\nu=0,1$ the following identities hold for $n\in \NN$, $n\equiv \nu +1~(2)$:
	\begin{gather*}
				\int_0^\infty {\rm d}x \CF_\nu^{(0,n)}(x,0) = (-1)^{\lfloor \frac{n-1}{2}\rfloor}
			\int_0^\infty {\rm d}x \bigl(\CG_\nu^{(0,n)}+\tilde\CG_\nu^{(0,n)} \bigr) (x,0),\\
			\int_0^\infty {\rm d}x \CF_\nu^{(n,0)}(0,x) = (-1)^{\lfloor \frac{n-1}{2}\rfloor}
			\int_0^\infty {\rm d}x \bigl(\CG_\nu^{(n,0)}+(-1)^\nu \tilde\CG_\nu^{(n,0)} \bigr) (0,x) -
			\frac{1}{\sqrt{2}} \left[\frac{{\rm d}^{n}}{{\rm d} y^{n}} {\rm e}^{\frac{3y^2}{4}} \right]_{y=0}
 \end{gather*}	
	and
	\[
		\CF_\nu^{(\bon)} (\boldsymbol{0}) =
		(-1)^{\lfloor\frac{n_1+n_2}{2}\rfloor} \bigl( \CG^{(\bon)}_\nu (\boldsymbol{0}) + {( -1 )^{n_{1}+1}}\tilde{\CG}^{(\bon)}_\nu (\boldsymbol{0}) \bigr)
	\]
	for $n_1+n_2 \equiv \nu ~(2)$.
\end{lem}
Now we are ready to prove Proposition \ref{prop:E01asymp}. As before, we can re-express $\BE_\nu^{\ast(\bos)}(\tau)$ when $ \operatorname{Re}\tau\in \QQ$ as
\begin{align}
\BE_\nu^{*(\bos)}\left(\frac{h}{k} +\frac{{\rm i}t}{2\pi}\right) ={}&\bigl(\sqrt{t}\bigr)^{-\nu}
\sum_{\boldsymbol{\mu} \in \CS} \eta_\nu(\boldsymbol{\mu})
\sum_{ \boldsymbol{\ell}\in (\ZZ/k\bar m)^2 } \bigg( \ex\left(-\frac{h}{k}Q(\bol+\bmu)\right)
 \nonumber\\
&
\times\sum_{\substack{\boldsymbol{n} \in \frac{1}{k\bar m}(\boldsymbol{\ell} +\boldsymbol{\mu} ) + \mathbb{Z}^2\\
k\bar m \boldsymbol{n} \in \bmu + {\mathbb N}_0^2}}
\CG_\nu' \left( k\bar m \sqrt{t} \boldsymbol{n} \right)+\ex\left(-\frac{h}{k}Q(-1+\mu_1-\ell_1,\mu_2+\ell_2)\right)\nonumber\\
&\times
\sum_{\substack{\boldsymbol{n} \in \frac{1}{k\bar m}(\boldsymbol{\ell}+ (1-\mu_1,\mu_2) ) + \mathbb{Z}^2\\ k\bar m \boldsymbol{n} \in (1-\mu_1,\mu_2) + {\mathbb N}_0^2}} \tilde\CG_\nu' \bigl( k\bar m \sqrt{t} \bon \bigr)
\bigg),\label{def:E0exp}
\end{align}
where
\begin{equation}\label{def:tot_G}
\CG_\nu' (x_1,x_2) = \CG_\nu(x_1,x_2) + \delta_{\nu,1}{\frac{1}{2\sqrt{2\pi }}}{\rm e}^{-({\frac{3}{2}}x_1^2+3x_1x_2 +x_2^2)} M^\ast\bigg(\sqrt{\frac{3}{2\pi}}x_1\bigg).
\end{equation}

Applying \eqref{def:EM} to \eqref{def:E0exp}, from Lemma \ref{lem:vanishingmainterm} we see that
\begin{align}
\BE_\nu^{*(\bos)}\left(\frac{h}{k} +\frac{{\rm i} t}{2\pi}\right) \sim{}& 2
\sum_{\bmu\in\ss} \eta_\nu{(\bmu)}
\sum_{ 0\leq \ell_1,\ell_2 < k\bar m } \ex\left(-\frac{h}{k}Q(\bol+\bmu)\right)\biggl( -
\sum_{\substack{ n> 1\\n\equiv \nu (2)} }\frac{ \left( k\bar{m} \right)^{n-2} t^{\frac{n-2-\nu}{2}} }{n!}\nonumber
\\
& \quad\times\bigg[
B_{n}\left(\frac{\ell_2+\mu_2}{k\bar{m}}\right) \int_0^\infty {\rm d}x \bigl(\CG_\nu^{(0,n-1)}+\tilde\CG_\nu^{(0,n-1)} \bigr)(x,0) \nonumber\\
&+B_{n}\left(\frac{\ell_1+\mu_1}{k\bar{m}}\right)\int_0^\infty {\rm d}x \bigl(\CG_\nu^{(n-1,0)}+(-1)^\nu \tilde\CG_\nu^{(n-1,0)} \bigr)(0,x)
 \bigg]	\nonumber\\
& +
\sum_{ \substack{ \mathbf{n}\in\mathbb{N}^2 \\ n_1 \equiv n_2 +\nu (\tmod 2) }}
\frac{\left( k\bar{m}\right)^{n_1+n_2-2} t^{\frac{n_1+n_2-2-\nu}{2}} }{n_1! n_2!}B_{n_1}\left(\frac{\ell_1+\mu_1}{k\bar{m}}\right)\nonumber\\
&
\quad\times B_{n_2}\left(\frac{\ell_2+\mu_2}{k\bar{m}}\right)\bigl(\CG_\nu^{(n_1-1,n_2-1)}-{( -1 )^{n_{1}+1}}\tilde \CG_\nu^{(n_1-1,n_2-1)}\bigr)
(\boldsymbol{0}) \bigr)\label{ap_newBE}
\end{align}
holds for $\nu=0$, where we have also used \eqref{id:ref1}--\eqref{id:ref2} to identify the contribution from $\bmu$ and~$\bar \bmu$. For $\nu=1$, one needs to take the additional term in \eqref{def:tot_G} into account. As shown in detail in \cite[Section~7]{bringmann2018higher}, the contributions of these terms to the asymptotic expansion vanish due to Lemma \ref{lem:vanishingtermsWB-odd}.

Similarly, combining Proposition \ref{prop:Fs12asymptotics} and Lemma \ref{lem:F1F3identities}, and again evoking Lemma \ref{lem:vanishingtermsWB-odd}, the comparison with \eqref{ap_newBE} shows that the Proposition \ref{prop:E01asymp} is true.

\subsection{Eichler integrals}
\label{subsec:EichlerIntegrals}

In this subsection, we will relate the companion of the generalised $A_2$ false theta function $F^{(\rho)}$ to certain Eichler integrals. More precisely, we will show that the companion function $\BE_\nu^{*(\bos)}$ in Proposition \ref{prop:E01asymp} is an Eichler integral given in Proposition \ref{prop:Estarless_eichler}, up to one-dimensional integrals specified in Lemma \ref{lem:Estar_vs_starless}.

To show this, for $\nu =0, 1$ we first define the following functions
\begin{align}
\BE_\nu^{(\bos)}(\tau) :={}& \frac{1}{2} \sum_{\boldsymbol{\mu} \in \ss} \eta_\nu(\boldsymbol{\mu}) \sum_{\mathbf{n} \in \boldsymbol{\mu} + \BZ^2} q^{-Q(\mathbf{n})} \left[ \left({\frac{1}{2\pi {\rm i}}} {\frac{\partial}{\partial z}}\right)^\nu \right.\nonumber\\&
\left. \times\left({\rm e}^{2\pi {\rm i} \nu n_2 z} M_2 \left( \sqrt{3} ; \sqrt{3v} (2n_1 + n_2), \sqrt{v} \left(n_2 - {\frac{2\im(z)}{v}}\right)\right) \right) \right]_{z=0} \label{def:E01starless}
\end{align}
that are closely related to the companion function $\BE_\nu^{*(\bos)}$. More precisely, their difference is given in terms of one-dimensional error function (see \eqref{def:M}) as
\begin{lem}\label{lem:Estar_vs_starless}
 \[
\mathbb{E}^{(\rho)}_{\nu}(\tau)=\mathbb{E}^{*(\rho)}_{\nu}(\tau)+\sum_{\bmu\in\tilde{\CS}}\eta_{\nu}(\bmu)X_\nu(\bmu),
\]
where $X_\nu$ are given by
\begin{align*}
 X_{0}( \bmu )={}& \left( \sum_{\bon\in\bmu+\BN_{0}^{2}}+\sum_{\bon\in( 1,1 )-\bmu+\BN_{0}^{2}}-\sum_{\bon\in(1-\mu_{1},\mu_{2})+\BN_{0}^{2}}-\sum_{\bon\in(\mu_{1},1-\mu_{2})+\BN_{0}^{2}} \right)\\
 &\times\bigl( \delta_{n_{1},0}( 1-\delta_{n_{2},0} )M\bigl( 2\sqrt{v}n_{2} \bigr)+\delta_{n_{2},0}( 1-\delta_{n_{1},0} )M\bigl( 2\sqrt{3v}n_{1} \bigr)-\delta_{n_{1},0}\delta_{n_{2},0} \bigr)q^{-Q( \bon)}\\
 ={}& \begin{cases}
 ( -1 )^{\mu_{1}}\bigl( \sum_{k=\mu_{2}+\BN_{0}}-\sum_{k=1-\mu_{2}+\BN_{0}} \bigr)M\bigl( 2\sqrt{v}k \bigr)q^{-k^{2}},& \mu_{1}\in\{ 0,1 \}\not \ni \mu_{2},\\
 ( -1 )^{\mu_{2}}\bigl( \sum_{k=\mu_{1}+\BN_{0}}-\sum_{k=1-\mu_{1}+\BN_{0}} \bigr)M\bigl( 2\sqrt{3v}k \bigr)q^{-3k^{2}},& \mu_{1}\not\in\{ 0,1 \}\ni\mu_{2},\\
 ( -1 )^{\mu_{1}+\mu_{2}+1},& \mu_{1}\in\{ 0,1 \}\ni\mu_{2},\\
 0,& \mu_{1}\not\in\{ 0,1 \}\not \ni \mu_{2}\\
 \end{cases}
\end{align*}
for $\nu = 0$ and
\begin{align*}
 X_{1}( \bmu )={}& \left( \sum_{\bon\in\bmu+\BN_{0}^{2}}-\sum_{\bon\in\left( 1,1 \right)-\bmu+\BN_{0}^{2}}-\sum_{\bon\in( 1-\mu_{1},\mu_{2} )-\bmu+\BN_{0}^{2}}+\sum_{\bon\in(\mu_{1},1-\mu_{2} )-\bmu+\BN_{0}^{2}} \right)\\
&\times \delta_{n_{1},0}\left( n_{2}M\bigl( 2\sqrt{v}n_{2} \bigr)+\frac{1}{4\pi\sqrt{v}}{\rm e}^{-4\pi n_{2}^{2}v} \right)q^{-n_{2}^2} \\
 ={}& \begin{cases}
( -1 )^{\mu_{1}}\bigl( \sum_{k\in\mu_{2}+\BN_{0}^{2}}+\sum_{k\in1-\mu_{2}+\BN_{0}^{2}} \bigr)\\
\quad\times\bigl( kM\bigl( 2\sqrt{v}k \bigr)+\frac{1}{4\pi\sqrt{v}}{\rm e}^{-4\pi k^{2}v} \bigr)q^{-k^{2}}, & \mu_{1}\in \{ 0,1 \},\\
 0, & \mu_{1}\not\in\{ 0,1 \}
 \end{cases}
\end{align*}
for $\nu = 1$.

\end{lem}
The proof can be found in Appendix \ref{apx:proofEstar_vs_starless}.
Note that in Section~\ref{subsec:companions}, we used $\mathbb{E}^{*(\rho)}_{\nu}$ for the application of Euler--Maclaurin formula, as $g_\nu(n_1,n_2)$ in \eqref{gnu} is continuous on \smash{${\mathbb R}^2_{\geq0}$} as a function of $(n_1,n_2)$, unlike the counterpart in \smash{$\mathbb{E}^{(\rho)}_{\nu}$}; the difference between the two functions then comes precisely from the cases when at least one of $n_1$ and $n_2$ vanishes. Moreover, from \cite{zwegers2008mock}
\[
M(x\sqrt{v})={\rm i}\frac{x}{\sqrt{2}} q^{\frac{x^2}{4}} \int_{-\bar{\tau}}^{{\rm i}\infty} \frac{{\rm e}^{\frac{\pi {\rm i} x^2 w}{2}}}{\sqrt{-{\rm i}(w+\tau)}}{\rm d}w
\]
we see that $X_\nu(\bmu)$ can be written as a linear combination of non-holomorphic Eichler integrals of rank one theta functions \eqref{dfn:theta}, and hence
\begin{equation*}
\mathbb{E}^{(\rho)}_{\nu}=\mathbb{E}^{*(\rho)}_{\nu} + z_{1\rm d}
\end{equation*}
in the notation of \eqref{intro_1dpiece}.

Finally, by carefully rewriting the integrals in the rank two generalised complementary error functions $M_2$ in the definition of \smash{$\BE_\nu^{(\bos)}$}, we arrive at the following relation between the companion and the Eichler integrals, as shown in Appendix \ref{app:proofPropIntegral}.

\begin{prop}\label{prop:Estarless_eichler}
\begin{gather*} 
	\BE_\nu^{(\bos)}(\tau) = \sum_{{w} \in W^+} \BE_{\nu,w}^{(\bos)}(\tau),
 \end{gather*}
where
\be\label{Eichler_Enu_Com}
	\BE_{\nu,w}^{(\bos)}(\tau) := \frac{\sqrt{3}}{4 \pi^{\nu}}
	\int_{-\bar{\tau}}^{{\rm i} \infty}\int_{z_1}^{{\rm i} \infty}
	\frac{ \Theta_{\nu,w}^{(\bos)}(
	{\bf z})}{({-{\rm i}(z_1+\tau)})^{1/2} (-{\rm i}(z_2+\tau))^{\nu+1/2}}{\rm d}z_2{\rm d}z_1
\ee
and
\be\label{def:Theta}
 \Theta_{\nu,w}^{(\bos)}(\mathbf{z}) = (m)^{2\nu-3}	\left( {3D\Delta w(\vec{s})} \right)^{1-\nu} \sum_{\delta\in\BZ/2}
\theta^1_{m,m\delta + \langle\vec{\rho},w(\vec{\sigma})\rangle} \left(\frac{z_1}{m}\right)
\theta^{1-\nu}_{m,m\delta + \langle \Delta\vec\omega, w(\vec{\sigma})\rangle}\left(\frac{3z_2}{m} \right)
\ee
are given by sums of products of two theta functions of one-dimensional lattices.
\end{prop}

Combining Lemma \ref{lem:Estar_vs_starless} and identity (\ref{eq:error_rewriting}), we establish that the companions of the generalised~$A_2$ false theta functions are given in terms of Eichler integrals of rank two and rank one theta functions.

\subsection{Quantum modularity}
\label{subsec:QM}

In this subsection we review the relation between the Eichler integrals discussed in the previous subsection and quantum modular forms.

To define quantum modular forms, we first recall the familiar definition of slash operators, acting on a (vector-valued) function on the compactified upper-half plane $\hat {\mathbb H} := \mathbb H\cup {\mathbb Q} \cup \{{\rm i}\infty\}$: given $k\in\frac{1}{2}\BZ$ and $n$-dimensional multiplier $\chi$ for $\Gamma\subset {\rm SL}_2(\mathbb Z)$, namely a group homomorphism~$\Gamma \to GL_n(\mathbb C)$, and for every $\gamma=\bigl(\begin{smallmatrix} a & b \\ c & d \end{smallmatrix}\bigr)\in\Gamma$, we define the action of the slash operator~$\lvert_{\chi,k}\gamma$ acting on $f=(f_r)\colon\hat {\mathbb H}\to \mathbb C^n$ as
\begin{equation*}
 f\lvert_{k,\chi}\gamma(\tau) := f(\gamma\tau)\chi(\gamma) (c\tau+d)^{-k},
\end{equation*}
where we have written $\gamma \tau = \frac{a\tau+b}{c\tau+d}$ as usual.

\begin{defn}[real-analytic vector-valued quantum modular form]\label{dfn:qmf}
A function $f\colon {\mathbb Q}\to\BC^n$, is a (vector-valued) quantum modular form of weight $k\in\frac{1}{2}\BZ$ with multiplier $\chi$ for $\Gamma \subset {\rm SL}_2(\BZ)$ if for every $\gamma\in\Gamma$ the vector-valued function, the cocycle\footnote{The terminology ``cocycle'' stems from the interpretation of $h_\gamma(\tau)$ in terms of Eichler cohomology via the Eichler--Shimura isomorphism \cite{Eichler1957, Shimura1959SurLI}.}
\[
	h_\gamma(\tau) := f(\tau) - f\lvert_{k,\chi}\gamma(\tau)
\]
can be extended to an open subset of $\BR$ and is real-analytic there. We will denote the vector space of such forms by ${Q}_k(\Gamma,\chi)$.
\end{defn}

\begin{defn}[vector-valued higher depth quantum modular form, see \cite{bringmann2019vectorvalued}]\label{dfn:higher_depth_qmf}
A function $f\colon{\mathbb Q}\to\BC^n$ is a quantum modular form of depth $N\in\BN$ and weight $k\in\frac{1}{2}\BZ$ with multiplier $\chi$ for $\Gamma\subset {\rm SL}_2(\BZ)$ if for every $\gamma\in\Gamma$
\[
	h_\gamma := f - f\lvert_{k,\chi}\gamma \in \bigoplus_j Q_{k_j}^{N_j} (\Gamma,\chi_j) \mathcal{O}(R),
\]
where $j$ runs over a finite set, $k_j\in\frac{1}{2}\BZ$, $N_j\in\BN$ with $\max (N_j)=N-1$, $\chi_j$ are multipliers, $\mathcal{O}(R)$ is the space of real-analytic functions on $R\subset\BR$ which contains an open subset of $\BR$. We also set $Q_k^1 (\Gamma,\chi)= Q_k (\Gamma,\chi)$, $Q_k^0 (\Gamma,\chi)=1$ and $Q_k^N (\Gamma,\chi)$ denotes the space of quantum modular forms of weight $k$, depth $N$, and with $n$-dimensional multiplier $\chi$ for $\Gamma$.
\end{defn}

\subsubsection*{Eichler integrals and quantum modular forms}
It is known that (holomorphic and non-holomorphic) Eichler integrals furnish examples of quantum modular forms. We define the following two vector-valued functions
\begin{equation}\label{dfn:Eich_period}
 f^\ast(\tau) := \int_{-\bar{\tau}}^{{\rm i}\infty} {\rm d}w
\frac{\overline{f(-\bar w)}}{(-{\rm i}(w+\tau))^{2-k}}, \qquad
r_{f,\frac{d}{c}} (x) := \int_{\frac{d}{c}}^{{\rm i}\infty} {\rm d}w
\frac{\overline{f(-\bar w)}}{(-{\rm i}(w+x))^{2-k}},
\end{equation}
for $f$ a vector-valued cusp form with multiplier $\chi$, and $\frac{d}{c}\in {\mathbb Q}$. We say $f^\ast$ is the {\it non-holomorphic Eichler integral} of the cusp form $f$. It is easy to verify that \smash{$r_{f,\frac{d}{c}}$} is a real analytic function on~${\mathbb R}$, which captures the error of modularity of $f^\ast$:
\[
 ( f^\ast- f^\ast\lvert_{2-k,\bar\chi}\gamma)(\tau) = r_{f,\frac{d}{c}} (\tau)
\]
for $\gamma \in \Gamma$, where $\bar\chi$ is the conjugate multiplier $\bar\chi(\gamma)=\overline{\chi(\gamma)}$. As a result, we have $f^\ast\in Q_{2-k}(\Gamma, \bar \chi)$.

Similarly, for $f_i\colon \BH \to \CC^{n_i}$, $i=1,2$ a pair of vector-valued cusp forms (or modular form if the weight is $1/2$) with weight $k_i$ and multiplier system $\chi_i$, we define the following matrix-valued (valued in $\CC^{n_1\times n_2}$) functions:
\be \label{dfn:double_eichler}
(f_1,f_2)^\ast (\tau) := \int_{-\bar{\tau}}^{{\rm i} \infty} {\rm d}w_1
\int_{w_1}^{{\rm i}\infty} {\rm d}w_2 \,
\frac{\overline{f_1(-\bar w_1)} \overline{f_2(-\bar w_2)}}{(-{\rm i}(w_1+\tau))^{2-k_1}(-{\rm i}(w_2+\tau))^{2-k_2}}
\ee
and
\[
r_{f_1,f_2,\frac{d}{c}} (x) := \int_{\frac{d}{c}}^{{\rm i} \infty} {\rm d}w_1
\int_{w_1}^{\frac{d}{c}} {\rm d}w_2 \,
\frac{\overline{f_1(-\bar w_1)} \overline{f_2(-\bar w_2)}}{(-{\rm i}(w_1+x))^{2-k_1}(-{\rm i}(w_2+x))^{2-k_2}} .
\]
The function $(f_1,f_2)^\ast$ is often referred to as a non-holomorphic double Eichler integral, or iterated non-holomorphic Eichler integral more generally.

One can show that for $\gamma \in\Gamma$,
\[
((f_1,f_2)^\ast - (f_1,f_2)^\ast\lvert_{4-k_1-k_2,\bar\chi_1,\bar\chi_2}\gamma)(\tau) =
r_{f_1,f_2,\frac{d}{c}} (\tau) + I_{f_1}(\tau) r_{f_2,\frac{d}{c}}(\tau),
\]
where the slash operator acts in the following way in terms of the components. Write $I_{i,j}:=(f_{1,i},f_{2,j})^\ast$ to denote the non-holomorphic double Eichler integral of the components of the vector-valued modular forms $f_1$ and $f_2$. Then
\begin{equation*}
 (I\lvert_{k,\bar\chi_1,\bar\chi_2}\gamma)_{i,j}(\tau)
 := (c\tau+d)^{-k}\sum_{i'=1}^{n_1}\sum_{j'=1}^{n_2}
 I_{i',j'}(\gamma\tau) \overline{(\chi_1(\gamma))_{i',i}}\overline{(\chi_2(\gamma))_{j',j}}.
\end{equation*}
We have \smash{$r_{f_1,f_2,\frac{d}{c}} (\tau) \in\mathcal{O} (\BR\backslash \{ -\frac{d}{c} \})$}, and \smash{$r_{f_1,f_2,\frac{d}{c}} (\tau) \in\mathcal{O}(\BR)$} if both $f_i$ are cusp forms \cite{bringmann2018higher}. From the above, we see that $({f_1,f_2})^\ast$ is a vector-valued depth two quantum modular form valued in~$\CC^{n_1\times n_2}$ with multiplier $\bar \chi$ given by
\begin{equation*}
 (\bar \chi(\gamma))_{(i,j),(i',j')}= \overline{(\chi_1(\gamma))_{i',i}}\overline{(\chi_2(\gamma))_{j',j}}.
\end{equation*}
In this paper, we will mainly encounter modular forms with real coefficients, satisfying
\[
\overline{f(-\bar \tau)} = f(\tau),
\]
and we will often use this property to simply write $f(\tau)$ in the integrand.

The above-mentioned quantum modular property of the double Eichler-integral $(f_1,f_2)^\ast$, together with the form of the companion of $F^{(\rho)}$ as given in Proposition \ref{prop:E01asymp}, its rewriting up to one-dimensional pieces in Lemma \ref{lem:Estar_vs_starless}, and the relation to double Eichler integrals shown in Proposition \ref{prop:Estarless_eichler}, leads to the following result.
\begin{thm}
 The generalized $A_2$ false theta functions defined in \eqref{dfn:gen_false} is, up to an overall rational power of $q$ and possibly the addition of a finite polynomial in $q$ and $q^{-1}$, a sum of depth two quantum modular forms.
\end{thm}

\section[Properties of Z\^{}\{SU(3)\}]{Properties of $\boldsymbol{\widehat{Z}^{{\rm SU}(3)}}$}\label{sec:Zhat}

In this section we turn to the main object of our study: \smash{$\widehat Z^{G}_\bv(M_3)$} for $G={\rm SU}(3)$ and the simplest interesting choice of $M_3$, namely negative Seifert manifolds with three exceptional fibers. In Section~\ref{subsec:topology}, we explain how they are assembled using the generalized $A_2$ false theta functions~$F^{(\rho)}$ as building blocks. Combining with the results of the quantum modularity of the latter as established in the previous section, we are led to Theorems \ref{thm1_intro} and \ref{thm:recursive}. While we do not have a proof for Conjecture \ref{conj_rec}, we provide evidence for it through studying numerous examples in Section~\ref{sec:examples}.

\subsection{Topology}
\label{subsec:topology}

A plumbed three-manifold $M_3$ can be defined
as the boundary of glued disk bundles
associated to its plumbing graph, which is a weighted graph $(V,E,a)$, here taken to be a planar tree-shaped graph and no loops. The weights $a(v)$ give the Euler number of the disk bundle corresponding to the vertex $v\in V$.
Gluing occurs when there is an edge connecting the two vertices $v$ and~$v' $.
 The data of the weighted graph $(V,E,a)$ is equivalent to that of the adjacency or plumbing matrix $M$ of the graph $(V,E,a)$, with entries
\begin{equation}\label{adjacencymatrix}
 M_{v,v' }=\begin{cases}
 a(v) & \mathrm{if}\quad v=v', \\
 1 & {\mathrm{if}}\quad v\neq v', v {\rm ~and~} v' {\rm ~are~connected}, \\
 0 & \mathrm{otherwise} .
 \end{cases}
\end{equation}

Seifert manifolds are examples of such plumbed three-manifolds. A Seifert manifold
\begin{equation*}
 M_3=\left(b;\left\{\frac{q_i}{p_i}\right\}_{i=1}^{n}\right)
\end{equation*}
with Seifert invariants $(q_1,p_1), \ldots, (q_n, p_n)$ is specified by a star-shaped plumbing graph with a~unique junction vertex $v_0$ from which emanate $n$ legs, which represent the exceptional fibers of~$M_3$. As mentioned before, we say that $M_3$ is a negative/positive Seifert manifold depending on the sign of $M^{-1}_{v_0,v_0}$. Along the $i$-th leg, the vertices \smash{$v^{(i)}_k$} and the corresponding weights \smash{$a^{(i)}_k$}
are given by the continued fraction expansion
\begin{equation}
 \label{eqn:ctndfrac}
 \frac{q_i}{p_i}=-\cfrac{1}{a_{1}^{(i)}-\cfrac{1}{a^{(i)}_2-
 \cfrac{1}{a^{(i)}_3-\cdots}}},
\end{equation}
while $a(v_0)=b$ and the orbifold Euler number ${\rm e} \in \mathbb{Q}$ is given by
\begin{equation*}
 b={\rm e}-\sum\limits_{i=1}^n \frac{q_i}{p_i} .
\end{equation*}
See \cite[Appendix A]{chengUpC} for further useful relations between the Seifert data and the plumbing graph.

For such manifolds, define $D$ to be the smallest positive integer such that $DM_{v_0,v}^{-1}\in {\mathbb Z}$ for all~$v\in V$, and let $m=D^{2} \big\lvert M_{v_0,v_0}^{-1}\big\lvert$.\footnote{Comparison of conventions: $m$ in this paper is what is written as $mD$ in \cite{chengUpC}.} Examples are given by Brieskorn spheres $M_3 = \Sigma(p_1,p_2,p_3)$, which have trivial integral homology and are determined by three coprime integers $p_1$, $p_2$, $p_3$ through the defining equation
\begin{equation}\label{dfn:brieskorn}
 \Sigma(p_1,p_2,p_3)=
\bigl\{(x,y,z)\in\mathbb{C}^3 \mid x^{p_1}+y^{p_2}+z^{p_3} =0 \bigr\} \cap S^5 .
\end{equation}
The Seifert data that specify the plumbing diagram are related to the integers $\{ p_1,p_2,p_3 \}$ by the following relation
\begin{equation*}
 b+\sum_{i=1}^3 \frac{q_i}{p_i} =- \frac{1}{p_1p_2p_3}.
\end{equation*}
For Brieskorn spheres, which satisfy $\lvert \det (M)\lvert =1$, we have $D=1$ and $m= \big\lvert M_{v_0,v_0}^{-1}\big\lvert$.

For a weakly-negative plumbed three-manifold $M_3$, we define the $\widehat Z^G$-invariants for any ADE gauge group $G$ in the following way, where we mostly adopt the same notation of as in \cite{chengUpC}. In particular, for a given simply-laced Lie group $G$ and a plumbing graph $(V,E,a)$, we let $\Lambda$ be the root lattice and
\begin{equation*}
 \Gamma_{M,G}:=M\BZ^{\otimes |Z|}\otimes_{\BZ} \Lambda .
\end{equation*}
For $\vec{\underline{x}}\in\BR^{\otimes |Z|}\otimes_{\BZ} \Lambda $, we define its norm to be given by the inverse plumbing matrix in the direction along the vertices and by the Cartan matrix in the root lattice directions:
\begin{equation*}
|| \underline{\vec{x}} ||^2 := \sum_{v,v'\in V} M^{-1}_{v,v'} \langle \vec{x}_v, \vec{x}_{v'} \rangle.
\end{equation*}

\begin{defn}[higher rank $\widehat Z$ invariants \cite{chengUpC, Chung:2018rea,Park:2019xey}]\label{dfn:blocks}
 Let $G$ be a simply-laced Lie group and~$M_3$ a weakly negative plumbed three-manifold with plumbing matrix $M$. Let $\vec{\underline{b}}$ be a generalized~$\text{Spin}^c$ structure on the manifold, given by
 \begin{equation}
\vec{\underline{b}}\in \bigl(\BZ^{|V|} \otimes_\BZ \Lambda +\vec{\underline{b}}_0\bigr) /
\Gamma_{M,G},
\label{eq:bfromcoker}
\end{equation}
where $ \vec{b}_{0,v}=\mathrm{deg}(v)\vec\rho$.

We define
 \begin{align}
\widehat{Z}^G_{\underline{\vec{b}}}(M_3;\tau) :={}&C^G(q) \int_{\cal C} d \underline{\vec{\xi}}\,\left( \prod_{v \in V} \Delta\bigl(\vec{\xi}_v\bigr)^{2 - \deg v} \right)\nonumber\\
&\times
 \sum_{w \in {W}} \sum_{\underline{\vec{\ell}} \in \Gamma_{M,G} +w(\underline{\vec{b}}) } q^{-\frac{1}{2}|| \underline{\vec{\ell}}||^2}
 \Bigg( \prod_{v' \in V} \ex^{\langle \vec{\ell}_{v'}, \vec{\xi}_{v'} \rangle} \Bigg),\label{def:ZhatHR-3mfdvoa}
 \end{align}
 where $W$ denotes the Weyl group of the root lattice of $G$, $w\bigl(\underline{\vec{b}}\bigr) $ denotes the diagonal action \smash{$w\bigl(\underline{\vec{b}}\bigr) =\bigl(w\bigl(\vec b_v\bigr), w\bigl(\vec b_{v'}\bigr),\dots\bigr)$} and the integration measure is given by
 \[
 \int_{\cal C} d \underline{\vec{\xi}} :={\rm p.v.} \int\prod_{v\in V} \prod_{i =1}^{\text{rank}G} {\frac{{\rm d}z_{i,v}}{2\pi {\rm i} z_{i,v}}},
\]
 with the contour ${\cal C}$ given by the Cauchy principal value integral around the unique circle in the $z_{i,v}$-plane.
 Letting $\pi_M$ be the number of positive eigenvalues of $M$ and $\sigma_M$ the signature of $M$, according to \cite{Park:2019xey},
 \[
 C^G(q)= (-1)^{|\Phi_+|\pi_M}q^{{\frac{3\sigma_M-{\rm Tr} M}{2}} |\vec \rho|^2},
 \]
 where $ \Phi_+$ is a set of positive roots for $G$ and $\vec \rho$ is a Weyl vector for $G$. Lastly, $\Delta$ is the Weyl determinant.
\end{defn}

As shown in \cite{chengUpC}, the $\widehat Z$-invariants for negative Seifert manifolds with three exceptional fibers and for $G={\rm SU}(3)$ can be expressed as combinations of the generalized $A_2$ false theta functions~\eqref{dfn:gen_false} in the following way.
Given $\vec{\underline{b}}$ and a choice of $\hat w= (w_1,w_2,w_3)\in W^{\otimes 3}$, one of the following two statements is true. Either there does not exist any root vector $\vec{\ell_0}$ such that
\begin{equation}\label{eqn:condition_k}
 \vec{\underline{b}}- \bigl(\vec{\ell}_0,w_1(\vec{\rho}),w_2(\vec{\rho}),w_3(\vec{\rho})\bigr)\in M\BZ^{\left|V\right|}\otimes_{\BZ}\Lambda
\end{equation}
or there exists a unique $\vec k_{\hat w} \in \Lambda/D\Lambda$ such that such that (\ref{eqn:condition_k}) holds if and only if $\vec k_{\hat w}= \vec{\ell_0}/D\Lambda$.

Now, let ${\cal W}_{\vec{\underline{b}}} \subseteq W^{\otimes 3}$ be the subset consisting of all $\hat w$ for which the latter is true. For $\hat w\in {\cal W}_{\vec{\underline{b}}}$, let
\begin{equation*}
\vec{s}_{\hat{w}} = D\sum\limits_{v_i \in \{v_1,v_2,v_3\}}
M^{-1}_{v_0,v_i}w_i(\vec{\rho}) .
\end{equation*}
The above defines $\vec{\sigma}_{\hat{w}}\in\Lambda/m\Lambda$ via
\[\vec{s}_{\hat{w}}= \vec{\sigma}_{\hat{w}}+\frac{m}{D} \vec{k}_{\hat{w}}. \]
The $\widehat Z$-invariant is then given by
\begin{equation}\label{s3:hombl-2}
\widehat{Z}^{{\rm SU}(3)}_{\vec{\underline{b}}}(M_3;\tau) =
C_{}(q)
\sum_{\hat{w}\in {\cal W}_{\vec{\underline{b}}} }
(-1)^{\ell (\hat{w})} F^{(\rho_{\hat w})} (\tau),
\end{equation}
where ${\ell (\hat{w})} := \sum_{i=1}^3 {\ell ({w_i})}$ is the total Weyl length, $\rho_{\hat w}=\bigl(\vec{\sigma}_{\hat{w}},\vec{k}_{\hat{w}},m,D\bigr)$ specifies the functions \smash{$F^{(\rho_{\hat w})} (\tau)$} from equation \eqref{dfn:gen_false}, and
\[ 
C_{}(q) = (-1)^{\pi_M} q^{3\sigma_M-{\rm Tr}M+\delta_M},\qquad
 \delta_{M}= \sum_{v\in V_{1}}\left( \frac{\left( M^{-1}_{v_{0},v} \right)^{2}}{M^{-1}_{v_0,v_0}}-M^{-1}_{v,v} \right),
\]
with $\pi_M$ and $\sigma_M$ denoting the number of positive eigenvalues resp.\ the signature of the adjacency matrix $M$. The additional power $q^{\delta_M}$ comes from performing the integral \eqref{def:ZhatHR-3mfdvoa} along the directions corresponding to the ``non-junction'' vertices with $v$ with degree less than three.

\subsection{Companions}
\label{subsec:companionsZ}
In this subsection we will put the results obtained so far together and derive the form of the companion function for \smash{$\widehat Z^{{\rm SU}(3)}_{\bv}(M_3)$} for negative Seifert $M_3$ with three exceptional fibers, before we further specialize to the case of Brieskorn spheres.

Combining 1) (\ref{s3:hombl-2}),
the expression of \smash{$\widehat Z^{{\rm SU}(3)}_{\bv}(M_3)$} in terms of the generalized $A_2$ false theta function $F^{(\varrho)}$, 2) Lemma \ref{lem:shifting} and \eqref{Section2-Fsplit},
 the splitting of $F^{(\varrho)}$ into components, 3) Proposition~\ref{prop:E01asymp},
 the companion of the components, and 4) Proposition \ref{prop:Estarless_eichler} and Lemma \ref{lem:Estar_vs_starless}, the iterated non-holomorphic Eichler integral expressions for the companions, we finally obtain the following.
\begin{prop}
\label{lem:rewritingBBE}
For a negative Seifert manifold $M_3$ with three exceptional fibers, a companion function \smash{$\widecheck{Z}^{{\rm SU}(3)}_{\bv}(M_3)$} of the rank two homological blocks \smash{$\widehat Z^{{\rm SU}(3)}_{\bv}(M_3)$} is, up to potential one-dimensional pieces, given by the following non-holomorphic double Eichler integral
\begin{align*}
\Zchlong{{\rm SU}(3)}{\bv}{M_3}
 ={}&z_{1\rm d}+ \frac{D}{m}\,C(q^{-1}) \sum\limits_{\hat{w}\in {\cal W}_\bv}(-1)^{\ell(\hat{w})} \sum_{\nu=0,1} \frac{\sqrt{3}}{4 \pi^{\nu}}	\left(\frac {3\Delta \vec{s}_{\hat w}}{m}\right)^{1-\nu}\\
 &\times\sum_{w\in W^+}\sum_{\delta\in\BZ/2} \bigl(\vartheta'_{w,\hat w,\delta},\vartheta^{1-\nu}_{w,\hat w,\delta}\bigr)^\ast(\tau),
\end{align*}
where the non-holomorphic double Eichler integral is of the theta functions
\begin{gather*}
 \vartheta'_{w,\hat w,\delta}(\tau)=
 \theta^1_{m, m\delta+\langle \vec\rho, w(\vec\sigma_{\hat w})\rangle}(\tau),\qquad
 \vartheta^{1-\nu}_{w,\hat w,\delta}(\tau) = \theta^{1-\nu}_{m,m\delta + \langle \Delta\vec\omega, w(\vec{\sigma}_{\hat w})\rangle}(3\tau).
\end{gather*}
\end{prop}

Note that the above, together with the quantum modular properties of the non-holomorphic double Eichler integrals discussed in Section~\ref{subsec:QM}, leads immediately to Theorem \ref{thm1_intro}.

\subsubsection*{Weil representations}
From the fact that $\theta^\nu_m=(\theta^\nu_{m,r})$ is a vector-valued modular form for $\nu=0,1$, we see from the discussion in Section~\ref{subsec:QM} that, potentially up to certain one-dimensional pieces, \smash{$\widecheck Z^{{\rm SU}(3)}_\bv(M_3)$} is a linear combination of components of vector-valued quantum modular forms of depth two. In what follows, we will investigate the recursive structure relating the quantum modular properties of \smash{$\widehat Z^{{\rm SU}(2)}_\bv(M_3)$} and \smash{$\widehat Z^{{\rm SU}(3)}_\bv(M_3)$}, or equivalently \smash{$\widecheck Z^{{\rm SU}(2)}_\bv(M_3)$} and \smash{$\widecheck Z^{{\rm SU}(3)}_\bv(M_3)$}. In order to do that, we need to take a closer look at the underlying representations of the metaplectic group \smash{$\widetilde{{\rm SL}_2}(\BZ)$}. For this purpose, we will introduce specific Weil representations specified by a positive integer~$m$ and a subgroup $K$ of the group of its exact divisors {Ex$_m$}, as mentioned in Section~\ref{sec:intro_sum}.

To such a group $K$ we associate a subrepresentation of $\Theta_m$ \eqref{dfn:theta}, which we write as $\Theta^{m+K}$, in the following way. First we make use of the fact that the space of matrices commuting with the $S$- and $T$-matrices of $\Theta_m$ is spanned by \cite{Cappelli:1987xt}
\begin{equation*}
 \Omega_m(n)_{r,r'} = \begin{cases}
1 &\mathrm{if}~r\equiv -r'\bmod\, 2n~\mathrm{and}~r\equiv r' \bmod\, 2m/n ,\\
0 & \mathrm{otherwise}, \ r,r'\in\BZ/2m
 \end{cases}
\end{equation*}
for $n|m$. Note that $\Omega_m(n)$ and $\Omega_m(n')$ commute for every pair of divisors $n$ and $n'$. For instance,~$\Omega(1)= \mathbf{1}_m$ is the identity matrix of size $2m\times 2m$.

Now define the corresponding projection operators
\begin{equation}\label{dfn:proj}
 P^\pm_m(n) :=\left( \mathbf{1}_m \pm \Omega_m(n) \right) /2, \qquad n\in{\rm Ex}_m,
\end{equation}
 satisfying $\bigl(P^\pm_m(n) \bigr)^2= P^\pm_m(n) $.

Since in our application we are mostly interested in Eichler integrals involving \[\theta^1_{m,r}(\tau) = \frac{1}{2\pi {\rm i}}\frac{\partial}{\partial z} \theta_{m,r}(\tau,z)\lvert_{z=0},\] which has the property $\theta^1_{m,r} = -\theta^1_{m,-r} $, or $P^-_m(m)\theta^1_{m}=\theta^1_{m}$, we will from now on focus on the subgroups $K$ satisfying $m\not\in K$ and define the projector
\begin{equation}\label{dfn:projB}
P^{m+K} = \left(\prod_{n\in K}P^+_m(n)\right)
P^-_m(m),
\end{equation}
using the notation of \cite{cheng20193d}. When $K$ is maximal, in the sense that $\mathrm{Ex}_m = K\cup (m\ast K)$, $\Theta^{m+K}:=P^{m+K} \Theta_m$ furnishes an irreducible representation of \smash{$\widetilde{{\rm SL}_2}(\BZ)$} when $m$ is square-free. In general, $\Theta^{m+K,{\rm irred}}:=P^{m+K,{\rm irred}} \Theta_m$ with
$K$ maximal and
\begin{equation}\label{irred_weil}
P^{m+K, {\rm irred}} := \left(\prod_{n\in K}P^+_m(n)\right) \left(
\prod_{f^2|m}\left(\mathbf{1}_m-\frac{1}{f} \Omega_m(f) \right)
\right)
P^-_m(m)
\end{equation}
is irreducible \cite{SkoruppaThesis, 2007arXiv0707.0718S}.

Using the above, we introduce the notation
\begin{equation}\label{WeylorbitSU2}
 \theta_r^{1,m+K}: =\sum_{r'\in\BZ/2m} P^{m+K}_{r,r'}\theta^1_{m,r'},
\end{equation}
which will be used extensively below.

In what follows, we will focus on the manifolds $M_3$ that are homological spheres, to obtain Theorem \ref{thm:intro_brieskorn_recursive}. First, we simplify the expression for the companions of \smash{$\widehat Z^{{\rm SU}(3)}_\bv(M_3)$} given in Proposition \ref{lem:rewritingBBE} in these cases.

\begin{lem}\label{lem:brieskorn_thetas}
 For Brieskorn spheres $\Sigma(p_1,p_2,p_3)$, we have
\begin{align*}
\Zchlong{{\rm SU}(3)}{\bv}{M_3}
 ={}&z_{1\rm d}+ \tfrac{|W|}{2m}C\bigl(q^{-1}\bigr) \sum\limits_{\hat{w}\in W^{\otimes 3}}(-1)^{\ell(\hat{w})} \\
 &\times\sum_{\nu=0,1} \frac{\sqrt{3}}{4 \pi^{\nu}}	\left(\frac {3\Delta\vec{s}_{\hat w}}{m}\right)^{1-\nu}\sum_{\delta\in\BZ/2} \bigl(\vartheta'_{\hat w,\delta},\vartheta^{1-\nu}_{\hat w,\delta}\bigr)^\ast(\tau)
,
\end{align*}
where the non-holomorphic double Eichler integral is of the theta functions
\begin{gather*}
 \vartheta'_{\hat w,\delta}(\tau)=
 \theta^1_{m, m\delta+\langle \vec\rho, \vec\sigma_{\hat w}\rangle}(\tau),\qquad
 \vartheta^{1-\nu}_{\hat w,\delta}(\tau) = \theta^{1-\nu}_{m,m\delta + \langle \Delta\vec\omega, \vec{\sigma}_{\hat w}\rangle}(3\tau),
 \end{gather*}
for $m=p_1p_2p_3$, $\bar p_i = m/p_i$ and
\begin{align*}
\vec\sigma_{\hat w} &= \vec{s}_{\hat{w}} = -\sum\limits_{i=1}^3 \bar{p}_i w_i(\vec{\rho}) .
\end{align*}
\end{lem}
The proof of this lemma can be found in Appendix \ref{pf_lem:rewritingBBE}.\footnote{Regarding the one-dimensional non-holomorphic Eichler integral $z_{1d}$,
we also comment that, when all $\vec\sigma_{\hat w}$ satisfy $0\leq \langle \vec\sigma_{\hat w}, \vec \omega_i\rangle \leq m$, the different contributions from $X_\nu$
in Lemma~\ref{lem:Estar_vs_starless}
to $\hat Z^{{\rm SU}(3)}(M_3)$ cancel for $M_3=\Sigma(p_1,p_2,p_3)$.}

It is know that the ${\rm SU}(2)$ companion for Brieskorn spheres with three exceptional fibers is given by~\cite{cheng20193d}
\begin{equation*}
 \Zchlong{{\rm SU}(2)}{\bv}{M_{3}}\stackrel{\dots}{=}\bigl(\theta^{1,m+K}_r\bigr)^\ast
\end{equation*}
up to an overall rational power of $q$ (and the addition of a finite polynomial in $q^{-1}$ for the case~$M_3=\Sigma(2,3,5)$), where
\begin{equation}\label{brieskorn_K}
 m=p_1p_2p_3,\qquad K=\{1,\bar p_1,\bar p_2,\bar p_3\},
\end{equation}
and $r=m-\bar p_1-\bar p_2-\bar p_3$. For the ${\rm SU}(3)$ companions, we have the following non-holomorphic double Eichler integral.
\begin{prop}\label{prop:sphere_recursive}
For Brieskorn spheres $\Sigma(p_1,p_2,p_3)$, using the same notation as in Lemma~{\rm \ref{lem:brieskorn_thetas}} and in \eqref{brieskorn_K}, we have
 \begin{gather*}
 \Zchlong{{\rm SU}(3)}{\bv}{M_{3}}
 =z_{1d}+ \frac{3\sqrt{3}}{2m}C\bigl(q^{-1}\bigr) \sum_{\nu=0,1} \pi^{-\nu}	\sum_{\delta\in\BZ/2} \sum_{r\in {\cal R}} \left(\frac{r}{m}\right)^{1-\nu}\bigl(\vartheta'_{r,\delta},\vartheta^{1-\nu}_{r,\delta}\bigr)^\ast(\tau)
,
\end{gather*}
where the non-holomorphic double Eichler integral is of the theta functions
 \begin{gather*}
 \vartheta'_{r,\delta}(\tau)=
 4\theta^{1,m+K}_{m\delta+\sum_i \bar p_i c_i^{(r)}}(\tau),\qquad
 \vartheta^{1-\nu}_{r,\delta}(\tau) = \theta^{1-\nu}_{m,m\delta + r}(3\tau).
 \end{gather*}

In the above, ${\cal R} \subset \ZZ/2m$ is given by
\[
{\cal R} = {\cal R}_0\cup {\cal R}_1\cup {\cal R}_2\cup {\cal R}_3
\]
and
\begin{gather*}
 {\cal R}_0 = \{0\},\qquad
 {\cal R}_1 = {\bf P}^+ \{\bar p_1\},\qquad
 {\cal R}_2 = {\bf P}^+ \{\bar p_1+\bar p_2, \bar p_1-\bar p_2\},\\
 {\cal R}_3 = {\bf P}^+ \{\bar p_1+\bar p_2-\bar p_3, -\bar p_1-\bar p_2+\bar p_3\},
\end{gather*}

where we denote by ${\bf P}^+$ by the group of even permutations of $(p_1,p_2,p_3)$. For each $r\in {\cal R}$, we set \smash{$c_i^{(r)}:= 2-|r_i|$} if $r=\sum_i r_i \bar p_i$.
\end{prop}

From the above, we see that \eqref{thm:recursive} holds, and in particular Theorem \ref{thm:intro_brieskorn_recursive} holds. That is, up to possible one-dimensional terms, the same \smash{$\widetilde{{\rm SL}_2(\BZ)}$} representation $\Theta^{m+K}$ governs not just~\smash{$\widehat Z^{{\rm SU}(2)}$} but also the ${\rm SU}(3)$ quantum modularity. Note, when all $p_i$s are square free, the underlying representation $\Theta^{m+K}$ is irreducible. Furthermore, when $2^2{\not\hspace{2.5pt}\mid}\, m$, one can replace $\theta^{1,m+K}_r$ in Proposition~\ref{prop:sphere_recursive} with the irreducible representation $\Theta^{m+K,{\rm irred}}$ (cf.\ (\ref{irred_weil})). The proof of the above proposition is given in Appendix~\ref{prf:sphere_recursive}.

\section{Examples}\label{sec:examples}
In this section we present in detail the structure of $\widehat{Z}$ invariants discussed in Section~\ref{sec:Zhat}. We further show the recursive structure, proven in Section~\ref{sec:Zhat} for homological spheres, is also present for other non-spherical negative Seifert manifolds with three exceptional fibers. In particular, we compute explicitly the underlying $\widetilde{SL}_2(\ZZ)$ Weil representations.

\subsection[Example: M(-1;1/4,3/5,1/7)]{Example: $\boldsymbol{ M\bigl(-1;\frac{1}{4},\frac{3}{5},\frac{1}{7}\bigr)}$}

We begin with the spherical Seifert manifold $X = M\left( -1,\frac{1}{4},\frac{3}{5},\frac{1}{7} \right)\cong\Sigma(4,5,7)$. To determine the plumbing matrix $M$ we compute continued fraction expansions of the Seifert data \eqref{eqn:ctndfrac}. From
\[
 \frac{3}{5}=\cfrac{-1}{-2-\cfrac{1}{-3}},
\]
we have
\begin{equation*}
 M=\left(
\begin{matrix}
 -1 & \hphantom{-}1 & \hphantom{-}0 & \hphantom{-}1 & \hphantom{-}1 \\
 \hphantom{-}1 & -4 & \hphantom{-}0 & \hphantom{-}0 & \hphantom{-}0 \\
 \hphantom{-}0 & \hphantom{-}0 & -3 & \hphantom{-}0 & \hphantom{-}1 \\
 \hphantom{-}1 & \hphantom{-}0 & \hphantom{-}0 & -7 & \hphantom{-}0 \\
 \hphantom{-}1 & \hphantom{-}0 & \hphantom{-}1 & \hphantom{-}0 & -2 \\
\end{matrix}
\right).
\end{equation*}
The corresponding plumbing graph has one junction vertex connecting to three legs. Since~$X$ is a homological sphere, the adjacency matrix $M$ is unimodular and consequently the only inequivalent generalized Spin$^c$ structure is
\[
\vec{\underline{b}}_{0}=\bigl( \wv,-\wv,-\wv,-\wv,\vec 0 \bigr).
\]
The unimodularity also leads to the parameters $D=1$ and $m=-M_{v_{0},v_{0}}^{-1}=140$.

Since $X$ is a spherical Seifert manifold the condition \eqref{eqn:condition_k} is always satisfied so \smash{$\mathcal{W}_{\vec{\underline{b}}}$} in equation~\eqref{s3:hombl-2} is equal to $W^{\otimes 3}$. Because
\[
 (-1)^{\ell(w\hat w)}F^{\varrho_{w\hat w}}(\tau)=(-1)^{\ell(\hat w)}F^{\varrho_{\hat w}}(\tau),
\]
where $w\hat{w}=(ww_1,ww_2,ww_3)$, we may simplify the sum over \smash{$\mathcal{W}_{\vec{\underline{b}}}$} in \eqref{s3:hombl-2} to a sum over representatives $\hat{w}$ in the conjugacy classes of $W^{\otimes 3}/W$
\[
\widehat{Z}^{{\rm SU}(3)}_{\vec{\underline{b}}}(M_3;\tau) =
\left|W\right|C_{}(q)
\sum_{\hat{w}\in W^{\otimes 3}/W }
(-1)^{\ell (\hat{w})}
F^{(\rho_{\hat w})} (\tau).
\]
For this manifold, we can choose the representatives $\hat{w}$ such that $\vec s_{\hat w} =(s_1,s_2)$ have components $s_{i} \in \{1,\dots , m\}$. These parameters and their associated total Weyl length $(-1)^{\ell({\hat{w}})}$ in \eqref{s3:hombl-2} are collected in Table~\ref{tab:s-brieskorn}.

\begin{table}\renewcommand{\arraystretch}{1.2}
 \centering
 \begin{minipage}{65mm}
 \begin{tabular}{ccc}
 \multicolumn{3}{c}{$\vec s_{\hat w} =( s_{1},s_{2} ),\ (-1)^{\ell(\hat{w})}=1$} \\\hline
 $(43,\,103)$ & $(103,\,43)$ & $(43,\,43)$ \\
 $(27,\,111)$ & $(27,\,51)$ & $(111,\,27)$ \\
 $(51,\,27)$ & $(43,\,19)$ & $(19,\,43)$ \\
 $(27,\,27)$ & $(13,\,118)$ & $(13,\,58)$ \\
 $(1,\,82)$ & $(61,\,22)$ & $(1,\,22)$ \\
 $(13,\,34)$ & $(118,\,13)$ & $(58,\,13)$ \\
 $(22,\,61)$ & $(82,\,1)$ & $(22,\,1)$ \\
 $(34,\,13)$ & $(27,\,6)$ & $(6,\,27)$ \\
 &$(13,\,13)$ &
 \end{tabular}
 \end{minipage}\qquad
 \begin{minipage}{55mm}
 \begin{tabular}{ccc}
 \multicolumn{3}{c}{$\vec s_{\hat w} =( s_{1},s_{2} ),\ (-1)^{\ell(\hat{w})}=-1$} \\\hline
 $(83,\,83)$ & $(83,\,23)$ & $(23,\,83)$ \\
 $(47,\,71)$ & $(71,\,47)$ & $(33,\,78)$ \\
 $(41,\,62)$ & $(41,\,2)$ & $(78,\,33)$ \\
 $(62,\,41)$ & $(2,\,41)$ \\
 \end{tabular}
 \vspace{2.75cm}
 \end{minipage}

\caption{$\vec s_{\hat w}$ and its parity $(-1)^{\ell(\hat{w})}$ for the 36 inequivalent representatives of $W^{\otimes 3}/W$ for $M_3=M\bigl(-1;\frac{1}{4},\frac{3}{5},\frac{1}{7}\bigr)$.}\label{tab:s-brieskorn}
\end{table}

Since $D=1$ and we can set $\vec{k}_{\hat{w}}=\vec{0}$, and therefore $\vec{\sigma}_{\hat w} = \vec{s}_{\hat w}$, whereby the $\widehat{Z}$ invariant in equation \eqref{s3:hombl-2} is
\[
 C(q)\sum_{\hat{w}\in\mathcal{W}_{\vec{\underline{b} }}}( -1 )^{\ell(\hat{w})}F^{(\varrho_{\hat{w}})}( \tau )= 6 q^{26} - 12 q^{37} - 12 q^{43} - 12 q^{49} + \mathcal{O}\bigl(q^{50}\bigr).
\]

For each $\vec{s}$, we can then compute the set $\tilde\CS$ \eqref{def:Scal2}. These values are collected in Table~\ref{tab:SstarBriesk}. As a~selected example, consider $( s_{1},s_{2}) = ( 83,83)$ which correspond to $\hat{w}=(aba,aba,aba)$, where~$a$,~$b$ are Weyl group elements given as in~\eqref{eq:abaction}. Using equation \eqref{eq:alphawidef0A}, we find the set $\tilde\CS$ contains
\begin{center}\renewcommand{\arraystretch}{1.2}
 \begin{tabular}{cc}
 $\bal^{( 1)}_{w}$ & $\bal^{( 2)}_{w}$ \\ \hline
 $\left(0,\,-\frac{83}{140}\right)$ & $\left(1,\,-\frac{83}{140}\right)$ \\
 $\left(\frac{83}{140},\,-\frac{83}{140}\right)$ & $\left(\frac{57}{140},\,\frac{83}{70}\right)$ \\
$\left(-\frac{83}{140},\,\frac{83}{70}\right)$ & $\left(\frac{223}{140},\,-\frac{83}{140}\right)$.
 \end{tabular}
\end{center}
The $\bal$ for all choices of $\vec{s}$ are collected in Table \ref{tab:SstarBriesk}.

For $\bal\in \tilde\CS$, let $\boldsymbol{\beta}$ be the unique vector satisfying $\bal \equiv \boldsymbol{\beta}~\bigl(\mathbb Z^2\bigr) $ and $\beta_1, \beta_2 \in \left[0,1\right)$. Lemma \ref{lem:shifting} justifies the splitting of the generalized $A_2$ false theta function into 1D and 2D contributions
\begin{alignat*}{3}
&F_{\nu}^{\text{1D}} (\tau ):=
\sum_{w\in\mathcal{W}_{\vec{\underline{b}}}} (-1 )^{\ell(\hat{w})}F_{\nu}^{ (\varrho_{\hat{w}} ),1D} (\tau ), \qquad&&
F_{\nu}^{(\varrho),\text{1D}} (\tau ):=\sum_{\bal\in\tilde\CS}\eta_\nu(\bal)\bigl(F_{\nu,\bal}^{(\varrho)} (\tau )-F_{\nu,{\boldsymbol\beta}}^{(\varrho)} (\tau )\bigr),&\\
& F_{\nu}^{\text{2D}} (\tau ):=
\sum_{w\in\mathcal{W}_{\vec{\underline{b}}}} (-1 )^{\ell(\hat{w})}F_{\nu}^{ (\varrho_{\hat{w}} ),\text{2D}} (\tau ), \qquad&& F_{\nu}^{(\varrho),\text{2D}} (\tau ):=\sum_{\bal\in\tilde\CS}\eta_\nu(\bal)F_{\nu,{\boldsymbol\beta}}^{(\varrho)} (\tau ).&
\end{alignat*}

For $X$ this splitting gives
\begin{gather*}
\tilde{F}= -\frac{9}{14} q^{4} -\frac{18}{35} q^{5} -\frac{33}{35} q^{7} -\frac{81}{70} q^{8} -\frac{57}{35} q^{10} -\frac{39}{35} q^{13} -\frac{81}{35} q^{14} \\
\phantom{\tilde{F}= }{}-\frac{261}{70} q^{16} -\frac{3}{35} q^{19} + \frac{123}{35} q^{22} + \frac{69}{35} q^{25}, \\
C(q)F_{0}^{\text{1D}} (m\tau ) = -\tilde{F}-\frac{99}{35} q^{26} - \frac{141}{35} q^{28} +\frac{18}{7} q^{37} +\frac{39}{35} q^{40} +\frac{81}{35} q^{41} + \mathcal{O}\bigl(q^{42} \bigr),\\
C(q)F_{0}^{\text{2D}} (m\tau ) = \tilde{F}+\frac{447}{70} q^{26} +\frac{141}{35} q^{28} - \frac{309}{35} q^{37} - \frac{39}{35} q^{40} - \frac{81}{35} q^{41} + \mathcal{O} \bigl(q^{42} \bigr),\\
C(q)F_{1}^{\text{1D}} (m\tau )= \tilde{F}+\frac{99}{35} q^{26} +\frac{141}{35} q^{28} - \frac{18}{7} q^{37} - \frac{39}{35} q^{40} - \frac{81}{35} q^{41} + \mathcal{O} \bigl(q^{42} \bigr),\\
C(q)F_{1}^{\text{2D}} (m\tau ) = -\tilde{F}-\frac{27}{70} q^{26} - \frac{141}{35} q^{28} - \frac{111}{35} q^{37} +\frac{39}{35} q^{40} +\frac{81}{35} q^{41} + \mathcal{O} \bigl(q^{42} \bigr) .
\end{gather*}
Here the one-dimensional contribution is
\[
 C(q) \bigl(F_{0}^{\text{1D}} (m\tau ) + F_{1}^{\text{1D}} (m\tau ) \bigr)=6 q^{109} {-6} q^{113} {-6} q^{121} {-6} q^{131} + 6 q^{157} + \mathcal{O} \bigl(q^{160} \bigr)
\]
and the total $\Zhshort{{\rm SU}(3)}{\vec{\underline{b}}}{X}$ has integral coefficients
\[
\Zhlong{{\rm SU}(3)}{\vec{\underline{b}}}{X} ={-6} q^{26} + 12 q^{37} + 12 q^{43} + 12 q^{49} + \mathcal{O}\bigl(q^{50}\bigr).
\]
The companion functions to the 2D contributions $F_{\nu}^{2\rm D}(m\tau)$ are double Eichler integrals, whose integrands computed using Lemma \ref{lem:rewritingBBE} contain
\begin{align}
\frac{1}{4}\sum_{\hat{w}\in W^{\otimes 3}}(-1)^{\ell(\hat{w})}\Theta^{(\varrho_{\hat{w}})}_{\nu,e}\left(\mathbf{z}\right)={}& {\left(63 \theta^{1}_{140,63} + 7 \theta^{1}_{140,7}\right)} \theta^{1,140+K}_{23}
+ {\left(15 \theta^{1}_{140,15} + 55 \theta^{1}_{140,55}\right)} \theta^{1,140+K}_{1} \nonumber\\
		&
- {\left( 7 \theta^{1}_{140,133} + 63 \theta^{1}_{140,77}\right)} \theta^{1,140+K}_{37}\nonumber\\
		&
	+ \left(8\theta^{1}_{140,132} + 48 \theta^{1}_{140,48} + 8\theta^{1}_{140,8} + 48 \theta^{1}_{140,92} \right) \theta^{1,140+K}_{118} \nonumber\\
	 &
+ \left(27 \theta^{1}_{140,113} + 13 \theta^{1}_{140,127} + 83 \theta^{1}_{140,57} + 43 \theta^{1}_{140,97} \right) \theta^{1,140+K}_{57}\nonumber\\
	&
+ {\left( 20 \theta^{1}_{140,120} + 20 \theta^{1}_{140,20}\right)} \theta^{1,140+K}_{6} \nonumber\\
	&
- \left( 28 \theta^{1}_{140,112} + 28 \theta^{1}_{140,28}\right) \theta^{1,140+K}_{2}\nonumber\\
	&
+ {\left(15 \theta^{1}_{140,125} + 55 \theta^{1}_{140,85}\right)} \theta^{1,140+K}_{29} \nonumber\\
	&
+ 35 \theta^{1}_{140,105} \theta^{1,140+K}_{9} - 35 \theta^{1}_{140,35} \theta^{1,140+K}_{19}\nonumber\\
	&
- \left(13 \theta^{1}_{140,13} + 27 \theta^{1}_{140,27} + 43 \theta^{1}_{140,43} + 83 \theta^{1}_{140,83}\right) \theta^{1,140+K}_{13}\!\!\!\label{eq:theta0_example1}
\end{align}
using the shorthand notation
\[
 \theta^{1}_{m,r} \theta^{1,140+K}_{r'}
\equiv
 \theta^{1}_{m,r} \left(3z_{2} \right)\theta^{1,140+K}_{r'}\left(z_{1} \right),
\]
and similarly
\begin{align}
\frac{1}{4}\sum_{\hat{w}\in W^{\otimes 3}}(-1)^{\ell(\hat{w})}\Theta^{(\varrho_{\hat{w}})}_{1,e}\left(\mathbf{z}\right)={}&{\left(\theta^{0}_{140,63} + \theta^{0}_{140,7}\right)} \theta^{1,140+K}_{23}
	+ {\left(\theta^{0}_{140,15} + \theta^{0}_{140,55}\right)} \theta^{1,140+K}_{1}\nonumber\\
		 &
	- {\left(\theta^{0}_{140,0} - \theta^{0}_{140,140}\right)} \theta^{1,140+K}_{26}
+ {\left(\theta^{0}_{140,133} + \theta^{0}_{140,77} \right)} \theta^{1,140+K}_{37} \nonumber\\
& 	- \left(\theta^{0}_{140,132} - \theta^{0}_{140,48} - \theta^{0}_{140,8} + \theta^{0}_{140,92} \right) \theta^{1,140+K}_{22}\nonumber\\
&
 - \theta^{0}_{140,105} \theta^{1,140+K}_{9} - \left(\theta^{0}_{140,113} + \theta^{0}_{140,127} + \theta^{0}_{140,57} \right.\nonumber\\
&
+ \left. \theta^{0}_{140,97} \right) \theta^{1,140+K}_{57}- \theta^{0}_{140,35} \theta^{1,140+K}_{19} \nonumber\\
&
- {\left(\theta^{0}_{140,120} - \theta^{0}_{140,20} \right)} \theta^{1,140+K}_{6}+ {\left(\theta^{0}_{140,112} - \theta^{0}_{140,2} \right)} \theta^{1,140+K}_{2} \nonumber\\
&
- {\left(\theta^{0}_{140,125} + \theta^{0}_{140,85} \right)} \theta^{1,140+K}_{29} \nonumber\\
&
- \left(\theta^{0}_{140,13} + \theta^{0}_{140,27} + \theta^{0}_{140,43} + \theta^{0}_{140,83} \right) \theta^{1,140+K}_{13},\label{eq:theta1_example1}
\end{align}
where $\theta^{0}_{m,r} \theta^{1,140+K}_{r'}
\equiv
\theta^{0}_{m,r} (3z_{2} )\theta^{1,140+K}_{r'}(z_{1} )$ and $K=\{ 1,p_{1}p_{2},p_{1}p_{3},p_{2}p_{3} \}=\{ 1,20,28,35 \}$. This example makes manifest the recursive structure described in Proposition~\ref{prop:sphere_recursive}.

\subsection[Example: M(-1;1/5,1/2,1/4)]{Example: $\boldsymbol{M\bigl(-1;\frac{1}{5},\frac{1}{2},\frac{1}{4}\bigr)}$}
Consider the Seifert manifold $X = M\left( -1;\frac{1}{5},\frac{1}{2},\frac{1}{4} \right)$, which has the plumbing matrix
\begin{equation*}
 M=\left(\begin{array}{@{}rrrr@{}}
-1 & 1 & 1 & 1 \\
1 & -5 & 0 & 0 \\
1 & 0 & -2 & 0 \\
1 & 0 & 0 & -4
\end{array}\right),
\end{equation*}
with $\mathrm{det}(M)=2$, $D=1$ and $m=20$. Since now $|\det M|>1$, $\mathcal{W}_{\vec{\underline{b}}}$ (cf.\ \eqref{s3:hombl-2}) will generically be a proper subset of \smash{$W^{\otimes 3}$}.

The manifold $X$ admits two inequivalent generalized Spin$^c$ structures: the trivial one $\vec{\underline{b}}_0$ and
\begin{equation*}
\vec{\underline{b}}_{1}= ( (1, -2 ), (1, 1 ), (-1, 2 ), (-3, 3 ) ).
\end{equation*}
 We focus on $\vec{\underline{b}}_{0}$ for simplicity. Using the same techniques as described above we can compute the set of $\vec{s}$, the $\widehat{Z}$ invariant and the companions \smash{$\mathbb{E}^{(\varrho)}_{\nu}(\tau)$}. The $\vec{s}$ and $\bal$ parameters are collected in Table~\ref{tab:balinrangePseudo}. The sum~\eqref{s3:hombl-2} is
\begin{gather*}
 \widehat Z^{{\rm SU}(3)}_\bv(M_3;\tau)= C(q)\sum_{\hat{w}\in\mathcal W_{\vec{\underline{b}}}}(-1)^{\ell(\hat{w})}F^{\left( \varrho_{\hat{w}} \right)}\left( q \right)= -6q + 12q^2 -12q^4 + 6q^5 + 6q^7 + \mathcal{O}\bigl(q^{10}\bigr),\\
 C(q)\bigl(F_{0}^{\text{1D}}(m\tau)+F_{1}^{\text{1D}}(m\tau)\bigr) ={-12} q^{16} + 12 q^{20} + 12 q^{22} {-12} q^{34} {-12} q^{38} + \mathcal{O}\left(q^{40}\right), \\
 C(q)F_{0}^{\text{2D}}(m\tau) = {-6} q + 12 q^{2} {-12} q^{4} + 6 q^{5} + 6 q^{7} + \mathcal{O}\bigl(q^{10}\bigr),
\\
 C(q)F_{1}^{\text{2D}}(m\tau) = \mathcal{O}\bigl( q^{500} \bigr),
\end{gather*}
while the double Eichler integral form of the companion has as integrand
\begin{align}
\frac{1}{80} \sum_{\hat{w}\in \mathcal{W}_{\vec{\underline{b}}}}\sum_{w\in W}(-1)^{\ell(\hat{w})}\Theta^{(\varrho_{\hat{w}})}_{0,w}(\mathbf z) ={}& \bigl(\theta^{1}_{20,1}+\theta^{1}_{20,9}\bigr)\theta_{11}^{1,20+4}-\theta^{1}_{20,5}\theta_{7}^{1,20+4}\nonumber\\
 &-\bigl(\theta^{1}_{20,11}+\theta^{1}_{20,19}\bigr)\theta_{1}^{1,20+4}+2\theta^{1}_{20,10}\bigl(\theta^{1,20+4}_{8}+\theta^{1,20+4}_{12}\bigr)\nonumber\\
 &-\theta^{1}_{20,15}\theta_{5}^{1,20+4}-2\bigl(\theta^{1}_{20,6}+\theta^{1}_{20,14}\bigr)\bigl(\theta^{1,20+4}_{4}+\theta^{1,20+4}_{16}\bigr)\!\!\!\label{eq:theta_example2}
\end{align}
and
\[\sum_{\hat{w}\in \mathcal{W}_{\vec{\underline{b}}}}\sum_{w\in W}\Theta^{(\varrho_{\hat{w}})}_{1,w}(\mathbf z)=0.\]
The vanishing of the companion $\BE_{1}^{2D}$ and the vanishing of~$F_{1}^{2D}$ to the order at which it is computed $\bigl(\mathcal{O}\bigl(q^{500}\bigr)\bigr)$ are necessary (but not sufficient) conditions for~$F_{1}^{2D}=0$.
In equation~\eqref{eq:theta_example2}, we used the same shorthand notation of equations \eqref{eq:theta0_example1} and~\eqref{eq:theta1_example1}. For both cases, the Weil representations are labelled by the set $K= \{ 1,4 \}$ which is the same as for the rank-one invariant \smash{$\widehat{Z}^{{\rm SU}(2)}_{\vec{\underline{b}}} (X;\tau)$} \cite{cheng20193d}. Furthermore, we find $\theta_{1}^{1,20+4}$ and $\theta_{11}^{1,20+4}$ in this sum which are found in the rank-one invariant too.

\subsection[Example: $M(-2;1/2,1/2,3/4)]{Example: $\boldsymbol{M\bigl(-2;\frac{1}{2},\frac{1}{2},\frac{3}{4}\bigr)}$}

Surprisingly for certain manifolds the sum of the two-dimensional contributions to the generalized~$A_2$ false theta function vanishes to the order at which we can compute it. One such example is the Seifert manifold $X=M\bigl(-2; \frac{1}{2},\frac{1}{2},\frac{3}{4} \bigr)$, with $\det(M)=4$, $D=1$, $m=4$ and $K=\{1\}$. Equation~\eqref{eq:alphawidef0A} leads to $\bal$ in Table~\ref{tab:balinrangeGeneral} and
\begin{gather*}
\begin{split}
& C(q)\sum_{\hat{w}\in\mathcal W_{\vec{\underline{b}}}} ( -1 )^{}F^{ ( \varrho_{\hat{w}} )} ( q )={-}36 q^{-1} + 36q^{1} {-}36 q^{10} + 36 q^{14} + \mathcal{O} \bigl(q^{20} \bigr),\\
& C(q)F_{0}^{\text{2D}} (m\tau )= C(q)F_{1}^{\text{2D}} (m\tau ) = \mathcal{O} \bigl( q^{500} \bigr),
\end{split}
\end{gather*}
for $\vec{\underline{b}}=\left(-\vec\rho,\vec\rho,\vec\rho,\vec\rho\right)$. Here the recursive structure is particularly simple to identify as
\begin{equation*}
\frac{1}{36}\widehat{Z}^{{\rm SU}(3)}_{\vec{\underline{b}}}(M_3;\tau)=q^{-4}+\frac{q^{21/2}}{2}\widehat{Z}^{{\rm SU}(2)}_{\vec{\underline{b}}^{\, {\rm SU}(2)}}(M_3;\tau) + \mathcal{O}\bigl(q^{500}\bigr),
\end{equation*}
where, $\vec{\underline{b}}^{ {\rm SU}(2)}=\left(-\vec\rho^{\ \prime},\vec\rho^{\ \prime},\vec\rho^{\ \prime},\vec\rho^{\ \prime}\right)$, where $\vec{\rho}^{\ \prime}$ is the Weyl vector of the $A_1$ root lattice.

\subsection{Further examples}

The above subsections have shown in great detail that the $\widehat{Z}$-invariant of Seifert manifolds, including non-spherical manifolds for which Conjecture \ref{conj_rec} is not yet proven, display a recursion relation across different ranks. In particular the companions for $\widehat Z^{{\rm SU}(2)}$ and $\widehat Z^{{\rm SU}(3)}$ are carefully analyzed.

In the following Table we provide further evidence of this phenomenon. We organize examples in blocks. In each block the data is organized as follows:
\begin{table}[ht]\renewcommand{\arraystretch}{1.2}
 \centering
 \begin{tabular}{|ll|}\hline
 Seifert data & $\sigma^{m+K}$ \\
 $m$, $D$ & $\sigma_{A_{1}}^{m+K}$ \\
 $m+K$ & $\sigma_{A_{2}}^{m+K}$ or $\bar \sigma_{A_{2}}^{m+K}$\\\hline
 \end{tabular}
\end{table}
 where $\sigma^{m+K}$ is the set of $r$ giving
inequivalent $\theta^{1,m+K}_r$ (\ref{WeylorbitSU2}), \smash{$\sigma_{A_1}^{m+K}$} is the minimal subset of $\sigma^{m+K}$ such that
(\ref{eqn:intro_Weyl1}) and (\ref{eqn:intro_Weyl2}) hold also when $\sigma^{m+K}$ is replaced by \smash{$\sigma_{A_1}^{m+K}$}, for all inequivalently choices of boundary conditions~$\bv$. Similarly,
 \smash{$\sigma_{A_2}^{m+K}$} is the minimal subset of $\sigma^{m+K}$ such that
(\ref{eqn:conj:recursive}) holds also when $\sigma^{m+K}$ is replaced by \smash{$\sigma_{A_1}^{m+K}$}, for all inequivalently choices of boundary conditions $\bv$. Note that we have $\sigma_{}^{m+K}\subset \sigma_{A_2}^{m+K}\subset \sigma_{A_1}^{m+K}$ in all cases we study.

\begin{table}[!t]\centering\renewcommand{\arraystretch}{1.22}
\begin{tabular}{|ll|}\hline
$(-2;1/2,\,1/2,\,3/5) $& $\sigma^{40} = \{1,-,39\}$ \\
$40,\ 4$& $\sigma_{A_{1}}^{40} = \{28,32,38\}$ \\
$40$ & $\bar\sigma_{A_{2}}^{40} = \{5,10,15,16,20,24,25,30,35\}$ \\
\hline
$(-1;1/2,\,1/3,\,1/8) $& $\sigma^{24 + 8} = \{1,2,4,5,7,8,10,13,16\}$ \\
$24,\ 1$& $\sigma_{A_{1}}^{24 + 8} = \{1,7\}$ \\
$24 + 8$ & $\sigma_{A_{2}}^{24 + 8} = \{1,2,7,10\}$ \\
\hline
$(-1;1/2,\,1/7,\,2/7) $& $\sigma^{14 + 7} = \{1,3,5,7\}$ \\
$14,\ 1$& $\sigma_{A_{1}}^{14 + 7} = \{3\}$ \\
$14 + 7$ & $\sigma_{A_{2}}^{14 + 7} = \{1,3,5\}$ \\
\hline
$(-1;1/4,\,1/7,\,4/7) $& $\sigma^{28 + 7} = \{1,2,3,5,6,7,9,10,13,14,17,21\}$ \\
$28,\ 1$& $\sigma_{A_{1}}^{28 + 7} = \{13,21\}$ \\
$28 + 7$ & $\sigma_{A_{2}}^{28 + 7} = \{1,2,3,5,6,7,9,10,13,14,17,21\}$ \\
\hline
$(-1;1/3,\,1/5,\,2/5) $& $\sigma^{15 + 5} = \{1,2,4,5,7,10\}$ \\
$15,\ 1$& $\sigma_{A_{1}}^{15 + 5} = \{4\}$ \\
$15 + 5$ & $\sigma_{A_{2}}^{15 + 5} = \{1,2,4,7\}$ \\
\hline
$(-1;1/3,\,1/3,\,1/4) $& $\sigma^{12 + 3} = \{1,2,3,5,6,9\}$ \\
$12,\ 1$& $\sigma_{A_{1}}^{12 + 3} = \{1,9\}$ \\
$12 + 3$ & $\bar\sigma_{A_{2}}^{12 + 3} = \{4,7,8,10,11\}$ \\
\hline
$(-2;1/2,\,1/2,\,12/13) $& $\sigma^{52} = \{1,-,51\}$ \\
$52,\ 2$& $\sigma_{A_{1}}^{52} = \{24,28,50\}$ \\
$52$ & $\sigma_{A_{2}}^{52} = \{2,4,9,11,15,17,22,24,28,30,35,37,41,43,48,50\}$ \\
\hline
$(-1;1/3,\,1/11,\,6/11) $& $\sigma^{33} = \{1,-,32\}$ \\
$33,\ 1$& $\sigma_{A_{1}}^{33} = \{16,22,28\}$ \\
$33$ & $\bar\sigma_{A_{2}}^{33} = \{3,4,6,7,9,12,15,18,21,24,26,27,29,30\}$ \\
\hline
$(-2;1/2,\,2/3,\,2/3) $& $\sigma^{6 + 3} = \{1,3\}$ \\
$6,\ 1$& $\sigma_{A_{1}}^{6 + 3} = \{1,3\}$ \\
$6 + 3$ & $\sigma_{A_{2}}^{6 + 3} = \{1,3\}$ \\
\hline
$(-2;1/2,\,1/2,\,8/9) $& $\sigma^{36} = \{1,-,35\}$ \\
$36,\ 2$& $\sigma_{A_{1}}^{36} = \{16,20,34\}$ \\
$36$ & $\sigma_{A_{2}}^{36} = \{2,4,5,7,11,13,14,16,20,22,23,25,29,31,32,34\}$ \\
\hline
$(-2;1/2,\,1/2,\,4/5) $& $\sigma^{20} = \{1,-,19\}$ \\
$20,\ 2$& $\sigma_{A_{1}}^{20} = \{8,12,18\}$ \\
$20$ & $\bar\sigma_{A_{2}}^{20} = \{5,10,15\}$ \\
\hline
$(-2;1/2,\,2/3,\,3/4) $& $\sigma^{12 + 4} = \{1,2,4,5,8\}$ \\
$12,\ 1$& $\sigma_{A_{1}}^{12 + 4} = \{1,5\}$ \\
$12 + 4$ & $\sigma_{A_{2}}^{12 + 4} = \{1,4,5,8\}$ \\
\hline
$(-1;1/2,\,1/3,\,1/9) $& $\sigma^{18 + 9} = \{1,3,5,7,9\}$ \\
$18,\ 1$& $\sigma_{A_{1}}^{18 + 9} = \{1,5\}$ \\
$18 + 9$ & $\sigma_{A_{2}}^{18 + 9} = \{1,5,7\}$ \\
\hline
\end{tabular}
\end{table}

\begin{table}[t]\centering\renewcommand{\arraystretch}{1.22}
\begin{tabular}{|ll|}\hline
$(-1;1/2,\,1/5,\,1/5) $& $\sigma^{10 + 5} = \{1,3,5\}$ \\
$10,\ 1$& $\sigma_{A_{1}}^{10 + 5} = \{1,5\}$ \\
$10 + 5$ & $\sigma_{A_{2}}^{10 + 5} = \{1,3,5\}$ \\
\hline
$(-1;1/2,\,2/5,\,1/15) $& $\sigma^{30 + 15} = \{1,3,5,7,9,11,13,15\}$ \\
$30,\ 1$& $\sigma_{A_{1}}^{30 + 15} = \{7,11\}$ \\
$30 + 15$ & $\sigma_{A_{2}}^{30 + 15} = \{1,5,7,11\}$ \\
\hline
$(-1;1/2,\,1/11,\,4/11) $& $\sigma^{22} = \{1,-,21\}$ \\
$22,\ 1$& $\sigma_{A_{1}}^{22} = \{7,11,15\}$ \\
$22$ & $\sigma_{A_{2}}^{22} = \{3,5,7,11,15,17,19\}$ \\
\hline
$(-2;1/2,\,1/2,\,6/7) $& $\sigma^{28} = \{1,-,27\}$ \\
$28,\ 2$& $\sigma_{A_{1}}^{28} = \{12,16,26\}$ \\
$28$ & $\bar\sigma_{A_{2}}^{28} = \{1,6,7,8,13,14,15,20,21,22,27\}$ \\
\hline
$(-1;1/2,\,1/4,\,1/5) $& $\sigma^{20 + 4} = \{1,2,3,4,6,7,8,11,12,16\}$ \\
$20,\ 1$& $\sigma_{A_{1}}^{20 + 4} = \{1,11\}$ \\
$20 + 4$ & $\sigma_{A_{2}}^{20 + 4} = \{1,3,4,7,8,11,12,16\}$ \\
\hline
$(-2;1/2,\,1/3,\,1/2) $& $\sigma^{24} = \{1,-,23\}$ \\
$24,\ 4$& $\sigma_{A_{1}}^{24} = \{16,20,22\}$ \\
$24$ & $\bar\sigma_{A_{2}}^{24} = \{3,6,9,12,15,18,21\}$ \\
\hline
\end{tabular}
\end{table}

\section{Discussions}\label{sec:discussions}

In the paper we continue the study of quantum modular properties of $\widehat Z^G$-invariants, extending the analysis to higher rank $G$. The results and conjectures of this paper lead to many further research questions and open questions, some of them we will hopefully report on in the future, which we list below.
\begin{itemize}\itemsep=0pt
 \item Conjecture \ref{conj_qmf} is plausible. Starting from the Definition \ref{dfn:blocks} of the rank-$r$ $\widehat Z$-invariants, after straightforwardly performing the contour integration in the directions spanned by all non-junction vertices, we are left with a rank $N=r\times n$ lattice sum in the integrand of the remaining contour integral. In the weakly-negative/positive case, the signature of the lattice is purely positive/negative. In particular, in the weakly-negative case we obtain a~sum over (derivatives of) rank $N$ false-theta-like function. It should be interesting to prove their quantum modularity explicitly. Similarly, for the weakly-positive case we expect to obtain a close cousin of higher depth mock modular form, though at present we do not have a universal recipe for defining $\widehat Z$-invariants for these cases.
 \item Beyond Conjecture~\ref{conj_qmf}, it would be very interesting to analyse quantum modularity of $\widehat Z$-invariants when the plumbed manifold is neither weakly-negative now weakly-positive, in other words when the space spanned by junction vertices has signature $(k,N-k)$ when~$k\neq 0$. For this purpose, it should be interesting to generalize the generalized error function \cite{Alexandrov:2016enp, Bringmann:2021dxg} to accommodate both the ``false'' as well as the ``mock'' directions.
 \item As mentioned in the introduction, Rademacher sum expressions are interesting for many purposes and are often available for holomorphic quantum modular forms of the kind we study here. It would be interesting to systematically develop the Rademacher sum techniques for general quantum modular forms. In terms of the physics on the field theory side, we wish to compare the $S^2\times S^1$ superconformal indices of the 3d theory $T[M_3]$, conjectured to be related to $\widehat Z$ by
 \[
 I^G(\tau)\sim \sum_b \widehat Z^G_b(\tau) \widehat Z^G_b(-\tau),
 \]
 with a summing over saddle point contributions from different gravity solutions. As argued in \cite{cheng20193d}, it is tempting to define $\widehat Z^G_b(M_3;-\tau)$ by identifying it with $\widehat Z^G_b(-M_3;\tau)$. On the gravity side, while we do not yet have a complete catalogue of supergravity solutions, the solutions described recently \cite{BenettiGenolini:2023rkq} in the ${\rm AdS}_4\times S^7$ context encouragingly take the form as geometries that might be matched with the different Rademacher contributions.
 \item Often, $\widehat Z$-invariants admit totally different expressions, arising from realizing $M_3$ not by plumbing but by surgery along knots \cite{Chung:2022ypb, Ekholm:2021irc, Ekholm:2020lqy,gukov2019two, Park:2019xey}, 
     or from alternative ways of expressing characters of logarithmic vertex algebras \cite{chengUpC}, leading to interesting $q$-series identities. While so far the analysis of quantum modularity relies mostly on the connection to lattice theta functions, it will be very interesting if modular properties can also be analyzed
 directly through these other expressions as well, as they are connected to yet different areas of mathematics and will lead to different applications.
 \item It will be very interesting to understand the nature of the recursive relation we observed in more concrete terms. We can think of the following routes for exploration. 1) Work out the recursion at higher rank in order to gain a more complete understanding of the recursive structure. 2) We already mentioned the analogy to the structure in higher rank Vafa--Witten theory \eqref{eqn:VW_recur}. It would be helpful to develop a similar interpretation for the 3d case. 3) Apart from the geometrical M-theory perspective, the Vafa--Witten recursion also admits an interpretation in terms of the reducible connections of the higher rank gauge group. From the ${\rm SL}(N,\mathbb C)$ Chern--Simons point of view, we believe it would be illuminating to work out the higher rank/higher depth analogue of \eqref{CS_transseries}, from which we should be able to see explicitly the role played by the lower rank flat connections. It is also desirable to compare with the resurgence analysis analogous to~\cite{Gukov:2016njj}. It will be particularly interesting to see what it means for the proposal in \cite{cheng20193d} to view the orbits of Weil representation as corresponding to the non-Abelian ${\rm SU}(2)$ flat connections on $M_3$, or relatedly to the different Wilson line insertions \cite{Gukov:2016gkn}. 
 \item According to the false-mock conjecture \cite{cheng20193d} 
  and its higher rank generalization, the recursion relation reported in Conjecture \ref{conj_rec} and Theorem \ref{thm:intro_brieskorn_recursive} should hold for $-M_3$, the orientation-flipped cousin of $M_3$, in a completely analogous way. It would be interesting to compute $\hat Z^G(-M_3)$ for higher rank $G$ and check it.
\end{itemize}

\appendix

\section{Special functions}\label{app:specialFn}

In this appendix we collect the definitions of the special functions used in the main text, as well as properties and relations that they satisfy. As for notations, we use throughout $q:={\rm e}^{2\pi {\rm i}\tau}$, where $\tau \in \BH$ and $v:=\operatorname{Im} \tau$. The functions
\[ 
E(u) := 2\int_0^u {\rm e}^{-\pi w^2} {\rm d}w, \qquad u\in \mathbb{R}
\]
and	
\be \label{def:M}
M(u):=\frac{{\rm i}}{\pi} \int_{\mathbb{R} - i u}
{\rm e}^{-\pi w^2 -2\pi {\rm i} u w} w^{-1} {\rm d}w, \qquad u\neq 0
\ee
are closely related to the error and the complementary error functions. A useful rewriting of~$M(u)$ is \cite{zwegers2008mock}
\begin{equation}
\label{eq:error_rewriting}
M(x\sqrt{v})={\rm i}\frac{x}{\sqrt{2}} q^{\frac{x^2}{4}} \int_{-\bar{\tau}}^{{\rm i}\infty} \frac{{\rm e}^{\frac{\pi {\rm i} x^2 w}{2}}}{\sqrt{-{\rm i}(w+\tau)}}{\rm d}w .
\end{equation}
They satisfy the relation
\[
M(u)=E(u)-\sgn(u),
\]
where
\[
\sgn(u):=\begin{cases}
	\hphantom{-}1 & \text{if}\ u>0, \\
	-1 & \text{if}\ u<0, \\
	\hphantom{-}0 &\text{if}\ u=0.
\end{cases}
\]

We also define
\[
M^\ast(u)=E(u)-\sgn^\ast(u),
\]
where
\[
\sgn^*(x) := \sgn (x) \quad \text{if}\quad x\neq 0\qquad\text{and} \qquad \sgn^*(0) :=1 .
\]

The generalised error function $E_2\colon\BR\times \BR^2 \to \BR$ is defined by
\[
E_2(\kappa ; \mathbf{u}) = E_2(\kappa ; u_1,u_2) := \int_{\BR^2} \sgn(w_1) \sgn(w_2+\kappa w_1)
{\rm e}^{-\pi\left( (w_1-u_1)^2 + (w_2-u_2)^2 \right)} {\rm d}w_1 {\rm d}w_2 .
\]
For $u_2,u_1-\kappa u_2 \neq 0$, the generalised complementary error function is
\be \label{def:M2}
M_2(\kappa ; \mathbf{u}) = M_2(\kappa ; u_1,u_2) := - \frac{1}{\pi^2} \int_{\mathbb{R} - {\rm i} u_2} \int_{\mathbb{R} - {\rm i} u_1}
\frac{{\rm e}^{ -\pi w_1^2 -\pi w_2^2 - 2\pi {\rm i}
		(u_1 w_1 + u_2 w_2) }}{w_2(w_1-\kappa w_2)} {\rm d}w_1{\rm d}w_2 .
\ee
These functions satisfy the following relation
\begin{align*} 
	M_2 (\kappa ; \mathbf{u})
	={}& E_2(\kappa ; u_1, u_2) + \sgn(u_1-\kappa u_2) \sgn(u_2) \\
	& -\sgn(u_2) E(u_1) - \sgn(u_1-\kappa u_2) E \left( \frac{\kappa u_1 + u_2 }{\sqrt{1+\kappa^2}} \right),
\end{align*}
 and
\begin{align}
	M_2^*(\kappa ; u_1, u_2) :={}& \sgn^*(u_1-\kappa u_2) \sgn^*(u_2)
	+ E_2(\kappa; u_1,u_2) \nonumber \\
	& - \sgn^*(u_2) E(u_1) - \sgn^*(u_1-\kappa u_2)
	E\left( \frac{\kappa u_1 + u_2}{\sqrt{1+\kappa^2}} \right) . \label{def:M2star}
\end{align}

The following identities hold for derivatives of the function $M_2(\kappa ; \mathbf{u})$ \cite{Alexandrov:2016enp}:
\begin{gather*}
M_2^{(0,1)}(\kappa ; \mathbf{u}) = \frac{2}{ \sqrt{1+\kappa^2} } {\rm e}^{- \frac{\pi (u_2+\kappa u_1)^2}{1+\kappa^2}}
M\left( \frac{u_1 - \kappa u_2}{\sqrt{1+\kappa^2}} \right),\\
M_2^{(1,0)}(\kappa ; \mathbf{u}) = 2 {\rm e}^{-\pi u_1^2} M(u_2) +
\frac{2\kappa}{\sqrt{1+\kappa^2}}
{\rm e}^{- \frac{\pi (u_2+\kappa u_1)^2}{1+\kappa^2}}
M\left( \frac{u_1 - \kappa u_2}{\sqrt{1+\kappa^2}} \right) .
\end{gather*}

\subsection*{Error function complements as integrals of theta functions}
Let $\theta_{i} ( \bmu,\mathbf{w} )$ be the following theta functions
\begin{align}
	\theta_1(\boldsymbol{\mu} ; \mathbf{w}) &= \sum_{\mathbf{n} \in \boldsymbol{\mu} + \BZ^2} (2n_1 + n_2) n_2 {\rm e}^{\frac{\pi {\rm i}}{2} \left( 3(2n_1 + n_2)^2 w_1 + n_2^2 w_2 \right)},\nonumber\\
	\theta_2(\boldsymbol{\mu} ; \mathbf{w}) &= \sum_{\mathbf{n} \in \boldsymbol{\mu} + \BZ^2} (3n_1 + 2n_2) n_1
	{\rm e}^{\frac{\pi {\rm i}}{2} \left( (3n_1 + 2n_2)^2 w_1 + 3n_1^2 w_2 \right)},\nonumber\\
	\theta_3 (\boldsymbol{\mu};\mathbf{w}) &= \sum_{\mathbf{n} \in \boldsymbol{\mu} + \BZ^2} (2n_1 + n_2)
	{\rm e}^{\frac{\pi {\rm i}}{2} \left( 3(2n_1 + n_2)^2 w_1 + n_2^2 w_2 \right)},\nonumber\\
	\theta_4 (\boldsymbol{\mu};\mathbf{w}) &= \sum_{\mathbf{n} \in \boldsymbol{\mu} + \BZ^2} (3n_1 + 2n_2)
	{\rm e}^{\frac{\pi {\rm i}}{2} \left( (3n_1 + 2n_2)^2 w_1 + 3n_1^2 w_2 \right)},\nonumber\\
	\theta_5 (\boldsymbol{\mu};\mathbf{w}) &= \sum_{\mathbf{n} \in \boldsymbol{\mu} + \BZ^2} n_1 \,
	{\rm e}^{\frac{\pi {\rm i}}{2} \left( (3n_1 + 2n_2)^2 w_1 + 3n_1^2 w_2 \right)} .\label{def:thetaellalw}
\end{align}
We can rewrite the error function complement $M_{2}\left( \kappa,\mathbf{u} \right)$ from equation \eqref{def:M2} as an iterated Eichler integral like in \cite{bringmann2018higher}
\begin{align*}
	&M_2\bigl(\sqrt{3};\sqrt{3v}(2n_1+n_2),\sqrt{v}n_2\bigr)\\
	&\qquad=-\frac{\sqrt{3}}{2}(2n_1+n_2)n_2q^{Q(\boldsymbol{n})}\int_{-\overline{\tau}}^{{\rm i} \infty}
	\frac{{\rm e}^{\frac{3\pi {\rm i}}{2}(2n_1+n_2)^2w_1}}{\sqrt{-{\rm i}(w_1+\tau)}}\int_{w_1}^{{\rm i} \infty}
	\frac{{\rm e}^{\frac{\pi {\rm i} n_2^2 w_2}{2}}}{\sqrt{-{\rm i}(w_2+\tau)}}{\rm d}w_2 {\rm d}w_1
	\\
	&\phantom{\qquad=}{}
	-\frac{\sqrt{3}}{2}(3n_1+2n_2)n_1q^{Q(\boldsymbol{n})} \int_{-\overline{\tau}}^{{\rm i} \infty}
	\frac{{\rm e}^{\frac{\pi {\rm i}}{2}(3n_1+2n_2)^2w_1}}{\sqrt{-{\rm i}(w_1+\tau)}}
	\int_{w_1}^{{\rm i} \infty}
	\frac{{\rm e}^{\frac{3\pi {\rm i} n_1^2 w_2}{2}}}{\sqrt{-{\rm i}(w_2+\tau)}}{\rm d}w_2 {\rm d}w_1,
\end{align*}
whereby
\begin{gather*}
	\sum_{\mathbf{n} \in \boldsymbol{\mu} + \BZ^2}
	M_2\bigl(\sqrt{3};\sqrt{3v}(2n_1+n_2),\sqrt{v}n_2\bigr)
	\\
	\qquad= -\frac{\sqrt{3}}{2}\int_{-\bar{\tau}}^{{\rm i} \infty}
	\int_{w_1}^{{\rm i} \infty}\frac{\theta_1({\boldsymbol\mu},
		{\bf w})+\theta_2({\boldsymbol\mu},
		{\bf w})}{\sqrt{-{\rm i}(w_1+\tau)}\sqrt{-{\rm i}(w_2+\tau)}}{\rm d}w_2{\rm d}w_1
\end{gather*}
and
\begin{gather*}
	\frac{1}{2\pi {\rm i}}\left[\frac{\partial}{\partial z}\left(M_2\left(\sqrt 3;\sqrt{3v}(2n_1+n_2),\sqrt{v}\left(n_2-\frac{2\operatorname{Im}(z)}{v}\right)\right){\rm e}^{2\pi {\rm i}n_2z}\right)\right]_{z=0}\\
	\qquad=\frac{\sqrt{3}}{2\pi}(2n_1+n_2)\int_{-\bar{\tau}}^{{\rm i}\infty} \frac{{\rm e}^{\frac{3{\rm i}\pi}{2}(2n_1+n_2)^2w_1}}{\sqrt{-{\rm i}(w_1+\tau)}}\int_{w_1}^{{\rm i}\infty} \frac{{\rm e}^{\frac{{\rm i}\pi n_2^2w_2}{2}}}{(-{\rm i}(w_2+\tau))^{\frac32}}{\rm d}w_2{\rm d}w_1\\
	\qquad\phantom{=}{}-\frac{\sqrt{3}}{4\pi}(3n_1+2n_2)\int_{-\bar{\tau}}^{{\rm i}\infty} \frac{{\rm e}^{-\frac{{\rm i}\pi}{2}(3n_1+2n_2)^2w_1}}{\sqrt{-{\rm i}(w_1+\tau)}}\int_{w_1}^{{\rm i}\infty} \frac{{\rm e}^{\frac{3{\rm i}\pi n_1^2w_2}{2}}}{(-{\rm i}(w_2+\tau))^{\frac32}}{\rm d}w_2{\rm d}w_1\\
	\qquad\phantom{=}{}-\frac{\sqrt{3}n_1}{4\pi}\int_{-\bar{\tau}}^{{\rm i}\infty} \frac{{\rm e}^{\frac{{\rm i}\pi}{2}(3n_1+2n_2)^2w_1}}{(-(w_1+\tau))^\frac32}\int_{w_1}^{{\rm i}\infty}\frac{{\rm e}^{\frac{3{\rm i}\pi n_1^2w_2}{2}}}{\sqrt{-{\rm i}(w_2+\tau)}}{\rm d}w_2{\rm d}w_1,
\end{gather*}
which implies
\begin{gather*}
	\frac{1}{2\pi {\rm i}}
	\sum_{\mathbf{n} \in \boldsymbol{\mu} + \BZ^2}
	\left[\frac{\partial}{\partial z}
	\left(M_2\left(\sqrt 3;\sqrt{3v}(2n_1+n_2),\sqrt{v}
	\left(n_2-\frac{2\operatorname{Im}(z)}{v}\right)\right)
	{\rm e}^{2\pi {\rm i}n_2z}\right)\right]_{z=0}\\
	\qquad=\frac{\sqrt{3}}{4 \pi}
	\int_{-\bar{\tau}}^{{\rm i} \infty}
	\int_{w_1}^{{\rm i} \infty}
	\frac{2\theta_3({\bmu}, {\bf w})-
		\theta_4({\bmu}, {\bf w})}{\sqrt{-{\rm i}(w_1+\tau)}(-{\rm i}(w_2+\tau)
		)^{\frac{3}{2}}}{\rm d}w_2{\rm d}w_1 \\
	\qquad\phantom{=}{}+\frac{\sqrt{3}}{4 \pi}\int_{-\bar{\tau}}^{{\rm i} \infty}\int_{w_1}^{{\rm i} \infty}\frac{\theta_5({\bmu}, {\bf w})}{(-{\rm i}(w_1+\tau))^{\frac{3}{2}}\sqrt{-{\rm i}(w_2+\tau)}}{\rm d}w_2{\rm d}w_1 .
\end{gather*}

\section{Proofs}
\subsection{Proof of Lemma \ref{lem:split_2_terms}}\label{pf_lem:split_2_terms}

Let $a, b\in W$ be the Weyl group elements whose action on a root $\vec k=\sum_{i=1,2} k_i^r \vec \alpha_i $ reads
\begin{gather}
	a\colon\ \vec{k} \mapsto (k_2^r-k_1^r)\vec{\al}_1 +k_2^r \vec{\al}_2,
	\qquad
	b\colon\ \vec{k} \mapsto k_1^r \vec{\al}_1 + (k_1^r-k_2^r) \vec{\al}_2 .
	\label{eq:abaction}
\end{gather}
They represent reflections with respect to the planes orthogonal to the simple roots $\vec\alpha_1$ resp.\ $\vec\alpha_2$.
In terms of these, we have $W=\{{\bf 1}, a, b, ab, ba, aba=bab\}$ and $W_+=\{{\bf 1}, ab, ba\}$.

From $0\leq \langle \vec k, \vec\omega_i\rangle$ for $i=1,2$, we conclude that at least one of the triple $\vec{k}$, $a\bigl(\vec{k}\bigr)$ and $b\bigl(\vec{k}\bigr)$ is in~$\bar P^+$. Evoking the identity
\begin{equation*}
	(-1)^{\ell(w')} F^{(\bos')} (\tau) = F^{(\bos)} (\tau)
\end{equation*}
for $\bos=\bigl(\vec{s},\vec k,m,D\bigr)$ and $\bos'=\bigl(w'(\vec{s}),w'(\vec k),m,D\bigr)$, from now on we assume that $\vec{k}\in \bar P^+$ without loss of generality.

In the sum over $\vec n$ in \eqref{dfn:gen_false}, write $\vec n=D\vec{m}+w\bigl(\vec{k}\bigr)$ for $\vec m\in \Lambda$. We have
\begin{equation*}
	\vec{m}+w\bigl(\vec{k}\bigr)\in \bar P^+ \quad \Leftrightarrow \quad m_i \geq \xi_{w,i}	, \quad i=1,2,
\end{equation*}
where $\xi_{w,i}$ are defined by
\begin{equation*}
	\xi_{w,i}:= \biggl\lceil - \frac{w\bigl(\vec{k}\bigr)\vert_i}{D} \biggr\rceil .
\end{equation*}

Since $0\leq \langle \vec k, \vec\omega_i\rangle < D$, we have
\begin{equation*}
	\big|\bigl(w\bigl(\vec{k}\bigr)\lvert_1-w\bigl(\vec{k}\bigr)\lvert_2\bigr)\big|\leq \max (2k_1+k_2, k_1+2k_2)< 3D,
\end{equation*} and hence $|\xi_{w,1}-\xi_{w,2}|\leq 3$.
The function $\min (n_i)$ is then given by
\begin{equation*}
	\min \bigl(Dm_i + w\bigl(\vec{k}\bigr)\vert_i\bigr) = \begin{cases}
		Dm_i + w\bigl(\vec{k}\bigr)\vert_i, & \text{for } m_i<m_{j}
		, \, i,j\in\{1,2\}, \\
		Dn + \min (w\bigl(\vec{k}\bigr)\vert_i), & \text{for } m_1=m_2=n
		.
	\end{cases}
\end{equation*}

For a given $w$, write the sum in \eqref{dfn:gen_false} as $F^{(\bos)} = \sum_w (-1)^{\ell(w)}F^{(\bos)}_w$. We now discuss $F^{(\bos)}_w$ in the following two cases.
\begin{itemize}\itemsep=0pt
	\item
	Case 1: $w\bigl(\vec{k}\bigr)\lvert_2\geq w\bigl(\vec{k}\bigr)\lvert_1$. \\
	In this case $\xi_{w,1}\geq \xi_{w,2}$ and
	\begin{align*}
		F_w^{(\bos)}(\tau) = & \sum_{\substack{m_2\geq m_1\geq \xi_{w,1}\\
				m_1\equiv m_2 (\mathrm{mod}3)}} \bigl(Dm_1 + w\bigl(\vec{k}\bigr)\vert_1\bigr) q^{p_{w,\vec{m}}}
		+ \sum_{\substack{m_1>m_2\geq \xi_{w,2}\\m_1\equiv m_2 (\mathrm{mod} 3)}} \bigl(Dm_2 + w\bigl(\vec{k}\bigr)\vert_2\bigr) q^{p_{w,\vec{m}}},
	\end{align*}
	where ${p_{w,\vec{m}}} = \frac{1}{2m}|-w(\vec\sigma) + m(\vec m)|^2$. By redefining the summation indices in the above equation in the following way
	\begin{equation*}
		\left( n_{1},n_{2} \right):=\begin{cases}
			\left(\frac{1}{3}\left( m_{1}-m_{2} \right),m_{2} \right), & m_{1}> m_{2},\\
			\left(\frac{1}{3}\left( m_{2}-m_{1} \right),m_{1} \right), & m_{2}\ge m_{1},
		\end{cases}
	\end{equation*}
	and shifting the summation ranges by $\xi_{w,1}$ resp.\ $\xi_{w,2}$, $F^{(\bos )}_w(\tau)$ can be rewritten as
	\begin{align*}
		F_w^{(\bos)} (q) ={}& \sum_{n_1,n_2\geq 0}
		\bigl(D(n_2+\xi_{w,1})+w\bigl(\vec{k}\bigr)\vert_1\bigr)q^{p_{w,(n_2+\xi_{w,1},3n_1+n_2+\xi_{w,1})}}
		 \\
		& +
		\bigl(D(n_2+\xi_{w,2})+w\bigl(\vec{k}\bigr)\vert_2\bigr)
		q^{p_{w,(3n_1+3+n_2+\xi_{w,2},n_2+\xi_{w,2})}} .
	\end{align*}
	Introducing then the quadratic form
	\begin{equation*}
		Q(\bon):= Q(n_1,n_2) = \frac{
			1}{2}|(n_2,3n_1+n_2)|^2=
		\bigl(3n_{1}^2+3n_1n_2+n_2^2\bigr),
	\end{equation*}
	we can write the function $F_w^{(\bos)} (q)$ in terms of this notation
	\begin{equation}\label{eqn:lem11_1}
		F^{(\bos)}_w (q)=
		\sum_{i=1,2} \sum_{n_1,n_2\geq 0} \bigl(D(n_2+\xi_{w,i})+w\bigl(\vec{k}\bigr)\vert_i\bigr)q^{mQ(\bon+\bal_w^{(i)})}
	\end{equation}
	with $ \bal_w^{(i)}$ given by \eqref{eq:alphawidef0A} with $x=0$.
	
	\item Case 2: $w\bigl(\vec{k}\bigr)\lvert_2< w\bigl(\vec{k}\bigr)\lvert_1$. \\
	This case can be treated analogously to Case 1, and gives \eqref{eqn:lem11_1} with $\bal_w^{(i)}$ given by \eqref{eq:alphawidef0A} with $x=1$.
\end{itemize}

The pairs of Weyl group elements $w,w'\in W$
\begin{equation*}
	(w,w')= (1,aba), (a,ba), (b,ab)
\end{equation*}
satisfy $w\bigl(\vec{k}\bigr)\lvert_i=-w'\bigl(\vec{k}\bigr)\lvert_j$ for $i\neq j$.

Since the condition $w\bigl(\vec{k}\bigr)\lvert_2\geq w\bigl(\vec{k}\bigr)\lvert_1$ is satisfied if and only if $w'\bigl(\vec{k}\bigr)\lvert_2\geq w'\bigl(\vec{k}\bigr)\lvert_1$, we have the relations
\[
\bal_{w'}^{(1)}=\brho-\bal_w^{(2)} \equiv \bar\bal_w^{(2)}, \qquad \bal_{w'}^{(2)}=\brho-\bal_w^{(1)} \equiv \bar\bal_w^{(1)}.
\]

Moreover, by shifting the summand we obtain
\begin{equation*}
		F_w^{(\bos)} (\tau)= \sum_{i=1,2} \sum_{\bon \in \mathbb{N}_0^2+\bal_w^{(i)}}
		\left(Dn_2+\frac{w(\vec{s})\vert_i D}{m}\right)q^{mQ(\bon)} .
\end{equation*}
Summing over all $w$, we arrive at the expressions in Lemma \ref{lem:split_2_terms}.

\subsection{Proof of Lemma \ref{lem:shifting}}
\label{pf_lem:shifting}

In the first step we consider ${\bal '}=\bal + (\delta \alpha_1,0)$. Then a routine computation shows that
\begin{gather*}
F_{0,\bal' }- F_{0,\bal}= \sum_{0\leq k\leq \delta \alpha_1 -1} q^{\frac{3}{4}(\alpha_1 +k)^2} \sum_{n\in \ZZ} \operatorname{sgn}(n+\tfrac{1}{2}) q^{\bigl(n+\alpha_2+\frac{3}{2}(\alpha_1 +k)\bigr)^2},\\
F_{1,\bal' }- F_{1,\bal}= -\sum_{0\leq k\leq \delta \alpha_1 -1} q^{\frac{3}{4}(\alpha_1 +k)^2} \sum_{n\in \ZZ} \operatorname{sgn}(n+\tfrac{1}{2}) (n+\alpha_2)q^{\bigl(n+\alpha_2+\frac{3}{2}(\alpha_1 +k)\bigr)^2},
\end{gather*}
while
\[
\sum_{n\in \ZZ} \operatorname{sgn}(n+\tfrac{1}{2}) q^{(n+\alpha_2+\tfrac{3}{2}(\alpha_1 +k))^2} - {\tilde\theta}^1\left[1,\alpha_2+\frac{3}{2}(\alpha_1 +k)\right](\t) \in \ZZ[q]
\]
since $ \operatorname{sgn}\bigl(n+\tfrac{1}{2}\bigr)- \operatorname{sgn}\bigl(n+\alpha_2+\tfrac{3}{2}(\alpha_1 +k)\bigr)$ has finite support. Similarly,
\begin{gather*}
\sum_{n\in \ZZ} \operatorname{sgn}\left(n+\tfrac{1}{2}\right) (n+\alpha_2)q^{(n+\alpha_2+\tfrac{3}{2}(\alpha_1 +k))^2} + \frac{3}{2}(\alpha_1 +k){\tilde\theta}^1\left[1,\alpha_2+\frac{3}{2}(\alpha_1 +k)\right](\t)\\
\qquad {}- {\tilde\theta}\left[1,\alpha_2+\frac{3}{2}(\alpha_1 +k)\right](\t) \in \ZZ[q] .
\end{gather*}
Second, we consider $\bb= {\bal '}+ (0,\delta \alpha_2) =\bal + (\delta \alpha_1,\delta \alpha_2)$. We have
\begin{gather*}
F_{0,\bb}- F_{0,\bal'}= \sum_{0\leq k\leq \delta \alpha_2 -1} q^{\frac{1}{4}(\alpha_2 +k)^2} \sum_{n\in \ZZ} \operatorname{sgn}\bigl(n+\tfrac{1}{2}\bigr) q^{3(n+\alpha_1'+\frac{1}{2}(\alpha_2 +k))^2},\\
F_{1,\bb}- F_{1,\bal'}= -\sum_{0\leq k\leq \delta \alpha_2 -1} q^{\frac{1}{4}(\alpha_2 +k)^2} (k+\alpha_2) \sum_{n\in \ZZ} \operatorname{sgn}\bigl(n+\tfrac{1}{2}\bigr) q^{3(n+\alpha_1'+\frac{1}{2}(\alpha_2 +k))^2}
\end{gather*}
and
\[
\sum_{n\in \ZZ} \operatorname{sgn}\bigl(n+\tfrac{1}{2}\bigr) q^{3(n+\alpha_1'+\frac{1}{2}(\alpha_2 +k))^2} - \tilde{\theta^1}\left[3,\alpha_1'+\frac{1}{2}(\alpha_2 +k)\right] \in \ZZ[q].
\]
Combining the above two steps proves the statement.

\subsection{Proof of Lemma \ref{lem:vanishingmainterm}}
\label{pf_lem:vanishingmainterm}

First note
\[
\sum_{\boldsymbol{\ell}\in (\ZZ/k\bar m)^2 } \ex\!\left(\frac{h}{k}Q(\bol+\bal)\right)\,
= \sum_{\boldsymbol{\ell}\in (\ZZ/k\bar m)^2 }\ex\!\left(\frac{h}{k}Q(\bol+{\bf 1}-\bal)\right)
\]
and the statement for $\nu=1$ immediately follows since $\eta_1(\bal)+\eta_1({\bf 1}-\bal)=0$ for all $\bal \in \CS$. Similarly, from the above identity we have
\begin{equation}\label{eqn:prooflem131}
\sum_{\bal\in\CS} \eta_0(\bal) \sum_{\boldsymbol{\ell}\in (\ZZ/k\bar m)^2 }\ex\left(\frac{h}{k}Q(\bol+\bal)\right)
= 2 \sum_{\bal\in\ss}\eta_0 (\bal) \sum_{\boldsymbol{\ell}\in (\ZZ/k\bar m)^2 }\ex\left(\frac{h}{k}Q(\bol+\bal)\right)\end{equation}
since $\eta_0(\bal)=\eta_0({\bf 1}-\bal)$ for all $\bal \in \CS$.

More generally, the sum over $\bol$ is invariant if one replaces the $\bal$ in summand with any~$\bal'$ as long as $\bal + \bal'\!\in \ZZ^2$ or $\bal -\bal'\in \ZZ^2$, as one can simultaneously shift $\bol$. Here we choose~\smash{$\bal_w^{(i)'}\! = \frac{1}{ m} \mathbf{a}_w^{(i)}$}, where
\begin{equation*}
\mathbf{a}_w^{(1)}=
\left(\Delta w(\vec{\sigma}), -
{w(\vec{\sigma})\lvert_1}{}\right), \qquad
\mathbf{a}_w^{(2)}=
\left(-\Delta w(\vec{\sigma}),
-{w(\vec{\sigma})\lvert_2}\right)
\end{equation*}
satisfy $\mathbf{a} \in \ZZ^2$ and
\begin{equation}\label{eqn:norm_a}
Q\left( \mathbf{a} \right)=\frac{1}{3} \left(\sigma _1^2+\sigma _2 \sigma _1+\sigma _2^2\right)=\frac{1}{2} |\boldsymbol{\sigma}|^2.
\end{equation}

Let $\left< \cdot,\cdot \right>_{Q}$ be twice the inner product induced by the quadratic form $Q$
\begin{equation*}
\left<\mathbf{v},\mathbf{w}\right>_Q:=Q(\mathbf{v}+\mathbf{w})-Q(\mathbf{v})-Q(\mathbf{w})=3\left(2v_1+v_2\right)w_1+\left(3v_1+2v_2\right)w_2=\left<\mathbf{w},\mathbf{v}\right>_Q.
\end{equation*}
Splitting the sum over $\bol$ into a sum over $\mathbf{N}$ and $\boldsymbol{\nu}$ by writing $\bol=\mathbf{N}+k\boldsymbol{\nu}$, we arrive at
\begin{align*}
\sum_{\boldsymbol{\ell}\in (\ZZ/k\bar m)^2 }\ex\left(\frac{h}{k}Q(\bol+\bal)\right) ={}& \sum_{\boldsymbol{\ell}\in (\ZZ/k\bar m)^2 }\ex\left(\frac{h}{k}Q(\bol+\frac{\mathbf{a}}{ m})\right)\\
={}& \sum_{\boldsymbol{\nu}\in (\ZZ/\bar m)^2 } \sum_{\boldsymbol{N}\in (\ZZ/k)^2 }\ex\left(\frac{h}{k}Q\left(\mathbf{N}+k\boldsymbol{\nu}+\frac{\mathbf{a}}{m}\right)\right)\\
={}& \ex\left( \frac{h}{k m^2} Q(\mathbf{a}) \right) \sum_{\boldsymbol{\nu}\in (\ZZ/\bar m)^2 } \ex\left(\frac{h/\delta}{\bar m}\left<\boldsymbol{\nu},\mathbf{a}\right>_Q \right) \\&\times \sum_{\boldsymbol{N}\in (\ZZ/k)^2 } \ex\left(\frac{h}{k}Q(\mathbf{N})\right) \ex\left(\frac{h}{km}\left<\mathbf{N},\mathbf{a}\right>_Q\right) .
\end{align*}
Focus on the factor \smash{$\sum_{\boldsymbol{\nu}\in (\ZZ/\bar m)^2 } \ex\left(\frac{h/\delta}{\bar m}\left<\boldsymbol{\nu},\mathbf{a}\right>_Q \right) $}, we see that the sum vanishes unless $\bar m\lvert 3a_1, a_2$, which is equivalent to $\vec \sigma \in \bar m \Lambda$, in which case
\begin{equation*}
 \sum_{\boldsymbol{\nu}\in (\ZZ/\bar m)^2 } \ex\left(\frac{h/\delta}{\bar m}\left<\boldsymbol{\nu},\mathbf{a}\right>_Q \right)=\bar m^2.
\end{equation*}

As a result, next we study the factor
\[\sum_{\boldsymbol{N}\in (\ZZ/k)^2 } \ex\left(\frac{h}{k}Q(\mathbf{N})\right) \ex\left(\frac{h}{km}\left<\mathbf{N},\mathbf{a}\right>_Q\right),\]
 when $\vec \sigma \in \bar m \Lambda$. Let $\delta^*$ be the modular inverse of $\delta$ mod $k$. This exists because $\delta$ is a divisor of $h$, which is coprime with $k$. Shifting the summation over ${\boldsymbol N}$ to ${\boldsymbol N}-\delta^*{\frac{\mathbf{a}}{\bar m}} $, we can cancel the~$\left<\mathbf{N},\mathbf{a}\right>_Q$ term and arrive at the result that
\[
\sum_{\boldsymbol{\ell}\in (\ZZ/k\bar m)^2 }\ex\left(\frac{h}{k}Q(\bol+\bal)\right) = \bar m^2 \ex\left( {\frac{h/\delta}{k}}\left(\delta^\ast+{\frac{1}{\delta}}\right) Q(\mathbf{a})\right),
\]
where $c$ and $c'$ are $\mathbf{a}$-independent constants that only depend on $h$, $k$ and $m$, $D$. Using the fact~$Q(\mathbf{a})$ is a constant over $\CS$ (see \eqref{eqn:norm_a}), we obtain from \eqref{eqn:prooflem131}
\begin{equation*}
\sum_{\bal\in\CS} \eta_0(\bal) \sum_{\boldsymbol{\ell}\in (\ZZ/k\bar m)^2 }\ex\left(\frac{h}{k}Q(\bol+\bal)\right)
= 2 c \ex\left( {\frac{c'}{2}}|\boldsymbol{\sigma}|^2 \right) \sum_{\bal\in\ss}\eta_0 (\bal),
\end{equation*}
which vanishes as a result of $\sum_{w\in W_+} w(\vec v) =0$ for any $\vec v$, and hence $ \sum_{\bal\in\ss}\eta_0 (\bal)=0$.

\subsection{Proof of Lemma \ref{lem:vanishingtermsWB-odd}}
\label{pf_lem:vanishingtermsWB-odd}

\begin{proof}
	We first rewrite, using $a_{i}:=m \bal_{i}$ and writing $\bol= \mathbf{N}+k\boldsymbol{\nu}$
	\begin{gather*}
		\sum_{0\leq \ell_1,\ell_2 < k\bar m}B_n\left( \frac{\ell_{1}+\alpha_{1}}{k\bar{m}} \right)\ex\left({\frac{h}{k}Q( \bol+\bal )}\right) \\
	\qquad	 =\ex\left(\frac{h}{k}Q( \bal )\right)\sum_{0\leq \mathbf{N}<k}
	\ex\left(\frac{{h/\delta}}{k\bar m}( \bar{m}\delta Q( \mathbf{N} )+3N_{1}( 2a_{1}+a_{2} )+N_{2}( 3a_{1}+2a_{2} ) )\right)
	\\
\phantom{\qquad	 =}{}	\times \sum_{0\leq \boldsymbol{\nu}<\bar{m}}B_{n}\left( \frac{N_{1}+k\nu_{1}+\alpha_{1}}{\bar{m}} \right)\ex\left(\frac{{h/\delta}}{\bar{m}}( 3\nu_{1}( 2a_{1}+a_{2} )+\nu_{2}( 3a_{1}+2a_{2} ) ) \right) .
\end{gather*}
The sum over $\nu_2$ shows that the quantity vanishes when $3a_{1}+2a_{2}$ is not divisible by $\bar m$. For both \smash{$\bal = \bal^{(1)}_w$} or \smash{$\bal = \bal^{(2)}_w$}, the condition is equivalent to the condition $\bar m\lvert \sum_{i=1,2} w(\vec \sigma)\lvert_i$. Writing $ \sum_{i=1,2} w(\vec \sigma)\lvert_i =\bar m y$, we write
	\begin{gather*}
	\bal_w^{(1)} = (\alpha_1,\alpha_2)~{\rm mod}~(0,1),\qquad
	\bal_w^{(2)} = \left(1-\alpha_1,1-\alpha_2-{\frac{y}{\delta}}\right)~{\rm mod}~(0,1).
\end{gather*}
	
Invoking the reflection property \eqref{reflectionpropertyBernoulli} of the Bernoulli polynomials, we have for $\bal = \bal_w^{(2)}$
	\begin{align*}
	& \sum_{\substack{0\leq \ell_1< k\bar m\\\ell_2\in \ZZ/k\bar m}} B_n\left( \frac{\ell_{1}+1-\alpha_{1}}{k\bar{m}} \right)
	\ex\left({\frac{h}{k}Q\left( \bol+{\bf 1}-\bal -\frac{y}{\delta}(0,1)\right)}\right)
	\\
 &\qquad =
	\sum_{\substack{0\leq \ell_1< k\bar m\\\ell_2\in \ZZ/k\bar m}}B_n\left(1- \frac{\ell_{1}+\alpha_{1}}{k\bar{m}} \right)\ex\left({\frac{h}{k}Q\left( \bol+\bal +\frac{y}{\delta}(0,1) \right)}\right) \\
 &\qquad= 	- \sum_{\substack{0\leq \ell_1< k\bar m\\\ell_2\in \ZZ/k\bar m}}B_n\left( \frac{\ell_{1}+\alpha_{1}}{k\bar{m}} \right)\ex\left({\frac{h}{k}Q\left( \bol+\bal +{y}\left(\frac{1}{\delta}-\delta^\ast\right)(0,1) \right)}\right) \\
 &\qquad	= - \sum_{\substack{0\leq \ell_1< k\bar m\\\ell_2\in \ZZ/k\bar m}}B_n\left( \frac{\ell_{1}+\alpha_{1}}{k\bar{m}} \right)\ex\left({\frac{h}{k}Q\left( \bol+\bal \right)}\right).
	\end{align*}
	
Going from the first to the second line, we have relabeled $\bol$ by $(k\bar m-1){\bf 1}-\bol$. Going to the third line, we have invoked the reflection property \eqref{reflectionpropertyBernoulli} of the Bernoulli polynomials, and shifted $\ell_2$ in the sum by $\delta^\ast y$, where $\delta^\ast \delta \equiv 1 ~~(k) $. In the last step, we used that $\left<\bal,(0,1)\right>_Q = -{\frac{1}{\delta} } y$. From this, we immediately see that the contributions from \smash{$\bal_w^{(1)}$} and \smash{$\bal_w^{(2)}$} cancel.
\end{proof}

\subsection{Proof of Lemma \ref{lem:Estar_vs_starless}}\label{apx:proofEstar_vs_starless}

For $\nu=0$ and $\varepsilon\in \{ 1,-1 \}$, we have
\begin{align*}
 \mathbb{E}_{0}^{( \varrho )}( \tau )
 ={}& \frac{1}{2}\sum_{\varepsilon\in\{1,-1\}}\sum_{\bmu\in \tilde{\CS}}\eta_{0}( \bmu )\sum_{\bon \in \bmu + \BZ}q^{-Q( \bon )}M_{2}\bigl( \sqrt{3};\sqrt{3v}( 2\varepsilon n_{1}+n_{2} ), \sqrt{v}n_{2} \bigr)\\
 ={}& \frac{1}{2}\sum_{\varepsilon\in\{1,-1\}}\sum_{\bmu\in\CS}\eta_{0}( \bmu )\left(\sum_{\bon\in\bmu+\BN_{0}^{2}}q^{-Q\left( \bon \right)}M_{2}\bigl( \sqrt{3};\sqrt{3v}( 2\varepsilon n_{1}+n_{2} ), \sqrt{v}n_{2} \bigr)\right.\\
 &\left.+\sum_{\bon\in ( 1-\mu_{1},\mu_{2} )+\BN_{0}^{2}}q^{-Q( -n_{1},n_{2} )}
 M_{2}\bigl( \sqrt{3};\sqrt{3v}( 2\varepsilon n_{1}+n_{2} ), \sqrt{v}n_{2} \bigr)\right) \\
 ={}&\frac{1}{2}\sum_{\bmu\in\CS}\eta_{0}( \bmu )\left( \sum_{\mathbf{n}\in \bmu + \BN_{0}^{2}} q^{-Q( \bon )}\bigl( M_{2}^{*}\bigl( \sqrt{3};\sqrt{3v}( 2n_{1}+n_{2} ),\sqrt{v}n_{2} \bigr)\right.\\
 & \left.+\delta_{n_{1},0}( 1-\delta_{n_{2},0} )M\bigl( 2\sqrt{v}n_{2} \bigr)+\delta_{n_{2},0}( 1-\delta_{n_{1},0} )M\bigl( 2\sqrt{3v}n_{1} \bigr)-\delta_{n_{1},0}\delta_{n_{2},0}\bigr)\right. \\
 & + \sum_{\mathbf{n}\in ( 1-\mu_{1},\mu_{2} ) + \BN_{0}^{2}} q^{-Q( -n_{1},n_{2} )}\bigl( M_{2}^{*}\bigl( \sqrt{3};\sqrt{3v}( -2n_{1}+n_{2} ),\sqrt{v}n_{2} \bigr)\\
 & -\delta_{n_{1},0}( 1-\delta_{n_{2},0} )M\bigl( 2\sqrt{v}n_{2} \bigr)-\delta_{n_{2},0}( 1-\delta_{n_{1},0} )M\bigl( 2\sqrt{3v}n_{1} \bigr)+\delta_{n_{1},0}\delta_{n_{2},0}\bigr)\bigr) \\
 ={}& \BE_{0}^{*( \varrho )}( \tau )+\frac{1}{2}\!\sum_{\bmu\in\tilde{\CS}}\eta( \bmu )\!\left( \sum_{\bon\in\bmu+\BN_{0}^{2}}\!+\!\sum_{\bon\in( 1,1 )-\bmu+\BN_{0}^{2}}\!-\!\sum_{\bon\in(1-\mu_{1},\mu_{2})+\BN_{0}^{2}}\!-\!\sum_{\bon\in(\mu_{1},1-\mu_{2})+\BN_{0}^{2}} \right)\\
 &\times\bigl( \delta_{n_{1},0}( 1-\delta_{n_{2},0} )M\bigl( 2\sqrt{v}n_{2} \bigr)+\delta_{n_{2},0}( 1-\delta_{n_{1},0} )M\bigl( 2\sqrt{3v}n_{1} \bigr)-\delta_{n_{1},0}\delta_{n_{2},0} \bigr)q^{-Q( \bon )}\\
 ={}& \BE_{0}^{*( \varrho )}(\tau)+\frac{1}{2}\sum_{\bmu \in \tilde \CS}\eta_{0}(\bmu)X_{0}( \bmu ).
\end{align*}

For $\nu=1$, we have
\begin{align*}
 \BE_{1}^{( \varrho )}( \tau )
 ={}& \frac{1}{2}\sum_{\bmu\in\CS}\eta_{1}( \bmu )\sum_{\bon\in\bmu+\BN_{0}^{2}}n_{2}M_{2}\bigl( \sqrt{3};\sqrt{3v}( 2n_{1}+n_{2} ) \bigr)q^{-Q( \bon )}\\
 &+\frac{1}{2}\sum_{\bmu\in\bar\CS}\bar\eta_{1}( \bmu )\sum_{\bon\in\bmu+\BN_{0}^{2}}n_{2}M_{2}\bigl( \sqrt{3};\sqrt{3v}(-2n_{1}+n_{2} ) \bigr)q^{-Q( -n_{1},n_{2} )}\\
 &+\frac{1}{4\pi\sqrt{v}}\sum_{\bmu\in\CS}\eta_{1}( \bmu )\sum_{\bon\in\bmu+\BN_{0}^{2}}{\rm e}^{-\pi( 3n_{1}+2n_{2} )^{2}v}M\bigl( \sqrt{3v}n_{1} \bigr)q^{-Q( \bon )}\\
 &+\frac{1}{4\pi\sqrt{v}}\sum_{\bmu\in\bar\CS}\bar\eta_{1}( \bmu )\sum_{\bon\in\bmu+\BN_{0}^{2}}{\rm e}^{-\pi( -3n_{1}+2n_{2} )^{2}v}M\bigl( -\sqrt{3v}n_{1} \bigr)q^{-Q( -n_{1},n_{2} )}\\
 ={}& \frac{1}{2}\sum_{\bmu\in\CS}\eta_{1}( \bmu )\Bigg[ \sum_{\bon\in\bmu+\BN_0^2}\bigl( n_{2}M_{2}\bigl( \sqrt{3};\sqrt{3v}( 2n_{1}+n_{2} ),\sqrt{v}n_{2} \bigr)+\frac{1}{4\pi\sqrt{v}}{\rm e}^{-\pi( 3n_{1}+2n_{2} )^{2}v}\\
 &\times M\bigl( \sqrt{3v}n_{1} \bigr)\bigr)q^{-Q( \bon )}+ \sum_{\bon\in( 1-\mu_{1},\mu_{2} )+\BN_0^2}\bigl( n_{2}M_{2}\bigl( \sqrt{3};\sqrt{3v}( -2n_{1}+n_{2} ),\sqrt{v}n_{2} \bigr)\\
 &+\frac{1}{4\pi\sqrt{v}}{\rm e}^{-\pi( -3n_{1}+2n_{2} )^{2}v}M\bigl( -\sqrt{3v}n_{1} \bigr)\bigr)q^{-Q( -n_{1},n_{2} )}\bigg]\\
 ={}& \BE_{1}^{*( \varrho )}( \tau )+\frac{1}{2}\sum_{\bmu\in\tilde\CS}\eta_{1}( \bmu )X_{1}( \bmu ).
\end{align*}

\subsection{Proof of Proposition \ref{prop:Estarless_eichler}}\label{app:proofPropIntegral}

Following \cite{bringmann2018higher}, we can rewrite $\mathbb{E}^{(\bos)}_{0}$ and $\mathbb{E}^{(\bos)}_{1}$ as
\[
	\BE_0^{(\bos)}(\tau) =
 -\frac{\sqrt{3}}{4} \sum_{\boldsymbol{\alpha} \in \tilde{\CS}}
	\eta_0(\boldsymbol{\alpha})
	\int_{-\bar{\tau}}^{{\rm i} \infty}\int_{z_1}^{{\rm i} \infty}
	\frac{\theta_1({\boldsymbol\alpha},
	{\bf z})+\theta_2({\boldsymbol\alpha},
	{\bf z})}{\sqrt{-{\rm i}(z_1+\tau)}\sqrt{-{\rm i}(z_2+\tau)}}{\rm d}z_2{\rm d}z_1
\]
and
\begin{align*}
\BE_1^{(\bos)}(\tau) ={}&\frac{\sqrt{3}}{8 \pi}
\sum_{\boldsymbol{\alpha} \in \tilde{\CS}}
\int_{-\bar{\tau}}^{{\rm i} \infty}
\int_{z_1}^{{\rm i} \infty}\frac{2\theta_3({\bal}, {\bf z})-
\theta_4({\bal},
{\bf z})}{\sqrt{-{\rm i}(z_1+\tau)}(-{\rm i}(z_2+\tau))^{\frac{3}{2}}}
{\rm d}z_2{\rm d}z_1 \\
 &+\frac{\sqrt{3}}{8 \pi}\sum_{\boldsymbol{\alpha} \in \tilde{\CS}}
\int_{-\bar{\tau}}^{{\rm i} \infty}
\int_{z_1}^{{\rm i} \infty}\frac{\theta_5({\bal},
{\bf z})}{(-{\rm i}(z_1+\tau))^{\frac{3}{2}}\sqrt{-{\rm i}(z_2+\tau)}} {\rm d}z_2{\rm d}z_1 .
\end{align*}
The functions $\theta_\ell({\bal}, {\bf z})$ are defined in equations \eqref{def:thetaellalw} and can be equivalently written as
\begin{align*}
\theta_1(\bal,\mathbf{z}) &=
\frac{1}{m^2} \sum_{\delta\in\ZZ/2}
\theta^1_{m,m(2\al_1+\al_2+\delta)}\left(\frac{3z_1}{m}\right)
\theta^1_{m,m(\al_2+\delta)}\left(\frac{z_2}{m}\right),
\\
\theta_2(\bal,\mathbf{z}) &=
\frac{1}{m^2} \sum_{\delta\in\ZZ/2}
\theta^1_{m,m(3\al_1+2\al_2+\delta)}\left(\frac{z_1}{m}\right)
\theta^1_{m,m(\al_1+\delta)}\left(\frac{3z_2}{m}\right),
\\
\theta_3 (\bal,\mathbf{z}) &=
\frac{1}{m} \sum_{\delta\in\ZZ/2}
\theta^1_{m,m(2\al_1+\al_2+\delta)}\left(\frac{3z_1}{m}\right)
\theta^0_{m,m(\al_2+\delta)}\left(\frac{z_2}{m}\right),
\\
\theta_4 (\bal,\mathbf{z}) &=
\frac{1}{m} \sum_{\delta\in\ZZ/2}
\theta^1_{m,m(3\al_1+2\al_2+\delta)}\left(\frac{z_1}{m}\right)
\theta^0_{m,m(\al_1+\delta)}\left(\frac{3z_2}{m}\right),
\\
\theta_5 (\bal,\mathbf{z}) &=
\frac{1}{m} \sum_{\delta\in\ZZ/2}
\theta^0_{m,m(3\al_1+2\al_2+\delta)}\left(\frac{z_1}{m}\right)
\theta^1_{m,m(\al_1+\delta)}\left(\frac{3z_2}{m}\right).
\end{align*}
Most of these terms however sum to zero as proved in the following lemma.
\begin{lem}
Using the above definitions,
 \begin{gather*}
 \sum_{\bal\in \tilde{\CS}}\eta_{0}( \bal )\theta_{1}( \bal,\mathbf{z} ) = 0,\qquad
 \sum_{\bal\in \tilde{\CS}}\theta_{3}( \bal,\mathbf{z} ) = 0,\qquad
 \sum_{\bal\in \tilde{\CS}}\theta_{5}( \bal,\mathbf{z} ) = 0.
 \end{gather*}
\end{lem}
\begin{proof}
 Due to the symmetries of the theta series and the sum over $\delta$, we only need to focus on the non-integer part of the $\bal$ defined in equation \eqref{eq:alphawidef0A}. By direct computation one can see that
 \begin{align*}
 &\eta_{0}\bigl( \bal^{( 1 )}_{w} \bigr)\theta_{1}\bigl( \bal^{( 1 )}_{w},\mathbf{z} \bigr) = -\eta_{0}\bigl( \bal^{( 1 )}_{baw} \bigr)\theta_{1}\bigl( \bal^{( 2 )}_{baw},\mathbf{z} \bigr),\\
& \theta_{3}\bigl( \bal^{( 1 )}_{w},\mathbf{z} \bigr) = -\theta_{3}\bigl( \bal^{( 2 )}_{baw},\mathbf{z} \bigr),\qquad
 \theta_{5}\bigl( \bal^{( 1 )}_{w},\mathbf{z} \bigr) = -\theta_{5}\bigl( \bal^{( 2 )}_{w},\mathbf{z} \bigr).
 \end{align*}
The result follows from the fact that $ba$ and $\id$ are in $W^{+}$.
\end{proof}

This yields
 \[ 
	\BE_\nu^{(\bos)}(\tau) =
	-\frac{\sqrt{3}}{4} (2\pi)^{-\nu} \sum_{\boldsymbol{\mu} \in \tilde \CS}
	\eta_\nu(\boldsymbol{\mu})
	\int_{-\bar{\tau}}^{{\rm i} \infty}\int_{z_1}^{{\rm i} \infty}
	\frac{ \Theta_\nu({\boldsymbol\mu};
	{\bf z})}{({-{\rm i}(z_1+\tau)})^{1/2} (-{\rm i}(z_2+\tau))^{\nu+1/2}}{\rm d}z_2{\rm d}z_1,
 \]
where
 \[
 \Theta_\nu(\boldsymbol{\mu}, \mathbf{z}) =
 (m)^{-2+\nu} \sum_{\delta\in\ZZ/2}
 \theta^1_{m,m(3\mu_1+2\mu_2+\delta)}\left(\frac{z_1}{m}\right)
 \theta^{1-\nu}_{m,m(\mu_1+\delta)}\left(\frac{3z_2}{m}\right) .
 \]
Substituting then the elements $\bal^{(1)}_{w}$ and $\bal^{(2)}_{w}$ of the set $\tilde\CS$ for each $w\in W^+$ and using the shift and symmetry properties of theta functions \smash{$\theta^\nu_{m,r}$} for $\nu=0,1$ allows to reduce the summation over $\boldsymbol{\mu} \in \tilde \CS$ to a summation over $\boldsymbol{w} \in W^+$ in the expression for \smash{$\BE_\nu^{(\bos)}(\tau)$} in terms of \smash{$ \Theta_{\nu,w}^{(\bos)}(\mathbf{z})$}.

\subsection{Proof of Lemma \ref{lem:rewritingBBE}}\label{pf_lem:rewritingBBE}

\begin{proof}
When $M_3=\Sigma(p_1,p_2,p_3)$ we have the unique $\bv=\bv_0$, $D=1$, and $m=p_1p_2p_3$. Using equation \eqref{s3:hombl-2}, Proposition \ref{prop:E01asymp} and Lemma \ref{lem:Estar_vs_starless}, we can express the rank two part of companion of the $\widehat{Z}$-invariant in terms of the functions \smash{$\BE_\nu^{(\bos)}(\tau)$} defined in \eqref{def:E01starless} as
	\begin{equation*}
	\widecheck{Z}^{{\rm SU}(3)}_{\bv_0}(\tau)=
	C\bigl(q^{-1}\bigr) \sum\limits_{\hat{w} \in
	W^{\otimes 3}}(-1)^{\ell(\hat w)}
 \bigl(\mathbb{E}^{(\bos_{\hat{w}})}_0(m\tau)+
 \mathbb{E}^{(\bos_{\hat{w}})}_1({m}\tau)\bigr),
	\end{equation*}
up to a one-dimensional piece, where we have \[\rho_{\hat w} = \bigl(\vec{\sigma}_{\hat w},\vec{k}_{\hat w},m,D\bigr) = (\vec{s}_{\hat w}, 0, p_1p_2p_3, 1).\]

Together with $\BE_{\nu,\tilde w}^{(\bos)} = -
\BE_{\nu,\tilde w (aba)}^{(\bos_{})}$ in the notation of \eqref{eq:abaction}, which can easily be seen from
\[
aba \, \vec\rho=-\vec\rho, \qquad aba \, \Delta \vec\omega = \Delta \vec\omega,
\]
we can extend the sum in Proposition \ref{prop:Estarless_eichler} to write
\[
	\BE_\nu^{(\bos)}(\tau) = {\frac{1}{2}} \sum_{{w} \in W} (-1)^{\ell(\hat w)} \BE_{\nu,w}^{(\bos)}(\tau).
\]
We then have the following identity
	\begin{gather*}
	C \bigl(q^{-1}\bigr)\sum_{\nu=0,1} \sum\limits_{\hat{w} \in
	W^{\otimes 3}}(-1)^{\ell(\hat w)}
\mathbb{E}^{(\bos_{\hat{w}})}_\nu(m\tau)
 \\ \qquad ={\frac{1}{2}}C (q^{-1})\sum_{\nu=0,1}
 \sum\limits_{\hat{w} \in
	W^{\otimes 3}} \sum_{{w} \in W} (-1)^{\ell(\hat w)} (-1)^{\ell(w)} \BE_{\nu,w}^{(\bos_{\hat w})}(m\tau) \\
	\qquad = {\frac{1}{2}}C \bigl(q^{-1}\bigr)\sum_{\nu=0,1}
 \sum\limits_{\hat{w} \in
	W^{\otimes 3}} \sum_{{w} \in W} (-1)^{\ell(w\hat w)}
	\BE_{\nu,e}^{(\bos_{w\hat w})}(m\tau)
	\\ \qquad = {\frac{1}{2}} |W|C (q^{-1})\sum_{\nu=0,1}
 \sum\limits_{\hat{w} \in
	W^{\otimes 3}} (-1)^{\ell(\hat w)}
	\BE_{\nu,e}^{(\bos_{\hat w})}(m\tau),
\end{gather*}
where we have used $\BE_{\nu,\tilde w}^{(\bos_{\hat w})} =
\BE_{\nu,e}^{(\bos_{\tilde w \hat w})}$ in the third line, which is manifest from \eqref{def:Theta}. Combining the above with the double Eichler integral expression in \eqref{Eichler_Enu_Com} for \smash{$\BE_{\nu,\tilde w}^{(\bos_{\hat w})}$} leads to the statement of the lemma.
\end{proof}

\subsection{Proof of Proposition \ref{prop:sphere_recursive}}\label{prf:sphere_recursive}
We take as our starting point Lemma \ref{lem:brieskorn_thetas}, which states
\begin{align*}
\widecheck{Z}^{{\rm SU}(3)}_{\bv_0}(M_3;\tau)
 ={}& z_{1d}+ \tfrac{|W|}{2m}C\bigl(q^{-1}\bigr) \sum\limits_{\hat{w}\in W^{\otimes 3}}(-1)^{\ell(\hat{w})} \sum_{\nu=0,1} \frac{\sqrt{3}}{4 \pi^{\nu}}	\left(\frac {3\Delta\vec{s}_{\hat w}}{m}\right)^{1-\nu}\\
 &\times\sum_{\delta\in\BZ/2} \bigl(\vartheta'_{w,\delta},\vartheta^{1-\nu}_{w,\delta}\bigr)^\ast(\tau),
\end{align*}
where the non-holomorphic double Eichler integral is of the theta functions
\begin{gather*}
 \vartheta'_{w,\delta}(\tau)=
 \theta^1_{m, m\delta+\langle \vec\rho, \vec\sigma\rangle}(\tau),\qquad
 \vartheta^{1-\nu}_{w,\delta}(\tau) = \theta^{1-\nu}_{m,m\delta + \langle \Delta\vec\omega, \vec{\sigma}\rangle}(3\tau),
 \end{gather*}
for
\begin{align*}
\vec\sigma_{\hat w} &= \vec{s}_{\hat{w}} = -\sum\limits_{i=1}^3 \bar{p}_i w_i(\vec{\rho}) .
\end{align*}

From $(aba)\vec\rho=-\vec\rho$, we have
\[
\vec\sigma_{(w_1(aba)^{\varepsilon_1}, w_2(aba)^{\varepsilon_2}, w_3(aba)^{\varepsilon_3} )}
= -\sum\limits_{i=1}^3 (-1)^{\varepsilon_i} \bar{p}_i w_i(\vec{\rho})
\]
for $\varepsilon_i\in \ZZ/2$. Then
\[
 \widecheck{Z}^{{\rm SU}(3)}_{\bv_0}(M_3;\tau)
 =z_{1d}+ \tfrac{|W|}{2m}C\bigl(q^{-1}\bigr)
 \sum\limits_{\hat{w}\in W_+^{\otimes 3}}\sum_{\nu=0,1} \frac{\sqrt{3}}{4 \pi^{\nu}} \sum_{\delta\in\BZ/2} \tilde\BE_{\nu,\delta}^{(\bos_{\hat{w}})}(\tau),
\]
where ${\tilde \BE_{\nu,\delta}}$ is the integral
\[
	{\tilde \BE_{\nu,\delta}^{(\rho_{\hat{w}})}}(\tau) := -
	\int_{-\bar{\tau}}^{{\rm i} \infty}\int_{z_1}^{{\rm i} \infty}
	\frac{ \tilde\Theta_{\nu,\delta}^{(\bos_{\hat{w}})}(\mathbf{z})}{({-{\rm i}(z_1+\tau)})^{1/2} (-{\rm i}(z_2+\tau))^{\nu+1/2}}{\rm d}z_2{\rm d}z_1
\]
of
\begin{align*}
 \tilde \Theta_{\nu,\delta}^{(\bos_{\hat{w}})}(\mathbf{z})={}& \sum_{\varepsilon_1,\varepsilon_2, \varepsilon_3\in \ZZ/2} (-1)^{\sum_{i}\varepsilon_i} {\bigg( \sum_{i}\tfrac{1}{p_i} (-1)^{\varepsilon_i} \langle \Delta\vec\omega, w_i(\vec{\rho})\rangle\bigg)^{1-\nu}} \\
 &\times \theta^1_{m,m\delta + \sum_{i} (-1)^{\varepsilon_i} \bar{p}_i \langle\vec{\rho}, w_i(\vec{\rho})\rangle} ({z_1}) \theta^{1-\nu}_{m,m\delta -\sum_{i} (-1)^{\varepsilon_i} \bar{p}_i \langle \Delta\vec\omega, w_i(\vec{\rho})\rangle}({3z_2}{m} )
\\
={}&2 \sum_{\varepsilon_1,\varepsilon_2\in \ZZ/2} \tilde \Theta_{\nu,\delta,(\varepsilon_1,\varepsilon_2)}^{(\bos_{\hat{w}})}(\mathbf{z}),
\end{align*}

\noindent where we have used that the summand is invariant under $(\varepsilon_1,\varepsilon_2,\varepsilon_3)\mapsto (1,1,1)+(\varepsilon_1,\varepsilon_2,\varepsilon_3)$, and we write

\begin{align*}
\tilde \Theta_{\nu,\delta,(\varepsilon_1,\varepsilon_2)}^{(\bos_{\hat{w}})}(\mathbf{z}) :={}&
\biggl((-1)^{\sum_{i}\varepsilon_i}
{\bigg(\sum_{i}\tfrac{1}{p_i} (-1)^{\varepsilon_i} \langle \Delta\vec\omega, w_i(\vec{\rho})\rangle\bigg)^{1-\nu}} \\
&\times\theta^1_{m,m\delta + \sum_{i} (-1)^{\varepsilon_i} \bar{p}_i \langle\vec{\rho}, w_i(\vec{\rho})\rangle} ({z_1}) \theta^{1-\nu}_{m,m\delta - \sum_{i} (-1)^{\varepsilon_i} \bar{p}_i \langle \Delta\vec\omega, w_i(\vec{\rho}) \rangle}({3z_2})\Bigr)\lvert_{\varepsilon_3=0}.
\end{align*}

To simplify notation, in this appendix we will often skip writing the arguments of the functions, with the understanding that $ \theta^{1-\nu}_r = \theta^{1-\nu}_r\left({3z_2}\right)$ and $\theta^1_r=\theta^1_r\left({z_1}\right)$.

Using
\[
\langle ab\vec\rho,\vec\rho \rangle= \langle ba\vec\rho,\vec\rho \rangle = \langle ab\vec\rho,\Delta\vec\omega \rangle
= -\langle ba\vec\rho,\Delta\vec\omega \rangle =-1
\]
and $\langle \vec\rho,\vec\rho \rangle = 2$, $\langle \vec\rho,\Delta\vec\omega \rangle =0$, as well as
\[
\theta^{1-\nu}_{m,m\delta+r}= (-1)^{\nu-1}\theta^{1-\nu}_{m,m\delta-r}\quad
\text{for all}\quad\delta\in \ZZ/2, \quad\nu\in \{0,1\},\quad r\in \ZZ/2m,
\]
it is straightforward to discuss the separate contributions individually.

{\em Case $1$: $\hat w=(e,e,e)$.}
From $\langle \vec\rho,\Delta\vec\omega \rangle =0$,
we see $ \tilde \Theta_{\nu,\delta}^{(\bos_{\hat{w}})}(\mathbf{z})=0$ for $\nu=0$, and
\begin{align*}
\tilde\Theta_{\nu,\delta}^{(\bos_{(e,e,e)})}(\mathbf{z})&={2}\theta^{1-\nu}_{m,m\delta } \sum_{\varepsilon_1,\varepsilon_2\in \ZZ/2} (-1)^{\sum_{i}\varepsilon_i}
\theta^1_{m,m\delta + 2 \sum_{i} {(-1)^{\varepsilon_i}\bar{p}_i} + 2\bar p_3}
\\
&={2} \theta^{1-\nu}_{m,m\delta } \sum_{\epsilon, \epsilon'}\epsilon\epsilon'
\theta^1_{m,m\delta + 2 \epsilon \bar p_1 + 2 \epsilon' \bar p_2 + 2\bar p_3}
\end{align*}
for $\nu=1$.

{\em Case $2$: $\hat w=(ab,e,e)$, $\hat w=(ba,e,e)$ and permutations.}
In the case of $\hat w=(ab,e,e)$, we have
\begin{gather*}
 \Theta_{\nu,\delta}^{(\bos_{(ab,e,e)})}(\mathbf{z}) =\left(\frac{1}{p_1}\right)^{1-\nu} \theta^{1-\nu}_{m,m\delta+\bar{p}_1}
 \sum_{\epsilon,\epsilon'\in \{1,-1\}} \epsilon \epsilon' \theta^1_{m,m\delta + \epsilon{\bar p_1}+ 2\epsilon'{\bar p_2}+2{\bar p_3}} .
\end{gather*}
Similarly, $\hat w=(ba,e,e)$ renders the same answer and we get
\begin{align*}
 \tilde\Theta_{\nu,\delta}^{(\bos_{(ba,e,e)})}(\mathbf{z})
&= \left(\frac{1}{p_1}\right)^{1-\nu}
 \theta^{1-\nu}_{m,m\delta+\bar p_1} \sum_{\epsilon,\epsilon'\in \{1,-1\}} \epsilon \epsilon' \theta^1_{m,m\delta + \epsilon\bar p_1+ 2\epsilon \bar p_2+ 2 \bar p_3} .
\end{align*}

All other six choices of $\hat w\in W_+^{\otimes 3}$ where only one of the three elements is different from $e\in W$ can be treated in exactly the same way, and we get the sum
\begin{gather*}\label{intro-PplusEvPerm}
\sum_{\substack{\hat w=(w_1,w_2,w_3)\\{\rm{one~of~the~}}w_i \neq e }} \tilde\Theta_{\nu,\delta}^{(\bos_{\hat w})}
 = 2 {\bf P}^+ \left(\left(\frac{1}{p_1}\right)^{1-\nu}
 \theta^{1-\nu}_{m,m\delta+\bar p_1} \sum_{\epsilon,\epsilon'\in \{1,-1\}} \epsilon\epsilon' \theta^1_{m,m\delta +\epsilon \bar p_1+ 2 \epsilon' \bar p_2+ 2\bar p_3 }\right),
\end{gather*}
where we denote by ${\bf P}^+$ by the group of even permutations of $(p_1,p_2,p_3)$.

{\em Case 3: $\hat w=(ab,ab,e)$, $\hat w=(ba,ba,e)$, $\hat w=(ab,ba,e)$, $\hat w=(ba,ab,e)$ and permutations.}
We observe that $(-1)^{\varepsilon} \langle \Delta\vec\omega, w(\vec{\rho}) \rangle$ is invariant under $\varepsilon\leftrightarrow \varepsilon+1$, $ab\leftrightarrow ba$.

From this, we obtain
\begin{gather*}
 \sum_{\varepsilon\in\ZZ/2}\Theta_{\nu,\delta,(\varepsilon,\varepsilon)}^{(\bos_{(ab,ab,e)})} +\Theta_{\nu,\delta,(\varepsilon,\varepsilon)}^{(\bos_{(ba,ba,e)})}
+ \Theta_{\nu,\delta,(\varepsilon,1+\varepsilon)}^{(\bos_{(ab,ba,e)})} +\Theta_{\nu,\delta,
(1+\varepsilon,\varepsilon)}^{(\bos_{(ba,ab,e)})} \\
 \qquad = 2 \left(\frac{1}{p_1}+\frac{1}{p_2}\right)^{1-\nu}
 \theta^{1-\nu}_{m,m\delta+\bar p_1+ \bar p_2} \sum_{\epsilon,\epsilon'\in \{1,-1\}} \epsilon \epsilon' \theta^1_{m,m\delta+ \epsilon\bar p_1+ \epsilon'\bar p_2+ 2\bar p_3}
\end{gather*}
and similarly
\begin{gather*}
\sum_{\varepsilon\in\ZZ/2}\Theta_{\nu,\delta,(\varepsilon,1+\varepsilon)}^{(\bos_{(ab,ab,e)})} +\Theta_{\nu,\delta,(\varepsilon,1+\varepsilon)}^{(\bos_{(ba,ba,e)})}
+ \Theta_{\nu,\delta,(\varepsilon,\varepsilon)}^{(\bos_{(ab,ba,e)})} +\Theta_{\nu,\delta,(\varepsilon,\varepsilon)}^{(\bos_{(ba,ab,e)})} \\
 \qquad = 2\left(\frac{1}{p_1}-\frac{1}{p_2}\right)^{1-\nu}
 \theta^{1-\nu}_{m,m\delta+\bar p_1-\bar p_2} \sum_{\epsilon,\epsilon'\in \{1,-1\}} \epsilon \epsilon' \theta^1_{m,m\delta+\epsilon \bar p_1+ \epsilon' \bar p_2+2\bar p_3} .
\end{gather*}

We also have images of the above under even permutations, corresponding to the cases where $w_1=e$ or $w_2=e$.

{\em Case 4: $\hat w=(ab,ab,ab)$, $\hat w=(ba,ba,ba)$.}
Similarly as before, we have
\begin{align*}
& \Theta_{\nu,\delta,(\varepsilon,\varepsilon)}^{(\bos_{(ab,ab,ab)})}= \Theta_{\nu,\delta,(\varepsilon,\varepsilon)}^{(\bos_{(ba,ba,ba)})} \left(\frac{(-1)^{\varepsilon+1}({\bar p_1}+{\bar p_2})-{\bar p_3}}{m}\right)^{1-\nu}\\
&\phantom{ \Theta_{\nu,\delta,(\varepsilon,\varepsilon)}^{(\bos_{(ab,ab,ab)})}= }{}\times
 \theta^{1-\nu}_{m,m\delta+(-1)^{\varepsilon}(\bar p_1+\bar p_2)+\bar p_3} \theta^1_{m,m\delta+ (-1)^{\varepsilon+1} (\bar p_1+\bar p_2) -\bar p_3},
\\
&\Theta_{\nu,\delta,(\varepsilon,\varepsilon+1)}^{(\bos_{(ab,ba,ab)})}=\Theta_{\nu,\delta,(\varepsilon,\varepsilon+1)}^{(\bos_{(ba,ab,ba)})}=
- \left(\frac{(-1)^{\varepsilon+1}({\bar p_1}+{\bar p_2})-{\bar p_3}}{m}\right)^{1-\nu}\\
&\phantom{\Theta_{\nu,\delta,(\varepsilon,\varepsilon+1)}^{(\bos_{(ab,ba,ab)})}=\Theta_{\nu,\delta,(\varepsilon,\varepsilon+1)}^{(\bos_{(ba,ab,ba)})}=}{}\times
 \theta^{1-\nu}_{m,m\delta+(-1)^{\varepsilon}(\bar p_1+\bar p_2)+\bar p_3} \theta^1_{m,m\delta+ (-1)^{\varepsilon+1} (\bar p_1-\bar p_2) -\bar p_3},
\\
&\Theta_{\nu,\delta,(\varepsilon+1,\varepsilon)}^{(\bos_{(ba,ab,ab)})}=\Theta_{\nu,\delta,(\varepsilon+1,\varepsilon)}^{(\bos_{(ab,ba,ba)})}=
- \left(\frac{(-1)^{\varepsilon+1}({\bar p_1}+{\bar p_2})-{\bar p_3}}{m}\right)^{1-\nu}\\
&\phantom{\Theta_{\nu,\delta,(\varepsilon+1,\varepsilon)}^{(\bos_{(ba,ab,ab)})}=\Theta_{\nu,\delta,(\varepsilon+1,\varepsilon)}^{(\bos_{(ab,ba,ba)})}=}{}\times
 \theta^{1-\nu}_{m,m\delta+(-1)^{\varepsilon}(\bar p_1+\bar p_2)+\bar p_3} \theta^1_{m,m\delta+ (-1)^{\varepsilon+1} (-\bar p_1+\bar p_2) -\bar p_3},
\\
&\Theta_{\nu,\delta,(\varepsilon+1,\varepsilon)}^{(\bos_{(ab,ab,ba)})}=\Theta_{\nu,\delta,(\varepsilon+1,\varepsilon)}^{(\bos_{(ba,ba,ab)})}=
- \left(\frac{(-1)^{\varepsilon+1}({\bar p_1}-{\bar p_2})-{\bar p_3}}{m}\right)^{1-\nu}\\
&\phantom{\Theta_{\nu,\delta,(\varepsilon+1,\varepsilon)}^{(\bos_{(ab,ab,ba)})}=\Theta_{\nu,\delta,(\varepsilon+1,\varepsilon)}^{(\bos_{(ba,ba,ab)})}=}{}\times
\theta^{1-\nu}_{m,m\delta+ (-1)^{\varepsilon}({\bar p_1}-{\bar p_2})+{\bar p_3}} \theta^1_{m,m\delta+ (-1)^{\varepsilon} (\bar p_1-\bar p_2) -\bar p_3},\\
& \Theta_{\nu,\delta,(\varepsilon,\varepsilon+1)}^{(\bos_{(ab,ab,ab)})}= \Theta_{\nu,\delta,(\varepsilon,\varepsilon+1)}^{(\bos_{(ba,ba,ba)})} = - \left(\frac{(-1)^{\varepsilon+1}({\bar p_1}-{\bar p_2})-{\bar p_3}}{m}\right)^{1-\nu}\\
&\phantom{ \Theta_{\nu,\delta,(\varepsilon,\varepsilon+1)}^{(\bos_{(ab,ab,ab)})}= \Theta_{\nu,\delta,(\varepsilon,\varepsilon+1)}^{(\bos_{(ba,ba,ba)})} =}{}\times
 \theta^{1-\nu}_{m,m\delta+(-1)^{\varepsilon}(\bar p_1-\bar p_2)+\bar p_3} \theta^1_{m,m\delta+ (-1)^{\varepsilon+1} (\bar p_1-\bar p_2) -\bar p_3}, \\
& \Theta_{\nu,\delta,(\varepsilon,\varepsilon)}^{(\bos_{(ba,ab,ab)})}= \Theta_{\nu,\delta,(\varepsilon,\varepsilon)}^{(\bos_{(ab,ba,ba)})} =- \left(\frac{(-1)^{\varepsilon}({\bar p_1}-{\bar p_2})-{\bar p_3}}{m}\right)^{1-\nu}\\
&\phantom{\Theta_{\nu,\delta,(\varepsilon,\varepsilon)}^{(\bos_{(ba,ab,ab)})}= \Theta_{\nu,\delta,(\varepsilon,\varepsilon)}^{(\bos_{(ab,ba,ba)})} =}{}\times
 \theta^{1-\nu}_{m,m\delta+(-1)^{\varepsilon+1}(\bar p_1-\bar p_2)+ \bar p_3} \theta^1_{m,m\delta+ (-1)^{\varepsilon+1} (\bar p_1+\bar p_2) -\bar p_3},
 \\
& \Theta_{\nu,\delta,(\varepsilon,\varepsilon)}^{(\bos_{(ab,ab,ba)})}= \Theta_{\nu,\delta,(\varepsilon,\varepsilon)}^{(\bos_{(ba,ba,ab)})} = \left(\frac{(-1)^{\varepsilon}({\bar p_1}+{\bar p_2})-{\bar p_3}}{m}\right)^{1-\nu}\\
&\phantom{ \Theta_{\nu,\delta,(\varepsilon,\varepsilon)}^{(\bos_{(ab,ab,ba)})}= \Theta_{\nu,\delta,(\varepsilon,\varepsilon)}^{(\bos_{(ba,ba,ab)})} =}{}\times
 \theta^{1-\nu}_{m,m\delta+(-1)^{\varepsilon+1}(\bar p_1+\bar p_2)+\bar p_3} \theta^1_{m,m\delta+ (-1)^{\varepsilon+1} (\bar p_1+\bar p_2) -\bar p_3} .
\end{align*}

Summing up, we get
\begin{gather*}
\Theta_{\nu,\delta,(\varepsilon,\varepsilon)}^{(\bos_{(ab,ab,ab)})} + \Theta_{\nu,\delta,(\varepsilon+1,\varepsilon)}^{(\bos_{(ba,ab,ab)})}+ \Theta_{\nu,\delta,(\varepsilon,\varepsilon+1)}^{(\bos_{(ab,ba,ab)})} +\Theta_{\nu,\delta,(\varepsilon+1,\varepsilon+1)}^{(\bos_{(ab,ab,ba)})} + (a\leftrightarrow b) \\
\qquad= 2\left(\frac{1}{p_1}+\frac{1}{p_2} +(-1)^{\varepsilon}\frac{1}{p_3}\right)^{1-\nu} \theta^{1-\nu}_{m,m\delta+(\bar p_1+\bar p_2)+(-1)^{\varepsilon}\bar p_3} \sum_{\epsilon, \epsilon'} \epsilon\epsilon' \theta^1_{m,m\delta+\epsilon \bar p_1+\epsilon '\bar p_2 -\bar p_3},\\
\Theta_{\nu,\delta,(\varepsilon,\varepsilon)}^{(\bos_{(ab,ba,ab)})} + \Theta_{\nu,\delta,(\varepsilon+1,\varepsilon)}^{(\bos_{(ba,ba,ab)})}+ \Theta_{\nu,\delta,(\varepsilon,\varepsilon+1)}^{(\bos_{(ab,ab,ab)})} +\Theta_{\nu,\delta,(\varepsilon+1,\varepsilon+1)}^{(\bos_{(ab,ba,ba)})} + (a\leftrightarrow b) \\
\qquad= 2
 \left(\frac{1}{p_1}-\frac{1}{p_2} +(-1)^{\varepsilon}\frac{1}{p_3}\right)^{1-\nu}
 \theta^{1-\nu}_{m,m\delta+(\bar p_1-\bar p_2)+(-1)^{\varepsilon}\bar p_3} \sum_{\epsilon, \epsilon'} \epsilon\epsilon' \theta^1_{m,m\delta+\epsilon \bar p_1+\epsilon' \bar p_2 -\bar p_3} .
\end{gather*}

Finally, summing up the contributions from all the above four cases, we define a set ${\cal R} \subset \ZZ/2m$, with
\[
{\cal R} = {\cal R}_0+{\cal R}_1+{\cal R}_2+{\cal R}_3
\]
and
\begin{gather*}
 {\cal R}_0 = \{0\},\qquad
 {\cal R}_1 = {\bf P}^+ \{\bar p_1\}, \qquad
 {\cal R}_2 = {\bf P}^+ \{\bar p_1+\bar p_2, \bar p_1-\bar p_2\},\\
 {\cal R}_3 = {\bf P}^+ \{\bar p_1+\bar p_2-\bar p_3, -\bar p_1-\bar p_2+\bar p_3\}.
\end{gather*}

For each $r\in {\cal R}$, we set $a_i^{(r)}:= 2-|r_i|$ if $r=\sum_i r_i \bar p_i$. For instance, we have \[\bigl(a_1^{(\bar p_1+\bar p_2)}, a_2^{(\bar p_1+\bar p_2)},a_3^{(\bar p_1+\bar p_2)}\bigr) =\bigl(a_1^{(\bar p_1-\bar p_2)}, a_2^{(\bar p_1-\bar p_2)},a_3^{(\bar p_1+\bar p_2)}\bigr) = (1,1,2).\]

Using the above definition, we can write
\begin{align*}
 \sum\limits_{\hat{w}\in W_+^{\otimes 3}} {\tilde \Theta_{\nu,\delta}}^{(\bos_{\hat{w}})} &= 2\sum_{r\in {\cal R}} \left(\frac{r}{m}\right)^{1-\nu} \theta^{1-\nu}_{m,m\delta+r} \sum_{\epsilon, \epsilon} \epsilon\epsilon' \theta^1_{m,m\delta+\epsilon \bar p_1 a_1^{(r)}
 +\epsilon' \bar p_2 a_2^{(r)} + \bar p_3 a_3^{(r)}} \\
 &=\sum_{r\in {\cal R}} \left(\frac{r}{m}\right)^{1-\nu} \theta^{1-\nu}_{m,m\delta+r} \sum_{\varepsilon_1,\varepsilon_2,\varepsilon_3\in \ZZ/2} (-1)^{\sum_i \varepsilon_i} \theta^1_{m,m\delta+\sum_i (-1)^{ \varepsilon_i}\bar p_i a_i^{(r)}
 } .
\end{align*}

From $(p_i,p_j)=1$, we see that $\Omega_m(\bar p_i)$ has precisely one non-zero entry in each row, since
\[
r+r' \equiv 0 (2\bar p_i),\qquad r-r'\equiv 0 (2 p_i)
\]
has a unique solution in $\ZZ/2m$ for $r'$ for any given $r\in \ZZ/2m$.

In particular, one can show that
\[
\sum_{r'\in \ZZ/2m} (\Omega_m(\bar p_i))_{m\delta +\sum_j (-1)^{ \varepsilon_j}\bar p_j a_j^{(r)},r'} X_{r'} =
X_{m\delta +\sum_j (-1)^{ \varepsilon_j+\delta_{i,j}+1}\bar p_j a_j^{(r)}}
\]
for all $\delta, a_j^{(r)}\in \ZZ$. Consider the representation of the metapletic group $\widetilde{SL}_2(\ZZ)$ corresponding to the subgroup $K=\{1,\bar p_1, \bar p_2,\bar p_3\}$ of the group of exact divisors. This representation is irreducible when all three $p_i$ are square-free. We have
\begin{align*}
 \bigl(P^{m+K}\theta_m\bigr)_{m\delta +\bar p_1 a_1+\bar p_2 a_2 +\bar p_3 a_3 }={}& \frac{1}{4}(\theta_{m,m\delta +\bar p_1 a_1+\bar p_2 a_2 + \bar p_3 a_3 } +
\theta_{m,m\delta +\bar p_1 a_1 -\bar p_2 a_2 - \bar p_3 a_3 }\\
&+\theta_{m,m\delta -\bar p_1 a_1+\bar p_2 a_2 - \bar p_3 a_3 }+\theta_{m,m\delta -\bar p_1 a_1-\bar p_2 a_2 + \bar p_3 a_3 }).
\end{align*}
Again using $\theta^1_{m,r}= -\theta^1_{m,-r}$, we see that
\begin{gather*}
\theta^{1,m+K}_{m\delta +\sum_i \bar p_i a_i^{(r)} }
:= \bigl(P^{m+K}\theta^1_m\bigr)_{m\delta+\sum_i \bar p_i a_i^{(r)} } = \frac{1}{4}\sum_{\epsilon, \epsilon} \epsilon\epsilon' \theta^1_{m,m\delta+\epsilon \bar p_1 a_1^{(r)}
 +\epsilon' \bar p_2 a_2^{(r)} + \bar p_3 a_3^{(r)}} .
\end{gather*}
As a result, we obtain the following expression
\begin{gather*}
 \sum\limits_{\hat{w}\in W_+^{\otimes 3}} {\tilde \Theta_{\nu,\delta}}^{(\bos_{\hat{w}})} = 8\sum_{r\in {\cal R}} \left(\frac{r}{m}\right)^{1-\nu} \theta^{1-\nu}_{m,m\delta+r} \theta^{1,m+K}_{m,m\delta+\sum_{i} \bar p_i a_i^{(r)}
 } .
\end{gather*}

\section{Tables}
In this appendix, we collect tables with computational data for
the examples presented in Section~\ref{sec:examples}.
Each of these tables is organized following in blocks with the same format, where each block specifies the contribution to the function
$\widecheck{Z}$ which comes from a generalised $A_2$ false theta function.
We remind the reader that the definition of a generalised
$A_2$ false theta function can be found in Section~\ref{sec:gena2};
this function is a building block for the companion function $\widehat{Z}$
and is associated to a set $\CS=\tilde\CS_{\hat{w}}$, which is in turn
determined with respect to $\widehat{Z}$ by a triplet of Weyl group
elements $\hat{w}$.

Symbolically, each block is organized in the following way

\begin{table}[h!]\centering \renewcommand{\arraystretch}{1.2}
\centering
\begin{tabular}{|lll|}
\hline
$\left(w_1,w_2,w_3\right)$ & $\bal^{(1)}_{1}$ & $\bal^{(2)}_{1}$\\
$\left(s_1,s_2\right)$ & $\bal^{(1)}_{ab}$ & $\bal^{(2)}_{ab}$\\
$\left(k_1,k_2\right)$ & $\bal^{(1)}_{ba}$ & $\bal^{(2)}_{ba}$\\\hline
\end{tabular}
\end{table}
\noindent using again the same notation (\ref{eq:abaction}) for Weyl group elements.
The first column contains the
triplet of Weyl elements $\hat{w}$ and the vectors $\vec{s}$ and
$\vec{k}$, while the second and third columns contain the values of
$\bal^{(1)}_{w}$, $\bal^{(2)}_{w}$, with $w$ restricted to elements of the rotation subgroup $W_+\subset W$.
\begin{landscape}
\begin{table}[th]\centering\small
 \begin{tabular}{|lll|lll|lll|}
 \hline
$(e, e, e)$ & $\bigl(0, -\frac{83}{140}\bigr)$ & $\bigl(1, -\frac{83}{140}\bigr)$ & $(e, e, a)$ & $\bigl(-\frac{1}{7}, -\frac{43}{140}\bigr)$ & $\bigl(\frac{8}{7}, -\frac{103}{140}\bigr)$ & $(e, e, b)$ & $\bigl(\frac{1}{7}, -\frac{103}{140}\bigr)$ & $\bigl(\frac{6}{7}, -\frac{43}{140}\bigr)$
\\
$(83, 83)$ & $\bigl(\frac{83}{140}, -\frac{83}{140}\bigr)$ & $\bigl(\frac{57}{140}, \frac{83}{70}\bigr)$ & $(43, 103)$ & $\bigl(\frac{83}{140}, -\frac{103}{140}\bigr)$ & $\bigl(\frac{57}{140}, \frac{73}{70}\bigr)$ & $(103, 43)$ & $\bigl(\frac{9}{20}, -\frac{43}{140}\bigr)$ & $\bigl(\frac{11}{20}, \frac{73}{70}\bigr)$
\\
$(0, 0)$ & $\bigl(-\frac{83}{140}, \frac{83}{70}\bigr)$ & $\bigl(\frac{223}{140}, -\frac{83}{140}\bigr)$ & $(0, 0)$ & $\bigl(-\frac{9}{20}, \frac{73}{70}\bigr)$ & $\bigl(\frac{29}{20}, -\frac{43}{140}\bigr)$ & $(0, 0)$ & $\bigl(-\frac{83}{140}, \frac{73}{70}\bigr)$ & $\bigl(\frac{223}{140}, -\frac{103}{140}\bigr)$ \\\hline
$(e, e, ab)$ & $\bigl(\frac{1}{7}, -\frac{83}{140}\bigr)$ & $\bigl(\frac{6}{7}, -\frac{23}{140}\bigr)$ & $(e, e, ba)$ & $\bigl(-\frac{1}{7}, -\frac{23}{140}\bigr)$ & $\bigl(\frac{8}{7}, -\frac{83}{140}\bigr)$ & $(e, e, aba)$ & $\bigl(0, -\frac{43}{140}\bigr)$ & $\bigl(1, -\frac{43}{140}\bigr)$
\\
$(83, 23)$ & $\bigl(\frac{43}{140}, -\frac{23}{140}\bigr)$ & $\bigl(\frac{97}{140}, \frac{53}{70}\bigr)$ & $(23, 83)$ & $\bigl(\frac{9}{20}, -\frac{83}{140}\bigr)$ & $\bigl(\frac{11}{20}, \frac{53}{70}\bigr)$ & $(43, 43)$ & $\bigl(\frac{43}{140}, -\frac{43}{140}\bigr)$ & $\bigl(\frac{97}{140}, \frac{43}{70}\bigr)$
\\
$(0, 0)$ & $\bigl(-\frac{9}{20}, \frac{53}{70}\bigr)$ & $\bigl(\frac{29}{20}, -\frac{83}{140}\bigr)$ & $(0, 0)$ & $\bigl(-\frac{43}{140}, \frac{53}{70}\bigr)$ & $\bigl(\frac{183}{140}, -\frac{23}{140}\bigr)$ & $(0, 0)$ & $\bigl(-\frac{43}{140}, \frac{43}{70}\bigr)$ & $\bigl(\frac{183}{140}, -\frac{43}{140}\bigr)$ \\\hline
$(e, a, e)$ & $\bigl(-\frac{1}{5}, -\frac{27}{140}\bigr)$ & $\bigl(\frac{6}{5}, -\frac{111}{140}\bigr)$ & $(e, a, b)$ & $\bigl(-\frac{2}{35}, -\frac{47}{140}\bigr)$ & $\bigl(\frac{37}{35}, -\frac{71}{140}\bigr)$ & $(e, a, ab)$ & $\bigl(-\frac{2}{35}, -\frac{27}{140}\bigr)$ & $\bigl(\frac{37}{35}, -\frac{51}{140}\bigr)$
\\
$(27, 111)$ & $\bigl(\frac{83}{140}, -\frac{111}{140}\bigr)$ & $\bigl(\frac{57}{140}, \frac{69}{70}\bigr)$ & $(47, 71)$ & $\bigl(\frac{9}{20}, -\frac{71}{140}\bigr)$ & $\bigl(\frac{11}{20}, \frac{59}{70}\bigr)$ & $(27, 51)$ & $\bigl(\frac{43}{140}, -\frac{51}{140}\bigr)$ & $\bigl(\frac{97}{140}, \frac{39}{70}\bigr)$
\\
$(0, 0)$ & $\bigl(-\frac{11}{28}, \frac{69}{70}\bigr)$ & $\bigl(\frac{39}{28}, -\frac{27}{140}\bigr)$ & $(0, 0)$ & $\bigl(-\frac{11}{28}, \frac{59}{70}\bigr)$ & $\bigl(\frac{39}{28}, -\frac{47}{140}\bigr)$ & $(0, 0)$ & $\bigl(-\frac{1}{4}, \frac{39}{70}\bigr)$ & $\bigl(\frac{5}{4}, -\frac{27}{140}\bigr)$
\\\hline
$(e, b, e)$ & $\bigl(\frac{1}{5}, -\frac{111}{140}\bigr)$ & $\bigl(\frac{4}{5}, -\frac{27}{140}\bigr)$ & $(e, b, a)$ & $\bigl(\frac{2}{35}, -\frac{71}{140}\bigr)$ & $\bigl(\frac{33}{35}, -\frac{47}{140}\bigr)$ & $(e, b, ba)$ & $\bigl(\frac{2}{35}, -\frac{51}{140}\bigr)$ & $\bigl(\frac{33}{35}, -\frac{27}{140}\bigr)$
\\
$(111, 27)$ & $\bigl(\frac{11}{28}, -\frac{27}{140}\bigr)$ & $\bigl(\frac{17}{28}, \frac{69}{70}\bigr)$ & $(71, 47)$ & $\bigl(\frac{11}{28}, -\frac{47}{140}\bigr)$ & $\bigl(\frac{17}{28}, \frac{59}{70}\bigr)$ & $(51, 27)$ & $\bigl(\frac{1}{4}, -\frac{27}{140}\bigr)$ & $\bigl(\frac{3}{4}, \frac{39}{70}\bigr)$
\\
$(0, 0)$ & $\bigl(-\frac{83}{140}, \frac{69}{70}\bigr)$ & $\bigl(\frac{223}{140}, -\frac{111}{140}\bigr)$ & $(0, 0)$ & $\bigl(-\frac{9}{20}, \frac{59}{70}\bigr)$ & $\bigl(\frac{29}{20}, -\frac{71}{140}\bigr)$ & $(0, 0)$ & $\bigl(-\frac{43}{140}, \frac{39}{70}\bigr)$ & $\bigl(\frac{183}{140}, -\frac{51}{140}\bigr)$
\\\hline
$(e, ab, a)$ & $\bigl(\frac{2}{35}, -\frac{43}{140}\bigr)$ & $\bigl(\frac{33}{35}, -\frac{19}{140}\bigr)$ & $(e, ba, b)$ & $\bigl(-\frac{2}{35}, -\frac{19}{140}\bigr)$ & $\bigl(\frac{37}{35}, -\frac{43}{140}\bigr)$ & $(e, aba, e)$ & $\bigl(0, -\frac{27}{140}\bigr)$ & $\bigl(1, -\frac{27}{140}\bigr)$
\\
$(43, 19)$ & $\bigl(\frac{27}{140}, -\frac{19}{140}\bigr)$ & $\bigl(\frac{113}{140}, \frac{31}{70}\bigr)$ & $(19, 43)$ & $\bigl(\frac{1}{4}, -\frac{43}{140}\bigr)$ & $\bigl(\frac{3}{4}, \frac{31}{70}\bigr)$ & $(27, 27)$ & $\bigl(\frac{27}{140}, -\frac{27}{140}\bigr)$ & $\bigl(\frac{113}{140}, \frac{27}{70}\bigr)$
\\
$(0, 0)$ & $\bigl(-\frac{1}{4}, \frac{31}{70}\bigr)$ & $\bigl(\frac{5}{4}, -\frac{43}{140}\bigr)$ & $(0, 0)$ & $\bigl(-\frac{27}{140}, \frac{31}{70}\bigr)$ & $\bigl(\frac{167}{140}, -\frac{19}{140}\bigr)$ & $(0, 0)$ & $\bigl(-\frac{27}{140}, \frac{27}{70}\bigr)$ & $\bigl(\frac{167}{140}, -\frac{27}{140}\bigr)$
\\\hline
$(a, e, e)$ & $\bigl(-\frac{1}{4}, -\frac{13}{140}\bigr)$ & $\bigl(\frac{5}{4}, -\frac{59}{70}\bigr)$ & $(a, e, b)$ & $\bigl(-\frac{3}{28}, -\frac{33}{140}\bigr)$ & $\bigl(\frac{31}{28}, -\frac{39}{70}\bigr)$ & $(a, e, ab)$ & $\bigl(-\frac{3}{28}, -\frac{13}{140}\bigr)$ & $\bigl(\frac{31}{28}, -\frac{29}{70}\bigr)$
\\
$(13, 118)$ & $\bigl(\frac{83}{140}, -\frac{59}{70}\bigr)$ & $\bigl(\frac{57}{140}, \frac{131}{140}\bigr)$ & $(33, 78)$ & $\bigl(\frac{9}{20}, -\frac{39}{70}\bigr)$ & $\bigl(\frac{11}{20}, \frac{111}{140}\bigr)$ & $(13, 58)$ & $\bigl(\frac{43}{140}, -\frac{29}{70}\bigr)$ & $\bigl(\frac{97}{140}, \frac{71}{140}\bigr)$
\\
$(0, 0)$ & $\bigl(-\frac{12}{35}, \frac{131}{140}\bigr)$ & $\bigl(\frac{47}{35}, -\frac{13}{140}\bigr)$ & $(0, 0)$ & $\bigl(-\frac{12}{35}, \frac{111}{140}\bigr)$ & $\bigl(\frac{47}{35}, -\frac{33}{140}\bigr)$ & $(0, 0)$ & $\bigl(-\frac{1}{5}, \frac{71}{140}\bigr)$ & $\bigl(\frac{6}{5}, -\frac{13}{140}\bigr)$
\\\hline
$(a, b, e)$ & $\bigl(-\frac{1}{20}, -\frac{41}{140}\bigr)$ & $\bigl(\frac{21}{20}, -\frac{31}{70}\bigr)$ & $(a, b, a)$ & $\bigl(-\frac{27}{140}, -\frac{1}{140}\bigr)$ & $\bigl(\frac{167}{140}, -\frac{41}{70}\bigr)$ & $(a, b, b)$ & $\bigl(\frac{13}{140}, -\frac{61}{140}\bigr)$ & $\bigl(\frac{127}{140}, -\frac{11}{70}\bigr)$
\\
$(41, 62)$ & $\bigl(\frac{11}{28}, -\frac{31}{70}\bigr)$ & $\bigl(\frac{17}{28}, \frac{103}{140}\bigr)$ & $(1, 82)$ & $\bigl(\frac{11}{28}, -\frac{41}{70}\bigr)$ & $\bigl(\frac{17}{28}, \frac{83}{140}\bigr)$ & $(61, 22)$ & $\bigl(\frac{1}{4}, -\frac{11}{70}\bigr)$ & $\bigl(\frac{3}{4}, \frac{83}{140}\bigr)$
\\
$(0, 0)$ & $\bigl(-\frac{12}{35}, \frac{103}{140}\bigr)$ & $\bigl(\frac{47}{35}, -\frac{41}{140}\bigr)$ & $(0, 0)$ & $\bigl(-\frac{1}{5}, \frac{83}{140}\bigr)$ & $\bigl(\frac{6}{5}, -\frac{1}{140}\bigr)$ & $(0, 0)$ & $\bigl(-\frac{12}{35}, \frac{83}{140}\bigr)$ & $\bigl(\frac{47}{35}, -\frac{61}{140}\bigr)$
\\\hline
$(a, b, ab)$ & $\bigl(\frac{13}{140}, -\frac{41}{140}\bigr)$ & $\bigl(\frac{127}{140}, -\frac{1}{70}\bigr)$ & $(a, b, aba)$ & $\bigl(-\frac{1}{20}, -\frac{1}{140}\bigr)$ & $\bigl(\frac{21}{20}, -\frac{11}{70}\bigr)$ & $(a, ab, e)$ & $\bigl(-\frac{1}{20}, -\frac{13}{140}\bigr)$ & $\bigl(\frac{21}{20}, -\frac{17}{70}\bigr)$
\\
$(41, 2)$ & $\bigl(\frac{3}{28}, -\frac{1}{70}\bigr)$ & $\bigl(\frac{25}{28}, \frac{43}{140}\bigr)$ & $(1, 22)$ & $\bigl(\frac{3}{28}, -\frac{11}{70}\bigr)$ & $\bigl(\frac{25}{28}, \frac{23}{140}\bigr)$ & $(13, 34)$ & $\bigl(\frac{27}{140}, -\frac{17}{70}\bigr)$ & $\bigl(\frac{113}{140}, \frac{47}{140}\bigr)$
\\
$(0, 0)$ & $\bigl(-\frac{1}{5}, \frac{43}{140}\bigr)$ & $\bigl(\frac{6}{5}, -\frac{41}{140}\bigr)$ & $(0, 0)$ & $\bigl(-\frac{2}{35}, \frac{23}{140}\bigr)$ & $\bigl(\frac{37}{35}, -\frac{1}{140}\bigr)$ & $(0, 0)$ & $\bigl(-\frac{1}{7}, \frac{47}{140}\bigr)$ & $\bigl(\frac{8}{7}, -\frac{13}{140}\bigr)$
\\\hline
$(b, e, e)$ & $\bigl(\frac{1}{4}, -\frac{59}{70}\bigr)$ & $\bigl(\frac{3}{4}, -\frac{13}{140}\bigr)$ & $(b, e, a)$ & $\bigl(\frac{3}{28}, -\frac{39}{70}\bigr)$ & $\bigl(\frac{25}{28}, -\frac{33}{140}\bigr)$ & $(b, e, ba)$ & $\bigl(\frac{3}{28}, -\frac{29}{70}\bigr)$ & $\bigl(\frac{25}{28}, -\frac{13}{140}\bigr)$
\\
$(118, 13)$ & $\bigl(\frac{12}{35}, -\frac{13}{140}\bigr)$ & $\bigl(\frac{23}{35}, \frac{131}{140}\bigr)$ & $(78, 33)$ & $\bigl(\frac{12}{35}, -\frac{33}{140}\bigr)$ & $\bigl(\frac{23}{35}, \frac{111}{140}\bigr)$ & $(58, 13)$ & $\bigl(\frac{1}{5}, -\frac{13}{140}\bigr)$ & $\bigl(\frac{4}{5}, \frac{71}{140}\bigr)$
\\
$(0, 0)$ & $\bigl(-\frac{83}{140}, \frac{131}{140}\bigr)$ & $\bigl(\frac{223}{140}, -\frac{59}{70}\bigr)$ & $(0, 0)$ & $\bigl(-\frac{9}{20}, \frac{111}{140}\bigr)$ & $\bigl(\frac{29}{20}, -\frac{39}{70}\bigr)$ & $(0, 0)$ & $\bigl(-\frac{43}{140}, \frac{71}{140}\bigr)$ & $\bigl(\frac{183}{140}, -\frac{29}{70}\bigr)$
\\\hline
$(b, a, e)$ & $\bigl(\frac{1}{20}, -\frac{31}{70}\bigr)$ & $\bigl(\frac{19}{20}, -\frac{41}{140}\bigr)$ & $(b, a, a)$ & $\bigl(-\frac{13}{140}, -\frac{11}{70}\bigr)$ & $\bigl(\frac{153}{140}, -\frac{61}{140}\bigr)$ & $(b, a, b)$ & $\bigl(\frac{27}{140}, -\frac{41}{70}\bigr)$ & $\bigl(\frac{113}{140}, -\frac{1}{140}\bigr)$
\\
$(62, 41)$ & $\bigl(\frac{12}{35}, -\frac{41}{140}\bigr)$ & $\bigl(\frac{23}{35}, \frac{103}{140}\bigr)$ & $(22, 61)$ & $\bigl(\frac{12}{35}, -\frac{61}{140}\bigr)$ & $\bigl(\frac{23}{35}, \frac{83}{140}\bigr)$ & $(82, 1)$ & $\bigl(\frac{1}{5}, -\frac{1}{140}\bigr)$ & $\bigl(\frac{4}{5}, \frac{83}{140}\bigr)$
\\
$(0, 0)$ & $\bigl(-\frac{11}{28}, \frac{103}{140}\bigr)$ & $\bigl(\frac{39}{28}, -\frac{31}{70}\bigr)$ & $(0, 0)$ & $\bigl(-\frac{1}{4}, \frac{83}{140}\bigr)$ & $\bigl(\frac{5}{4}, -\frac{11}{70}\bigr)$ & $(0, 0)$ & $\bigl(-\frac{11}{28}, \frac{83}{140}\bigr)$ & $\bigl(\frac{39}{28}, -\frac{41}{70}\bigr)$
\\\hline
$(b, a, ba)$ & $\bigl(-\frac{13}{140}, -\frac{1}{70}\bigr)$ & $\bigl(\frac{153}{140}, -\frac{41}{140}\bigr)$ & $(b, a, aba)$ & $\bigl(\frac{1}{20}, -\frac{11}{70}\bigr)$ & $\bigl(\frac{19}{20}, -\frac{1}{140}\bigr)$ & $(b, ba, e)$ & $\bigl(\frac{1}{20}, -\frac{17}{70}\bigr)$ & $\bigl(\frac{19}{20}, -\frac{13}{140}\bigr)$
\\
$(2, 41)$ & $\bigl(\frac{1}{5}, -\frac{41}{140}\bigr)$ & $\bigl(\frac{4}{5}, \frac{43}{140}\bigr)$ & $(22, 1)$ & $\bigl(\frac{2}{35}, -\frac{1}{140}\bigr)$ & $\bigl(\frac{33}{35}, \frac{23}{140}\bigr)$ & $(34, 13)$ & $\bigl(\frac{1}{7}, -\frac{13}{140}\bigr)$ & $\bigl(\frac{6}{7}, \frac{47}{140}\bigr)$
\\
$(0, 0)$ & $\bigl(-\frac{3}{28}, \frac{43}{140}\bigr)$ & $\bigl(\frac{31}{28}, -\frac{1}{70}\bigr)$ & $(0, 0)$ & $\bigl(-\frac{3}{28}, \frac{23}{140}\bigr)$ & $\bigl(\frac{31}{28}, -\frac{11}{70}\bigr)$ & $(0, 0)$ & $\bigl(-\frac{27}{140}, \frac{47}{140}\bigr)$ & $\bigl(\frac{167}{140}, -\frac{17}{70}\bigr)$
\\\hline
$(ab, a, e)$ & $\bigl(\frac{1}{20}, -\frac{27}{140}\bigr)$ & $\bigl(\frac{19}{20}, -\frac{3}{70}\bigr)$ & $(ba, b, e)$ & $\bigl(-\frac{1}{20}, -\frac{3}{70}\bigr)$ & $\bigl(\frac{21}{20}, -\frac{27}{140}\bigr)$ & $(aba, e, e)$ & $\bigl(0, -\frac{13}{140}\bigr)$ & $\bigl(1, -\frac{13}{140}\bigr)$
\\
$(27, 6)$ & $\bigl(\frac{13}{140}, -\frac{3}{70}\bigr)$ & $\bigl(\frac{127}{140}, \frac{33}{140}\bigr)$ & $(6, 27)$ & $\bigl(\frac{1}{7}, -\frac{27}{140}\bigr)$ & $\bigl(\frac{6}{7}, \frac{33}{140}\bigr)$ & $(13, 13)$ & $\bigl(\frac{13}{140}, -\frac{13}{140}\bigr)$ & $\bigl(\frac{127}{140}, \frac{13}{70}\bigr)$
\\
$(0, 0)$ & $\bigl(-\frac{1}{7}, \frac{33}{140}\bigr)$ & $\bigl(\frac{8}{7}, -\frac{27}{140}\bigr)$ & $(0, 0)$ & $\bigl(-\frac{13}{140}, \frac{33}{140}\bigr)$ & $\bigl(\frac{153}{140}, -\frac{3}{70}\bigr)$ & $(0, 0)$ & $\bigl(-\frac{13}{140}, \frac{13}{70}\bigr)$ & $\bigl(\frac{153}{140}, -\frac{13}{140}\bigr)$
\\\hline
\end{tabular}
\caption{$\boldsymbol{\alpha}$ of $M\bigl( -1;\frac{1}{4},\frac{3}{5}, \frac{1}{7} \bigr)$ for values of $\vec{s}$ with components in $\{1,\dots,139\}$.} \label{tab:SstarBriesk}
 \end{table}

\begin{table}[th]\centering\small
\begin{tabular}{|lll|lll|lll|}
 \hline
 $(e, e, e)$ & $\bigl(0, -\frac{19}{20}\bigr)$ & $\bigl(1, -\frac{19}{20}\bigr)$ & $(e, e,aba)$ & $\bigl(0, -\frac{9}{20}\bigr)$ & $\bigl(1, -\frac{9}{20}\bigr)$ & $(e, a, a)$ & $\bigl(-\frac{3}{4}, \frac{11}{20}\bigr)$ & $\bigl(\frac{7}{4}, -\frac{17}{10}\bigr)$
 \\
$(19, 19)$ & $\bigl(\frac{19}{20}, -\frac{19}{20}\bigr)$ & $\bigl(\frac{1}{20}, \frac{19}{10}\bigr)$ & $(9, 9)$ & $\bigl(\frac{9}{20}, -\frac{9}{20}\bigr)$ & $\bigl(\frac{11}{20}, \frac{9}{10}\bigr)$ & $(-11, 34)$ & $\bigl(\frac{19}{20}, -\frac{17}{10}\bigr)$ & $\bigl(\frac{1}{20}, \frac{23}{20}\bigr)$
\\
$(0, 0)$ & $\bigl(-\frac{19}{20}, \frac{19}{10}\bigr)$ & $\bigl(\frac{39}{20}, -\frac{19}{20}\bigr)$ & $(0, 0)$ & $\bigl(-\frac{9}{20}, \frac{9}{10}\bigr)$ & $\bigl(\frac{29}{20}, -\frac{9}{20}\bigr)$ & $(0, 0)$ & $\bigl(-\frac{1}{5}, \frac{23}{20}\bigr)$ & $\bigl(\frac{6}{5}, \frac{11}{20}\bigr)$
\\\hline
$(e, a, ab)$ & $\bigl(-\frac{1}{4}, \frac{1}{20}\bigr)$ & $\bigl(\frac{5}{4}, -\frac{7}{10}\bigr)$ & $(e, b, b)$ & $\bigl(\frac{3}{4}, -\frac{17}{10}\bigr)$ & $\bigl(\frac{1}{4}, \frac{11}{20}\bigr)$ & $(e, b, ba)$ & $\bigl(\frac{1}{4}, -\frac{7}{10}\bigr)$ & $\bigl(\frac{3}{4}, \frac{1}{20}\bigr)$
\\
$(-1, 14)$ & $\bigl(\frac{9}{20}, -\frac{7}{10}\bigr)$ & $\bigl(\frac{11}{20}, \frac{13}{20}\bigr)$ & $(34, -11)$ & $\bigl(\frac{1}{5}, \frac{11}{20}\bigr)$ & $\bigl(\frac{4}{5}, \frac{23}{20}\bigr)$ & $(14, -1)$ & $\bigl(\frac{1}{5}, \frac{1}{20}\bigr)$ & $\bigl(\frac{4}{5}, \frac{13}{20}\bigr)$
\\
$(0, 0)$ & $\bigl(-\frac{1}{5}, \frac{13}{20}\bigr)$ & $\bigl(\frac{6}{5}, \frac{1}{20}\bigr)$ & $(0, 0)$ & $\bigl(-\frac{19}{20}, \frac{23}{20}\bigr)$ & $\bigl(\frac{39}{20}, -\frac{17}{10}\bigr)$ & $(0, 0)$ & $\bigl(-\frac{9}{20}, \frac{13}{20}\bigr)$ & $\bigl(\frac{29}{20}, -\frac{7}{10}\bigr)$
\\\hline
$(e, ab, a)$ & $\bigl(\frac{1}{4}, -\frac{9}{20}\bigr)$ & $\bigl(\frac{3}{4}, \frac{3}{10}\bigr)$ & $(e, ab, ab)$ & $\bigl(\frac{3}{4}, -\frac{19}{20}\bigr)$ & $\bigl(\frac{1}{4}, \frac{13}{10}\bigr)$ & $(e, ba, b)$ & $\bigl(-\frac{1}{4}, \frac{3}{10}\bigr)$ & $\bigl(\frac{5}{4}, -\frac{9}{20}\bigr)$
\\
$(9, -6)$ & $\bigl(-\frac{1}{20}, \frac{3}{10}\bigr)$ & $\bigl(\frac{21}{20}, \frac{3}{20}\bigr)$ & $(19, -26)$ & $\bigl(-\frac{11}{20}, \frac{13}{10}\bigr)$ & $\bigl(\frac{31}{20}, -\frac{7}{20}\bigr)$ & $(-6, 9)$ & $\bigl(\frac{1}{5}, -\frac{9}{20}\bigr)$ & $\bigl(\frac{4}{5}, \frac{3}{20}\bigr)$
\\
$(0, 0)$ & $\bigl(-\frac{1}{5}, \frac{3}{20}\bigr)$ & $\bigl(\frac{6}{5}, -\frac{9}{20}\bigr)$ & $(0, 0)$ & $\bigl(-\frac{1}{5}, -\frac{7}{20}\bigr)$ & $\bigl(\frac{6}{5}, -\frac{19}{20}\bigr)$ & $(0, 0)$ & $\bigl(\frac{1}{20}, \frac{3}{20}\bigr)$ & $\bigl(\frac{19}{20}, \frac{3}{10}\bigr)$
\\\hline
$(e, ba, ba)$ & $\bigl(-\frac{3}{4}, \frac{13}{10}\bigr)$ & $\bigl(\frac{7}{4}, -\frac{19}{20}\bigr)$ & $(e, aba, e)$ & $\bigl(0, \frac{1}{20}\bigr)$ & $\bigl(1, \frac{1}{20}\bigr)$ & $(e, aba, aba)$ & $\bigl(0, \frac{11}{20}\bigr)$ & $\bigl(1, \frac{11}{20}\bigr)$
\\
$(-26, 19)$ & $\bigl(\frac{1}{5}, -\frac{19}{20}\bigr)$ & $\bigl(\frac{4}{5}, -\frac{7}{20}\bigr)$ & $(-1, -1)$ & $\bigl(-\frac{1}{20}, \frac{1}{20}\bigr)$ & $\bigl(\frac{21}{20}, -\frac{1}{10}\bigr)$ & $(-11, -11)$ & $\bigl(-\frac{11}{20}, \frac{11}{20}\bigr)$ & $\bigl(\frac{31}{20}, -\frac{11}{10}\bigr)$
\\
$(0, 0)$ & $\bigl(\frac{11}{20}, -\frac{7}{20}\bigr)$ & $\bigl(\frac{9}{20}, \frac{13}{10}\bigr)$ & $(0, 0)$ & $\bigl(\frac{1}{20}, -\frac{1}{10}\bigr)$ & $\bigl(\frac{19}{20}, \frac{1}{20}\bigr)$ & $(0, 0)$ & $\bigl(\frac{11}{20}, -\frac{11}{10}\bigr)$ & $\bigl(\frac{9}{20}, \frac{11}{20}\bigr)$
\\\hline
$(a, e, e)$ & $\bigl(-\frac{1}{5}, -\frac{11}{20}\bigr)$ & $\bigl(\frac{6}{5}, -\frac{23}{20}\bigr)$ & $(a, e, aba)$ & $\bigl(-\frac{1}{5}, -\frac{1}{20}\bigr)$ & $\bigl(\frac{6}{5}, -\frac{13}{20}\bigr)$ & $(a, a, a)$ & $\bigl(-\frac{19}{20}, \frac{19}{20}\bigr)$ & $\bigl(\frac{39}{20}, -\frac{19}{10}\bigr)$
\\
$(11, 23)$ & $\bigl(\frac{19}{20}, -\frac{23}{20}\bigr)$ & $\bigl(\frac{1}{20}, \frac{17}{10}\bigr)$ & $(1, 13)$ & $\bigl(\frac{9}{20}, -\frac{13}{20}\bigr)$ & $\bigl(\frac{11}{20}, \frac{7}{10}\bigr)$ & $(-19, 38)$ & $\bigl(\frac{19}{20}, -\frac{19}{10}\bigr)$ & $\bigl(\frac{1}{20}, \frac{19}{20}\bigr)$
\\
$(0, 0)$ & $\bigl(-\frac{3}{4}, \frac{17}{10}\bigr)$ & $\bigl(\frac{7}{4}, -\frac{11}{20}\bigr)$ & $(0, 0)$ & $\bigl(-\frac{1}{4}, \frac{7}{10}\bigr)$ & $\bigl(\frac{5}{4}, -\frac{1}{20}\bigr)$ & $(0, 0)$ & $\bigl(0, \frac{19}{20}\bigr)$ & $\bigl(1, \frac{19}{20}\bigr)$
\\\hline
$(a, a, ab)$ & $\bigl(-\frac{9}{20}, \frac{9}{20}\bigr)$ & $\bigl(\frac{29}{20}, -\frac{9}{10}\bigr)$ & $(a, b, b)$ & $\bigl(\frac{11}{20}, -\frac{13}{10}\bigr)$ & $\bigl(\frac{9}{20}, \frac{7}{20}\bigr)$ & $(a, b, ba)$ & $\bigl(\frac{1}{20}, -\frac{3}{10}\bigr)$ & $\bigl(\frac{19}{20}, -\frac{3}{20}\bigr)$
\\
$(-9, 18)$ & $\bigl(\frac{9}{20}, -\frac{9}{10}\bigr)$ & $\bigl(\frac{11}{20}, \frac{9}{20}\bigr)$ & $(26, -7)$ & $\bigl(\frac{1}{5}, \frac{7}{20}\bigr)$ & $\bigl(\frac{4}{5}, \frac{19}{20}\bigr)$ & $(6, 3)$ & $\bigl(\frac{1}{5}, -\frac{3}{20}\bigr)$ & $\bigl(\frac{4}{5}, \frac{9}{20}\bigr)$
\\
$(0, 0)$ & $\bigl(0, \frac{9}{20}\bigr)$ & $\bigl(1, \frac{9}{20}\bigr)$ & $(0, 0)$ & $\bigl(-\frac{3}{4}, \frac{19}{20}\bigr)$ & $\bigl(\frac{7}{4}, -\frac{13}{10}\bigr)$ & $(0, 0)$ & $\bigl(-\frac{1}{4}, \frac{9}{20}\bigr)$ & $\bigl(\frac{5}{4}, -\frac{3}{10}\bigr)$
\\\hline
$(a, ab, a)$ & $\bigl(\frac{1}{20}, -\frac{1}{20}\bigr)$ & $\bigl(\frac{19}{20}, \frac{1}{10}\bigr)$ & $(a, ab, ab)$ & $\bigl(\frac{11}{20}, -\frac{11}{20}\bigr)$ & $\bigl(\frac{9}{20}, \frac{11}{10}\bigr)$ & $(a, ba, b)$ & $\bigl(-\frac{9}{20}, \frac{7}{10}\bigr)$ & $\bigl(\frac{29}{20}, -\frac{13}{20}\bigr)$ \\
$(1, -2)$ & $\bigl(-\frac{1}{20}, \frac{1}{10}\bigr)$ & $\bigl(\frac{21}{20}, -\frac{1}{20}\bigr)$ & $(11, -22)$ & $\bigl(-\frac{11}{20}, \frac{11}{10}\bigr)$ & $\bigl(\frac{31}{20}, -\frac{11}{20}\bigr)$ & $(-14, 13)$ & $\bigl(\frac{1}{5}, -\frac{13}{20}\bigr)$ & $\bigl(\frac{4}{5}, -\frac{1}{20}\bigr)$
\\
$(0, 0)$ & $\bigl(0, -\frac{1}{20}\bigr)$ & $\bigl(1, -\frac{1}{20}\bigr)$ & $(0, 0)$ & $\bigl(0, -\frac{11}{20}\bigr)$ & $\bigl(1, -\frac{11}{20}\bigr)$ & $(0, 0)$ & $\bigl(\frac{1}{4}, -\frac{1}{20}\bigr)$ & $\bigl(\frac{3}{4}, \frac{7}{10}\bigr)$
\\\hline
$(a, ba, ba)$ & $\bigl(-\frac{19}{20}, \frac{17}{10}\bigr)$ & $\bigl(\frac{39}{20}, -\frac{23}{20}\bigr)$ & $(a, aba, e)$ & $\bigl(-\frac{1}{5}, \frac{9}{20}\bigr)$ & $\bigl(\frac{6}{5}, -\frac{3}{20}\bigr)$ & $(a, aba, aba)$ & $\bigl(-\frac{1}{5}, \frac{19}{20}\bigr)$ & $\bigl(\frac{6}{5}, \frac{7}{20}\bigr)$
\\
$(-34, 23)$ & $\bigl(\frac{1}{5}, -\frac{23}{20}\bigr)$ & $\bigl(\frac{4}{5}, -\frac{11}{20}\bigr)$ & $(-9, 3)$ & $\bigl(-\frac{1}{20}, -\frac{3}{20}\bigr)$ & $\bigl(\frac{21}{20}, -\frac{3}{10}\bigr)$ & $(-19, -7)$ & $\bigl(-\frac{11}{20}, \frac{7}{20}\bigr)$ & $\bigl(\frac{31}{20}, -\frac{13}{10}\bigr)$
\\
$(0, 0)$ & $\bigl(\frac{3}{4}, -\frac{11}{20}\bigr)$ & $\bigl(\frac{1}{4}, \frac{17}{10}\bigr)$ & $(0, 0)$ & $\bigl(\frac{1}{4}, -\frac{3}{10}\bigr)$ & $\bigl(\frac{3}{4}, \frac{9}{20}\bigr)$ & $(0, 0)$ & $\bigl(\frac{3}{4}, -\frac{13}{10}\bigr)$ & $\bigl(\frac{1}{4}, \frac{19}{20}\bigr)$
\\\hline
$(b, e, e)$ & $\bigl(\frac{1}{5}, -\frac{23}{20}\bigr)$ & $\bigl(\frac{4}{5}, -\frac{11}{20}\bigr)$ & $(b, e, aba)$ & $\bigl(\frac{1}{5}, -\frac{13}{20}\bigr)$ & $\bigl(\frac{4}{5}, -\frac{1}{20}\bigr)$ & $(b, a, a)$ & $\bigl(-\frac{11}{20}, \frac{7}{20}\bigr)$ & $\bigl(\frac{31}{20}, -\frac{13}{10}\bigr)$
\\
$(23, 11)$ & $\bigl(\frac{3}{4}, -\frac{11}{20}\bigr)$ & $\bigl(\frac{1}{4}, \frac{17}{10}\bigr)$ & $(13, 1)$ & $\bigl(\frac{1}{4}, -\frac{1}{20}\bigr)$ & $\bigl(\frac{3}{4}, \frac{7}{10}\bigr)$ & $(-7, 26)$ & $\bigl(\frac{3}{4}, -\frac{13}{10}\bigr)$ & $\bigl(\frac{1}{4}, \frac{19}{20}\bigr)$
\\
$(0, 0)$ & $\bigl(-\frac{19}{20}, \frac{17}{10}\bigr)$ & $\bigl(\frac{39}{20}, -\frac{23}{20}\bigr)$ & $(0, 0)$ & $\bigl(-\frac{9}{20}, \frac{7}{10}\bigr)$ & $\bigl(\frac{29}{20}, -\frac{13}{20}\bigr)$ & $(0, 0)$ & $\bigl(-\frac{1}{5}, \frac{19}{20}\bigr)$ & $\bigl(\frac{6}{5}, \frac{7}{20}\bigr)$
\\\hline
$(b, a, ab)$ & $\bigl(-\frac{1}{20}, -\frac{3}{20}\bigr)$ & $\bigl(\frac{21}{20}, -\frac{3}{10}\bigr)$ & $(b, b, b)$ & $\bigl(\frac{19}{20}, -\frac{19}{10}\bigr)$ & $\bigl(\frac{1}{20}, \frac{19}{20}\bigr)$ & $(b, b, ba)$ & $\bigl(\frac{9}{20}, -\frac{9}{10}\bigr)$ & $\bigl(\frac{11}{20}, \frac{9}{20}\bigr)$
\\
$(3, 6)$ & $\bigl(\frac{1}{4}, -\frac{3}{10}\bigr)$ & $\bigl(\frac{3}{4}, \frac{9}{20}\bigr)$ & $(38, -19)$ & $\bigl(0, \frac{19}{20}\bigr)$ & $\bigl(1, \frac{19}{20}\bigr)$ & $(18, -9)$ & $\bigl(0, \frac{9}{20}\bigr)$ & $\bigl(1, \frac{9}{20}\bigr)$
\\
$(0, 0)$ & $\bigl(-\frac{1}{5}, \frac{9}{20}\bigr)$ & $\bigl(\frac{6}{5}, -\frac{3}{20}\bigr)$ & $(0, 0)$ & $\bigl(-\frac{19}{20}, \frac{19}{20}\bigr)$ & $\bigl(\frac{39}{20}, -\frac{19}{10}\bigr)$ & $(0, 0)$ & $\bigl(-\frac{9}{20}, \frac{9}{20}\bigr)$ & $\bigl(\frac{29}{20}, -\frac{9}{10}\bigr)$
\\\hline
$(b, ab, a)$ & $\bigl(\frac{9}{20}, -\frac{13}{20}\bigr)$ & $\bigl(\frac{11}{20}, \frac{7}{10}\bigr)$ & $(b, ab, ab)$ & $\bigl(\frac{19}{20}, -\frac{23}{20}\bigr)$ & $\bigl(\frac{1}{20}, \frac{17}{10}\bigr)$ & $(b, ba, b)$ & $\bigl(-\frac{1}{20}, \frac{1}{10}\bigr)$ & $\bigl(\frac{21}{20}, -\frac{1}{20}\bigr)$
\\
$(13, -14)$ & $\bigl(-\frac{1}{4}, \frac{7}{10}\bigr)$ & $\bigl(\frac{5}{4}, -\frac{1}{20}\bigr)$ & $(23, -34)$ & $\bigl(-\frac{3}{4}, \frac{17}{10}\bigr)$ & $\bigl(\frac{7}{4}, -\frac{11}{20}\bigr)$ & $(-2, 1)$ & $\bigl(0, -\frac{1}{20}\bigr)$ & $\bigl(1, -\frac{1}{20}\bigr)$
\\
$(0, 0)$ & $\bigl(-\frac{1}{5}, -\frac{1}{20}\bigr)$ & $\bigl(\frac{6}{5}, -\frac{13}{20}\bigr)$ & $(0, 0)$ & $\bigl(-\frac{1}{5}, -\frac{11}{20}\bigr)$ & $\bigl(\frac{6}{5}, -\frac{23}{20}\bigr)$ & $(0, 0)$ & $\bigl(\frac{1}{20}, -\frac{1}{20}\bigr)$ & $\bigl(\frac{19}{20}, \frac{1}{10}\bigr)$
\\\hline
$(b, ba, ba)$ & $\bigl(-\frac{11}{20}, \frac{11}{10}\bigr)$ & $\bigl(\frac{31}{20}, -\frac{11}{20}\bigr)$ & $(b, aba, e)$ & $\bigl(\frac{1}{5}, -\frac{3}{20}\bigr)$ & $\bigl(\frac{4}{5}, \frac{9}{20}\bigr)$ & $(b, aba, aba)$ & $\bigl(\frac{1}{5}, \frac{7}{20}\bigr)$ & $\bigl(\frac{4}{5}, \frac{19}{20}\bigr)$
\\
$(-22, 11)$ & $\bigl(0, -\frac{11}{20}\bigr)$ & $\bigl(1, -\frac{11}{20}\bigr)$ & $(3, -9)$ & $\bigl(-\frac{1}{4}, \frac{9}{20}\bigr)$ & $\bigl(\frac{5}{4}, -\frac{3}{10}\bigr)$ & $(-7, -19)$ & $\bigl(-\frac{3}{4}, \frac{19}{20}\bigr)$ & $\bigl(\frac{7}{4}, -\frac{13}{10}\bigr)$
\\
$(0, 0)$ & $\bigl(\frac{11}{20}, -\frac{11}{20}\bigr)$ & $\bigl(\frac{9}{20}, \frac{11}{10}\bigr)$ & $(0, 0)$ & $\bigl(\frac{1}{20}, -\frac{3}{10}\bigr)$ & $\bigl(\frac{19}{20}, -\frac{3}{20}\bigr)$ & $(0, 0)$ & $\bigl(\frac{11}{20}, -\frac{13}{10}\bigr)$ & $\bigl(\frac{9}{20}, \frac{7}{20}\bigr)$
\\\hline
\end{tabular}
\end{table}

\begin{table}[th]\centering\small
\begin{tabular}{|lll|lll|lll|}
 \hline
$(ab, e, e)$ & $\bigl(\frac{1}{5}, -\frac{19}{20}\bigr)$ & $\bigl(\frac{4}{5}, -\frac{7}{20}\bigr)$ & $(ab, e, aba)$ & $\bigl(\frac{1}{5}, -\frac{9}{20}\bigr)$ & $\bigl(\frac{4}{5}, \frac{3}{20}\bigr)$ & $(ab, a, a)$ & $\bigl(-\frac{11}{20}, \frac{11}{20}\bigr)$ & $\bigl(\frac{31}{20}, -\frac{11}{10}\bigr)$ \\
$(19, 7)$ & $\bigl(\frac{11}{20}, -\frac{7}{20}\bigr)$ & $\bigl(\frac{9}{20}, \frac{13}{10}\bigr)$ & $(9, -3)$ & $\bigl(\frac{1}{20}, \frac{3}{20}\bigr)$ & $\bigl(\frac{19}{20}, \frac{3}{10}\bigr)$ & $(-11, 22)$ & $\bigl(\frac{11}{20}, -\frac{11}{10}\bigr)$ & $\bigl(\frac{9}{20}, \frac{11}{20}\bigr)$
\\
$(0, 0)$ & $\bigl(-\frac{3}{4}, \frac{13}{10}\bigr)$ & $\bigl(\frac{7}{4}, -\frac{19}{20}\bigr)$ & $(0, 0)$ & $\bigl(-\frac{1}{4}, \frac{3}{10}\bigr)$ & $\bigl(\frac{5}{4}, -\frac{9}{20}\bigr)$ & $(0, 0)$ & $\bigl(0, \frac{11}{20}\bigr)$ & $\bigl(1, \frac{11}{20}\bigr)$
\\\hline
$(ab, a, ab)$ & $\bigl(-\frac{1}{20}, \frac{1}{20}\bigr)$ & $\bigl(\frac{21}{20}, -\frac{1}{10}\bigr)$ & $(ab, b, b)$ & $\bigl(\frac{19}{20}, -\frac{17}{10}\bigr)$ & $\bigl(\frac{1}{20}, \frac{23}{20}\bigr)$ & $(ab, b, ba)$ & $\bigl(\frac{9}{20}, -\frac{7}{10}\bigr)$ & $\bigl(\frac{11}{20}, \frac{13}{20}\bigr)$
\\
$(-1, 2)$ & $\bigl(\frac{1}{20}, -\frac{1}{10}\bigr)$ & $\bigl(\frac{19}{20}, \frac{1}{20}\bigr)$ & $(34, -23)$ & $\bigl(-\frac{1}{5}, \frac{23}{20}\bigr)$ & $\bigl(\frac{6}{5}, \frac{11}{20}\bigr)$ & $(14, -13)$ & $\bigl(-\frac{1}{5}, \frac{13}{20}\bigr)$ & $\bigl(\frac{6}{5}, \frac{1}{20}\bigr)$
\\
$(0, 0)$ & $\bigl(0, \frac{1}{20}\bigr)$ & $\bigl(1, \frac{1}{20}\bigr)$ & $(0, 0)$ & $\bigl(-\frac{3}{4}, \frac{11}{20}\bigr)$ & $\bigl(\frac{7}{4}, -\frac{17}{10}\bigr)$ & $(0, 0)$ & $\bigl(-\frac{1}{4}, \frac{1}{20}\bigr)$ & $\bigl(\frac{5}{4}, -\frac{7}{10}\bigr)$
\\\hline
$(ab, ab, a)$ & $\bigl(\frac{9}{20}, -\frac{9}{20}\bigr)$ & $\bigl(\frac{11}{20}, \frac{9}{10}\bigr)$ & $(ab, ab, ab)$ & $\bigl(\frac{19}{20}, -\frac{19}{20}\bigr)$ & $\bigl(\frac{1}{20}, \frac{19}{10}\bigr)$ & $(ab, ba, b)$ & $\bigl(-\frac{1}{20}, \frac{3}{10}\bigr)$ & $\bigl(\frac{21}{20}, \frac{3}{20}\bigr)$
\\
$(9, -18)$ & $\bigl(-\frac{9}{20}, \frac{9}{10}\bigr)$ & $\bigl(\frac{29}{20}, -\frac{9}{20}\bigr)$ & $(19, -38)$ & $\bigl(-\frac{19}{20}, \frac{19}{10}\bigr)$ & $\bigl(\frac{39}{20}, -\frac{19}{20}\bigr)$ & $(-6, -3)$ & $\bigl(-\frac{1}{5}, \frac{3}{20}\bigr)$ & $\bigl(\frac{6}{5}, -\frac{9}{20}\bigr)$
\\
$(0, 0)$ & $\bigl(0, -\frac{9}{20}\bigr)$ & $\bigl(1, -\frac{9}{20}\bigr)$ & $(0, 0)$ & $\bigl(0, -\frac{19}{20}\bigr)$ & $\bigl(1, -\frac{19}{20}\bigr)$ & $(0, 0)$ & $\bigl(\frac{1}{4}, -\frac{9}{20}\bigr)$ & $\bigl(\frac{3}{4}, \frac{3}{10}\bigr)$
\\\hline
$(ab, ba, ba)$ & $\bigl(-\frac{11}{20}, \frac{13}{10}\bigr)$ & $\bigl(\frac{31}{20}, -\frac{7}{20}\bigr)$ & $(ab, aba, e)$ & $\bigl(\frac{1}{5}, \frac{1}{20}\bigr)$ & $\bigl(\frac{4}{5}, \frac{13}{20}\bigr)$ & $(ab, aba, aba)$ & $\bigl(\frac{1}{5}, \frac{11}{20}\bigr)$ & $\bigl(\frac{4}{5}, \frac{23}{20}\bigr)$
\\
$(-26, 7)$ & $\bigl(-\frac{1}{5}, -\frac{7}{20}\bigr)$ & $\bigl(\frac{6}{5}, -\frac{19}{20}\bigr)$ & $(-1, -13)$ & $\bigl(-\frac{9}{20}, \frac{13}{20}\bigr)$ & $\bigl(\frac{29}{20}, -\frac{7}{10}\bigr)$ & $(-11, -23)$ & $\bigl(-\frac{19}{20}, \frac{23}{20}\bigr)$ & $\bigl(\frac{39}{20}, -\frac{17}{10}\bigr)$
\\
$(0, 0)$ & $\bigl(\frac{3}{4}, -\frac{19}{20}\bigr)$ & $\bigl(\frac{1}{4}, \frac{13}{10}\bigr)$ & $(0, 0)$ & $\bigl(\frac{1}{4}, -\frac{7}{10}\bigr)$ & $\bigl(\frac{3}{4}, \frac{1}{20}\bigr)$ & $(0, 0)$ & $\bigl(\frac{3}{4}, -\frac{17}{10}\bigr)$ & $\bigl(\frac{1}{4}, \frac{11}{20}\bigr)$
\\\hline
$(ba, e, e)$ & $\bigl(-\frac{1}{5}, -\frac{7}{20}\bigr)$ & $\bigl(\frac{6}{5}, -\frac{19}{20}\bigr)$ & $(ba, e, aba)$ & $\bigl(-\frac{1}{5}, \frac{3}{20}\bigr)$ & $\bigl(\frac{6}{5}, -\frac{9}{20}\bigr)$ & $(ba, a, a)$ & $\bigl(-\frac{19}{20}, \frac{23}{20}\bigr)$ & $\bigl(\frac{39}{20}, -\frac{17}{10}\bigr)$
\\
$(7, 19)$ & $\bigl(\frac{3}{4}, -\frac{19}{20}\bigr)$ & $\bigl(\frac{1}{4}, \frac{13}{10}\bigr)$ & $(-3, 9)$ & $\bigl(\frac{1}{4}, -\frac{9}{20}\bigr)$ & $\bigl(\frac{3}{4}, \frac{3}{10}\bigr)$ & $(-23, 34)$ & $\bigl(\frac{3}{4}, -\frac{17}{10}\bigr)$ & $\bigl(\frac{1}{4}, \frac{11}{20}\bigr)$
\\
$(0, 0)$ & $\bigl(-\frac{11}{20}, \frac{13}{10}\bigr)$ & $\bigl(\frac{31}{20}, -\frac{7}{20}\bigr)$ & $(0, 0)$ & $\bigl(-\frac{1}{20}, \frac{3}{10}\bigr)$ & $\bigl(\frac{21}{20}, \frac{3}{20}\bigr)$ & $(0, 0)$ & $\bigl(\frac{1}{5}, \frac{11}{20}\bigr)$ & $\bigl(\frac{4}{5}, \frac{23}{20}\bigr)$
\\\hline
$(ba, a, ab)$ & $\bigl(-\frac{9}{20}, \frac{13}{20}\bigr)$ & $\bigl(\frac{29}{20}, -\frac{7}{10}\bigr)$ & $(ba, b, b)$ & $\bigl(\frac{11}{20}, -\frac{11}{10}\bigr)$ & $\bigl(\frac{9}{20}, \frac{11}{20}\bigr)$ & $(ba, b, ba)$ & $\bigl(\frac{1}{20}, -\frac{1}{10}\bigr)$ & $\bigl(\frac{19}{20}, \frac{1}{20}\bigr)$
\\
$(-13, 14)$ & $\bigl(\frac{1}{4}, -\frac{7}{10}\bigr)$ & $\bigl(\frac{3}{4}, \frac{1}{20}\bigr)$ & $(22, -11)$ & $\bigl(0, \frac{11}{20}\bigr)$ & $\bigl(1, \frac{11}{20}\bigr)$ & $(2, -1)$ & $\bigl(0, \frac{1}{20}\bigr)$ & $\bigl(1, \frac{1}{20}\bigr)$
\\
$(0, 0)$ & $\bigl(\frac{1}{5}, \frac{1}{20}\bigr)$ & $\bigl(\frac{4}{5}, \frac{13}{20}\bigr)$ & $(0, 0)$ & $\bigl(-\frac{11}{20}, \frac{11}{20}\bigr)$ & $\bigl(\frac{31}{20}, -\frac{11}{10}\bigr)$ & $(0, 0)$ & $\bigl(-\frac{1}{20}, \frac{1}{20}\bigr)$ & $\bigl(\frac{21}{20}, -\frac{1}{10}\bigr)$
\\\hline
$(ba, ab, a)$ & $\bigl(\frac{1}{20}, \frac{3}{20}\bigr)$ & $\bigl(\frac{19}{20}, \frac{3}{10}\bigr)$ & $(ba, ab, ab)$ & $\bigl(\frac{11}{20}, -\frac{7}{20}\bigr)$ & $\bigl(\frac{9}{20}, \frac{13}{10}\bigr)$ & $(ba, ba, b)$ & $\bigl(-\frac{9}{20}, \frac{9}{10}\bigr)$ & $\bigl(\frac{29}{20}, -\frac{9}{20}\bigr)$
\\
$(-3, -6)$ & $\bigl(-\frac{1}{4}, \frac{3}{10}\bigr)$ & $\bigl(\frac{5}{4}, -\frac{9}{20}\bigr)$ & $(7, -26)$ & $\bigl(-\frac{3}{4}, \frac{13}{10}\bigr)$ & $\bigl(\frac{7}{4}, -\frac{19}{20}\bigr)$ & $(-18, 9)$ & $\bigl(0, -\frac{9}{20}\bigr)$ & $\bigl(1, -\frac{9}{20}\bigr)$
\\
$(0, 0)$ & $\bigl(\frac{1}{5}, -\frac{9}{20}\bigr)$ & $\bigl(\frac{4}{5}, \frac{3}{20}\bigr)$ & $(0, 0)$ & $\bigl(\frac{1}{5}, -\frac{19}{20}\bigr)$ & $\bigl(\frac{4}{5}, -\frac{7}{20}\bigr)$ & $(0, 0)$ & $\bigl(\frac{9}{20}, -\frac{9}{20}\bigr)$ & $\bigl(\frac{11}{20}, \frac{9}{10}\bigr)$
\\\hline
$(ba, ba, ba)$ & $\bigl(-\frac{19}{20}, \frac{19}{10}\bigr)$ & $\bigl(\frac{39}{20}, -\frac{19}{20}\bigr)$ & $(ba, aba, e)$ & $\bigl(-\frac{1}{5}, \frac{13}{20}\bigr)$ & $\bigl(\frac{6}{5}, \frac{1}{20}\bigr)$ & $(ba, aba, aba)$ & $\bigl(-\frac{1}{5}, \frac{23}{20}\bigr)$ & $\bigl(\frac{6}{5}, \frac{11}{20}\bigr)$
\\
$(-38, 19)$ & $\bigl(0, -\frac{19}{20}\bigr)$ & $\bigl(1, -\frac{19}{20}\bigr)$ & $(-13, -1)$ & $\bigl(-\frac{1}{4}, \frac{1}{20}\bigr)$ & $\bigl(\frac{5}{4}, -\frac{7}{10}\bigr)$ & $(-23, -11)$ & $\bigl(-\frac{3}{4}, \frac{11}{20}\bigr)$ & $\bigl(\frac{7}{4}, -\frac{17}{10}\bigr)$
\\
$(0, 0)$ & $\bigl(\frac{19}{20}, -\frac{19}{20}\bigr)$ & $\bigl(\frac{1}{20}, \frac{19}{10}\bigr)$ & $(0, 0)$ & $\bigl(\frac{9}{20}, -\frac{7}{10}\bigr)$ & $\bigl(\frac{11}{20}, \frac{13}{20}\bigr)$ & $(0, 0)$ & $\bigl(\frac{19}{20}, -\frac{17}{10}\bigr)$ & $\bigl(\frac{1}{20}, \frac{23}{20}\bigr)$
\\\hline
$(aba, e, e)$ & $\bigl(0, -\frac{11}{20}\bigr)$ & $\bigl(1, -\frac{11}{20}\bigr)$ & $(aba, e, aba)$ & $\bigl(0, -\frac{1}{20}\bigr)$ & $\bigl(1, -\frac{1}{20}\bigr)$ & $(aba, a, a)$ & $\bigl(-\frac{3}{4}, \frac{19}{20}\bigr)$ & $\bigl(\frac{7}{4}, -\frac{13}{10}\bigr)$
\\
$\bigl(11, 11\bigr)$ & $\bigl(\frac{11}{20}, -\frac{11}{20}\bigr)$ & $\bigl(\frac{9}{20}, \frac{11}{10}\bigr)$ & $(1, 1)$ & $\bigl(\frac{1}{20}, -\frac{1}{20}\bigr)$ & $\bigl(\frac{19}{20}, \frac{1}{10}\bigr)$ & $(-19, 26)$ & $\bigl(\frac{11}{20}, -\frac{13}{10}\bigr)$ & $\bigl(\frac{9}{20}, \frac{7}{20}\bigr)$
\\
$(0, 0)$ & $\bigl(-\frac{11}{20}, \frac{11}{10}\bigr)$ & $\bigl(\frac{31}{20}, -\frac{11}{20}\bigr)$ & $(0, 0)$ & $\bigl(-\frac{1}{20}, \frac{1}{10}\bigr)$ & $\bigl(\frac{21}{20}, -\frac{1}{20}\bigr)$ & $(0, 0)$ & $\bigl(\frac{1}{5}, \frac{7}{20}\bigr)$ & $\bigl(\frac{4}{5}, \frac{19}{20}\bigr)$
\\\hline
$(aba, a, ab)$ & $\bigl(-\frac{1}{4}, \frac{9}{20}\bigr)$ & $\bigl(\frac{5}{4}, -\frac{3}{10}\bigr)$ & $(aba, b, b)$ & $\bigl(\frac{3}{4}, -\frac{13}{10}\bigr)$ & $\bigl(\frac{1}{4}, \frac{19}{20}\bigr)$ & $(aba, b, ba)$ & $\bigl(\frac{1}{4}, -\frac{3}{10}\bigr)$ & $\bigl(\frac{3}{4}, \frac{9}{20}\bigr)$
 \\
$(-9, 6)$ & $\bigl(\frac{1}{20}, -\frac{3}{10}\bigr)$ & $\bigl(\frac{19}{20}, -\frac{3}{20}\bigr)$ & $(26, -19)$ & $\bigl(-\frac{1}{5}, \frac{19}{20}\bigr)$ & $\bigl(\frac{6}{5}, \frac{7}{20}\bigr)$ & $(6, -9)$ & $\bigl(-\frac{1}{5}, \frac{9}{20}\bigr)$ & $\bigl(\frac{6}{5}, -\frac{3}{20}\bigr)$
\\
$(0, 0)$ & $\bigl(\frac{1}{5}, -\frac{3}{20}\bigr)$ & $\bigl(\frac{4}{5}, \frac{9}{20}\bigr)$ & $(0, 0)$ & $\bigl(-\frac{11}{20}, \frac{7}{20}\bigr)$ & $\bigl(\frac{31}{20}, -\frac{13}{10}\bigr)$ & $(0, 0)$ & $\bigl(-\frac{1}{20}, -\frac{3}{20}\bigr)$ & $\bigl(\frac{21}{20}, -\frac{3}{10}\bigr)$
\\\hline
$(aba, ab, a)$ & $\bigl(\frac{1}{4}, -\frac{1}{20}\bigr)$ & $\bigl(\frac{3}{4}, \frac{7}{10}\bigr)$ & $(aba, ab, ab)$ & $\bigl(\frac{3}{4}, -\frac{11}{20}\bigr)$ & $\bigl(\frac{1}{4}, \frac{17}{10}\bigr)$ & $(aba, ba, b)$ & $\bigl(-\frac{1}{4}, \frac{7}{10}\bigr)$ & $\bigl(\frac{5}{4}, -\frac{1}{20}\bigr)$
\\
$(1, -14)$ & $\bigl(-\frac{9}{20}, \frac{7}{10}\bigr)$ & $\bigl(\frac{29}{20}, -\frac{13}{20}\bigr)$ & $(11, -34)$ & $\bigl(-\frac{19}{20}, \frac{17}{10}\bigr)$ & $\bigl(\frac{39}{20}, -\frac{23}{20}\bigr)$ & $(-14, 1)$ & $\bigl(-\frac{1}{5}, -\frac{1}{20}\bigr)$ & $\bigl(\frac{6}{5}, -\frac{13}{20}\bigr)$
\\
$(0, 0)$ & $\bigl(\frac{1}{5}, -\frac{13}{20}\bigr)$ & $\bigl(\frac{4}{5}, -\frac{1}{20}\bigr)$ & $(0, 0)$ & $\bigl(\frac{1}{5}, -\frac{23}{20}\bigr)$ & $\bigl(\frac{4}{5}, -\frac{11}{20}\bigr)$ & $(0, 0)$ & $\bigl(\frac{9}{20}, -\frac{13}{20}\bigr)$ & $\bigl(\frac{11}{20}, \frac{7}{10}\bigr)$
\\\hline
$(aba, ba, ba)$ & $\bigl(-\frac{3}{4}, \frac{17}{10}\bigr)$ & $\bigl(\frac{7}{4}, -\frac{11}{20}\bigr)$ & $(aba, aba, e)$ & $\bigl(0, \frac{9}{20}\bigr)$ & $\bigl(1, \frac{9}{20}\bigr)$ & $(aba, aba, aba)$ & $\bigl(0, \frac{19}{20}\bigr)$ & $\bigl(1, \frac{19}{20}\bigr)$
\\
$(-34, 11)$ & $\bigl(-\frac{1}{5}, -\frac{11}{20}\bigr)$ & $\bigl(\frac{6}{5}, -\frac{23}{20}\bigr)$ & $(-9, -9)$ & $\bigl(-\frac{9}{20}, \frac{9}{20}\bigr)$ & $\bigl(\frac{29}{20}, -\frac{9}{10}\bigr)$ & $(-19, -19)$ & $\bigl(-\frac{19}{20}, \frac{19}{20}\bigr)$ & $\bigl(\frac{39}{20}, -\frac{19}{10}\bigr)$
\\
$(0, 0)$ & $\bigl(\frac{19}{20}, -\frac{23}{20}\bigr)$ & $\bigl(\frac{1}{20}, \frac{17}{10}\bigr)$ & $(0, 0)$ & $\bigl(\frac{9}{20}, -\frac{9}{10}\bigr)$ & $\bigl(\frac{11}{20}, \frac{9}{20}\bigr)$ & $(0, 0)$ & $\bigl(\frac{19}{20}, -\frac{19}{10}\bigr)$ & $\bigl(\frac{1}{20}, \frac{19}{20}\bigr)$ \\ \hline
\end{tabular}
 \caption{$\bal$ of $M\bigl( -1;\frac{1}{3},\frac{1}{2}, \frac{1}{4} \bigr)$ for each $\vec s$.} \label{tab:balinrangePseudo}
\end{table}

\begin{table}[th]\centering\small
\begin{tabular}{|lll|lll|lll|}\hline
$(e, e, e)$ & $\bigl(0, -\frac{5}{4}\bigr)$ & $\bigl(1, -\frac{5}{4}\bigr)$ & $(e, e, aba)$ & $\bigl(0, -\frac{3}{4}\bigr)$ & $\bigl(1, -\frac{3}{4}\bigr)$ & $\left(e, aba, e\right)$ & $\left(0, -\frac{1}{4}\right)$ & $\left(1, -\frac{1}{4}\right)$
\\
$(5, 5)$ & $\bigl(\frac{5}{4}, -\frac{5}{4}\bigr)$ & $\bigl(-\frac{1}{4}, \frac{5}{2}\bigr)$ & $(3, 3)$ & $\bigl(\frac{3}{4}, -\frac{3}{4}\bigr)$ & $\bigl(\frac{1}{4}, \frac{3}{2}\bigr)$ & $(1, 1)$ & $\bigl(\frac{1}{4}, -\frac{1}{4}\bigr)$ & $\bigl(\frac{3}{4}, \frac{1}{2}\bigr)$
\\
$(0, 0)$ & $\bigl(-\frac{5}{4}, \frac{5}{2}\bigr)$ & $\bigl(\frac{9}{4}, -\frac{5}{4}\bigr)$ & $(0, 0)$ & $\bigl(-\frac{3}{4}, \frac{3}{2}\bigr)$ & $\bigl(\frac{7}{4}, -\frac{3}{4}\bigr)$ & $(0, 0)$ & $\bigl(-\frac{1}{4}, \frac{1}{2}\bigr)$ & $\bigl(\frac{5}{4}, -\frac{1}{4}\bigr)$
\\\hline
$(e, aba, aba)$ & $\bigl(0, \frac{1}{4}\bigr)$ & $\bigl(1, \frac{1}{4}\bigr)$ & $(a, a, a)$ & $\bigl(-\frac{5}{4}, \frac{5}{4}\bigr)$ & $\bigl(\frac{9}{4}, -\frac{5}{2}\bigr)$ & $(a, a, ab)$ & $\bigl(-\frac{3}{4}, \frac{3}{4}\bigr)$ & $\bigl(\frac{7}{4}, -\frac{3}{2}\bigr)$
\\
$(-1, -1)$ & $\bigl(-\frac{1}{4}, \frac{1}{4}\bigr)$ & $\bigl(\frac{5}{4}, -\frac{1}{2}\bigr)$ & $(-5, 10)$ & $\bigl(\frac{5}{4}, -\frac{5}{2}\bigr)$ & $\bigl(-\frac{1}{4}, \frac{5}{4}\bigr)$ & $(-3, 6)$ & $\bigl(\frac{3}{4}, -\frac{3}{2}\bigr)$ & $\bigl(\frac{1}{4}, \frac{3}{4}\bigr)$
\\
$(0, 0)$ & $\bigl(\frac{1}{4}, -\frac{1}{2}\bigr)$ & $\bigl(\frac{3}{4}, \frac{1}{4}\bigr)$ & $(0, 0)$ & $\bigl(0, \frac{5}{4}\bigr)$ & $\bigl(1, \frac{5}{4}\bigr)$ & $(0, 0)$ & $\bigl(0, \frac{3}{4}\bigr)$ & $\bigl(1, \frac{3}{4}\bigr)$
\\\hline
$(a, ab, a)$ & $\bigl(-\frac{1}{4}, \frac{1}{4}\bigr)$ & $\bigl(\frac{5}{4}, -\frac{1}{2}\bigr)$ & $(a, ab, ab)$ & $\bigl(\frac{1}{4}, -\frac{1}{4}\bigr)$ & $\bigl(\frac{3}{4}, \frac{1}{2}\bigr)$ & $(b, b, b)$ & $\bigl(\frac{5}{4}, -\frac{5}{2}\bigr)$ & $\bigl(-\frac{1}{4}, \frac{5}{4}\bigr)$
\\
$(-1, 2)$ & $\bigl(\frac{1}{4}, -\frac{1}{2}\bigr)$ & $\bigl(\frac{3}{4}, \frac{1}{4}\bigr)$ & $(1, -2)$ & $\bigl(-\frac{1}{4}, \frac{1}{2}\bigr)$ & $\bigl(\frac{5}{4}, -\frac{1}{4}\bigr)$ & $(10, -5)$ & $\bigl(0, \frac{5}{4}\bigr)$ & $\bigl(1, \frac{5}{4}\bigr)$
\\
$(0, 0)$ & $\bigl(0, \frac{1}{4}\bigr)$ & $\bigl(1, \frac{1}{4}\bigr)$ & $(0, 0)$ & $\bigl(0, -\frac{1}{4}\bigr)$ & $\bigl(1, -\frac{1}{4}\bigr)$ & $(0, 0)$ & $\bigl(-\frac{5}{4}, \frac{5}{4}\bigr)$ & $\bigl(\frac{9}{4}, -\frac{5}{2}\bigr)$
\\\hline
$(b, b, ba)$ & $\bigl(\frac{3}{4}, -\frac{3}{2}\bigr)$ & $\bigl(\frac{1}{4}, \frac{3}{4}\bigr)$ & $(b, ba, b)$ & $\bigl(\frac{1}{4}, -\frac{1}{2}\bigr)$ & $\bigl(\frac{3}{4}, \frac{1}{4}\bigr)$ & $(b, ba, ba)$ & $\bigl(-\frac{1}{4}, \frac{1}{2}\bigr)$ & $\bigl(\frac{5}{4}, -\frac{1}{4}\bigr)$
\\
$(6, -3)$ & $\bigl(0, \frac{3}{4}\bigr)$ & $\bigl(1, \frac{3}{4}\bigr)$ & $(2, -1)$ & $\bigl(0, \frac{1}{4}\bigr)$ & $\bigl(1, \frac{1}{4}\bigr)$ & $(-2, 1)$ & $\bigl(0, -\frac{1}{4}\bigr)$ & $\bigl(1, -\frac{1}{4}\bigr)$
\\
$(0, 0)$ & $\bigl(-\frac{3}{4}, \frac{3}{4}\bigr)$ & $\bigl(\frac{7}{4}, -\frac{3}{2}\bigr)$ & $(0, 0)$ & $\bigl(-\frac{1}{4}, \frac{1}{4}\bigr)$ & $\bigl(\frac{5}{4}, -\frac{1}{2}\bigr)$ & $(0, 0)$ & $\bigl(\frac{1}{4}, -\frac{1}{4}\bigr)$ & $\bigl(\frac{3}{4}, \frac{1}{2}\bigr)$
\\\hline
$(ab, a, a)$ & $\bigl(-\frac{1}{4}, \frac{1}{4}\bigr)$ & $\bigl(\frac{5}{4}, -\frac{1}{2}\bigr)$ & $(ab, a, ab)$ & $\bigl(\frac{1}{4}, -\frac{1}{4}\bigr)$ & $\bigl(\frac{3}{4}, \frac{1}{2}\bigr)$ & $(ab, ab, a)$ & $\bigl(\frac{3}{4}, -\frac{3}{4}\bigr)$ & $\bigl(\frac{1}{4}, \frac{3}{2}\bigr)$
\\
$(-1, 2)$ & $\bigl(\frac{1}{4}, -\frac{1}{2}\bigr)$ & $\bigl(\frac{3}{4}, \frac{1}{4}\bigr)$ & $(1, -2)$ & $\bigl(-\frac{1}{4}, \frac{1}{2}\bigr)$ & $\bigl(\frac{5}{4}, -\frac{1}{4}\bigr)$ & $(3, -6)$ & $\bigl(-\frac{3}{4}, \frac{3}{2}\bigr)$ & $\bigl(\frac{7}{4}, -\frac{3}{4}\bigr)$
\\
$(0, 0)$ & $\bigl(0, \frac{1}{4}\bigr)$ & $\bigl(1, \frac{1}{4}\bigr)$ & $(0, 0)$ & $\bigl(0, -\frac{1}{4}\bigr)$ & $\bigl(1, -\frac{1}{4}\bigr)$ & $(0, 0)$ & $\bigl(0, -\frac{3}{4}\bigr)$ & $\bigl(1, -\frac{3}{4}\bigr)$
\\\hline
$(ab, ab, ab)$ & $\bigl(\frac{5}{4}, -\frac{5}{4}\bigr)$ & $\bigl(-\frac{1}{4}, \frac{5}{2}\bigr)$ & $(ba, b, b)$ & $\bigl(\frac{1}{4}, -\frac{1}{2}\bigr)$ & $\bigl(\frac{3}{4}, \frac{1}{4}\bigr)$ & $(ba, b, ba)$ & $\bigl(-\frac{1}{4}, \frac{1}{2}\bigr)$ & $\bigl(\frac{5}{4}, -\frac{1}{4}\bigr)$
\\
$(5, -10)$ & $\bigl(-\frac{5}{4}, \frac{5}{2}\bigr)$ & $\bigl(\frac{9}{4}, -\frac{5}{4}\bigr)$ & $(2, -1)$ & $\bigl(0, \frac{1}{4}\bigr)$ & $\bigl(1, \frac{1}{4}\bigr)$ & $(-2, 1)$ & $\bigl(0, -\frac{1}{4}\bigr)$ & $\bigl(1, -\frac{1}{4}\bigr)$
\\
$(0, 0)$ & $\bigl(0, -\frac{5}{4}\bigr)$ & $\bigl(1, -\frac{5}{4}\bigr)$ & $(0, 0)$ & $\bigl(-\frac{1}{4}, \frac{1}{4}\bigr)$ & $\bigl(\frac{5}{4}, -\frac{1}{2}\bigr)$ & $(0, 0)$ & $\bigl(\frac{1}{4}, -\frac{1}{4}\bigr)$ & $\bigl(\frac{3}{4}, \frac{1}{2}\bigr)$
\\\hline
$(ba, ba, b)$ & $\bigl(-\frac{3}{4}, \frac{3}{2}\bigr)$ & $\bigl(\frac{7}{4}, -\frac{3}{4}\bigr)$ & $(ba, ba, ba)$ & $\bigl(-\frac{5}{4}, \frac{5}{2}\bigr)$ & $\bigl(\frac{9}{4}, -\frac{5}{4}\bigr)$ & $(aba, e, e)$ & $\bigl(0, -\frac{1}{4}\bigr)$ & $\bigl(1, -\frac{1}{4}\bigr)$
\\
$(-6, 3)$ & $\bigl(0, -\frac{3}{4}\bigr)$ & $\bigl(1, -\frac{3}{4}\bigr)$ & $(-10, 5)$ & $\bigl(0, -\frac{5}{4}\bigr)$ & $\bigl(1, -\frac{5}{4}\bigr)$ & $(1, 1)$ & $\bigl(\frac{1}{4}, -\frac{1}{4}\bigr)$ & $\bigl(\frac{3}{4}, \frac{1}{2}\bigr)$
\\
$(0, 0)$ & $\bigl(\frac{3}{4}, -\frac{3}{4}\bigr)$ & $\bigl(\frac{1}{4}, \frac{3}{2}\bigr)$ & $(0, 0)$ & $\bigl(\frac{5}{4}, -\frac{5}{4}\bigr)$ & $\bigl(-\frac{1}{4}, \frac{5}{2}\bigr)$ & $(0, 0)$ & $\bigl(-\frac{1}{4}, \frac{1}{2}\bigr)$ & $\bigl(\frac{5}{4}, -\frac{1}{4}\bigr)$
\\\hline
$(aba, e, aba)$ & $\bigl(0, \frac{1}{4}\bigr)$ & $\bigl(1, \frac{1}{4}\bigr)$ & $(aba, aba, e)$ & $\bigl(0, \frac{3}{4}\bigr)$ & $\bigl(1, \frac{3}{4}\bigr)$ & $(aba, aba, aba)$ & $\bigl(0, \frac{5}{4}\bigr)$ & $\bigl(1, \frac{5}{4}\bigr)$
\\
$(-1, -1)$ & $\bigl(-\frac{1}{4}, \frac{1}{4}\bigr)$ & $\bigl(\frac{5}{4}, -\frac{1}{2}\bigr)$ & $(-3, -3)$ & $\bigl(-\frac{3}{4}, \frac{3}{4}\bigr)$ & $\bigl(\frac{7}{4}, -\frac{3}{2}\bigr)$ & $(-5, -5)$ & $\bigl(-\frac{5}{4}, \frac{5}{4}\bigr)$ & $\bigl(\frac{9}{4}, -\frac{5}{2}\bigr)$
 \\
$(0, 0)$ & $\bigl(\frac{1}{4}, -\frac{1}{2}\bigr)$ & $\bigl(\frac{3}{4}, \frac{1}{4}\bigr)$ & $(0, 0)$ & $\bigl(\frac{3}{4}, -\frac{3}{2}\bigr)$ & $\bigl(\frac{1}{4}, \frac{3}{4}\bigr)$ & $(0, 0)$ & $\bigl(\frac{5}{4}, -\frac{5}{2}\bigr)$ & $\bigl(-\frac{1}{4}, \frac{5}{4}\bigr)$
\\\hline
\end{tabular}
\caption{$\bal$ of $M\bigl( -2;\frac{1}{2},\frac{1}{2},\frac{3}{4} \bigr)$ for each $\vec s$.}\label{tab:balinrangeGeneral}
\end{table}

\end{landscape}

\subsection*{Acknowledgments} We would like to thank Boudewijn Bosch,
Francesca Ferrari and Sergei Gukov for fruitful discussions.
The authors would also like to thank the referees for their suggestions.
The work of M.C.\ is supported by ERC starting grant H2020 \# 640159 and
NWO vidi grant (number 016.Vidi.189.182), and the Ministry of Science and Technology of
Taiwan (110-2115-M-001-018-MY3). The work of I.C.\ is partly supported by the
ERC starting grant H2020 \# 640159 and
NWO vidi grant (number 016.Vidi.189.182). The work of D.P.\ is supported by the NWO vidi grant (number 016.Vidi.189.182). I.C.\ and D.P.\ would like to thank Academia Sinica for hospitality during the final stage of the project.

\pdfbookmark[1]{References}{ref}
\LastPageEnding

\end{document}